\pdfoutput=1

\newcommand\runtit{{Composable and Finite Computational Security of
    Quantum Message Trans.}}
\newcommand\runaut{{F. Banfi, U. Maurer, C. Portmann, and J. Zhu}}
\newcommand\tit{{Composable and Finite Computational Security of
    Quantum Message Transmission}}
\newcommand\aut{{Fabio Banfi, Ueli Maurer, Christopher Portmann, and Jiamin Zhu}}
\newcommand\instit{Department of Computer Science\\ETH Zurich\\8092 Zurich, Switzerland}
\newcommand\emails{\{%
    \href{mailto:fbanfi@inf.ethz.ch}{\nolinkurl{fbanfi}},%
    \href{mailto:maurer@inf.ethz.ch}{\nolinkurl{maurer}},%
    \href{mailto:chportma@inf.ethz.ch}{\nolinkurl{chportma}},%
    \href{mailto:zhujia@inf.ethz.ch}{\nolinkurl{zhujia}}%
\}@inf.ethz.ch}

\documentclass[runningheads,11pt]{llncs}

\newif\iffinal
\finaltrue % UNCOMMENT FOR FINAL VERSION

\newif\ifsub
\DeclareMathSymbol\Gamma\mathalpha{operators}{"00}
\DeclareMathSymbol\Delta\mathalpha{operators}{"01}
\DeclareMathSymbol\Theta\mathalpha{operators}{"02}
\DeclareMathSymbol\Lambda\mathalpha{operators}{"03}
\DeclareMathSymbol\Xi\mathalpha{operators}{"04}
\DeclareMathSymbol\Pi\mathalpha{operators}{"05}
\DeclareMathSymbol\Sigma\mathalpha{operators}{"06}
\DeclareMathSymbol\Upsilon\mathalpha{operators}{"07}
\DeclareMathSymbol\Phi\mathalpha{operators}{"08}
\DeclareMathSymbol\Psi\mathalpha{operators}{"09}
\DeclareMathSymbol\Omega\mathalpha{operators}{"0A}

\usepackage{amsmath}
\usepackage{amssymb}
\usepackage{amsthm}
\usepackage[english]{babel}
\usepackage{mathtools}  %added by Christopher
\usepackage[multiple]{footmisc}
\usepackage{hyperref}
\usepackage{xcolor}
\usepackage{braket}
\usepackage{dsfont}
\usepackage{tcolorbox}
\usepackage{algorithm}
\usepackage[noend]{algpseudocode}
\usepackage[margin=10pt,font=small,labelfont=bf,labelsep=period,hypcap=true]{caption}
\usepackage{bm}
\usepackage[noadjust]{cite}
\usepackage{dashrule}

\usepackage[open,numbered]{bookmark} % added by Christopher, additional pdfbookmark commands, in particular \bookmarksetup{startatroot}[open,openlevel=0,numbered]
\usepackage[shortcuts]{extdash} % added by Christopher, for hyphenation of words after dashes 
\usepackage{underscore} % added by Christopher, solves issues with underscrores in DOIs
\usepackage{eprint} % added by Christopher, downloaded from theoryofcomputing.org. automatically makes hyperlinks out of doi and arXiv numbers.
\usepackage[mathscr]{euscript} % added by Christopher, usage: \mathscr{ABC}
\usepackage{tikz} %added by Christopher, inline graphics.
\usetikzlibrary{arrows.meta,fit,positioning,shapes}
\usepackage{enumitem} % more itemize commands, in particular for
\setlist{nolistsep}

\usepackage[margin=4cm]{geometry}

\hypersetup
{
    colorlinks=true,
    linkcolor={blue!70!black},
    citecolor={green!80!black},
    urlcolor={blue!70!black},
    filecolor={blue!70!black},
    pdfinfo=
    {
        Title  = \tit,
        Author = \aut
    }
}

\addto\extrasenglish{
}

\tcbuselibrary{most}

\let\originalleft\left
\let\originalright\right
\renewcommand{\left}{\mathopen{}\mathclose\bgroup\originalleft}
\renewcommand{\right}{\aftergroup\egroup\originalright}

\algrenewcommand\alglinenumber[1]{\color{darkgray}\footnotesize#1:} % Algo line #

\algblockdefx[NAME]{Orac}{EndOrac}[2][Unknown]{\textbf{oracle #1}(#2):}{\textbf{end oracle}}
\algblockdefx[NAME]{Proc}{EndProc}[2][Unknown]{\textbf{proc #1}(#2):}{\textbf{end proc}}
\algblockdefx[NAME]{Iface}{EndIface}[2][Unknown]{\textbf{iface} {#1}(#2):}{\textbf{end iface}}
\algnotext{EndIf}
\algnotext{EndFor}
\algnotext{EndOrac}
\algnotext{EndProc}
\algnotext{EndIface}

\newcommand{\algstrut}[1][\algruledefaultfactor]{\vrule width 0pt depth .25\baselineskip height #1\baselineskip\relax}
\newcommand*{\algrule}[1][\algorithmicindent]{\hspace*{.5em}\vrule\algstrut\hspace*{\dimexpr#1-.5em}}
\makeatletter
\newcount\ALG@printindent@tempcnta
\def\ALG@printindent{%
    \ifnum \theALG@nested>0% is there anything to print
    \ifx\ALG@text\ALG@x@notext% is this an end group without any text?
    \else
    \unskip
    \ALG@printindent@tempcnta=1
    \loop
    \algrule[\csname ALG@ind@\the\ALG@printindent@tempcnta\endcsname]%
    \advance \ALG@printindent@tempcnta 1
    \ifnum \ALG@printindent@tempcnta<\numexpr\theALG@nested+1\relax% can't do <=, so add one to RHS and use < instead
    \repeat
    \fi
    \fi
}%
\tcbset{
    base/.style={
        colback=white,
        colback=black!2,
        colframe=black,
        colbacktitle=gray!25,
        coltitle=black,
        size=small,
        boxrule=0.3mm, % exactly as in size=small
        top=1mm,
        arc=4pt
    },
    sys/.style={
        top=3mm,
        enhanced,
        attach boxed title to top left={xshift=4mm, yshift=-2.5mm},
        boxed title style={boxrule=0.3pt, arc=3.5pt, colframe=black}
    }
}
\newtcolorbox{algobox}[3][]{
    width=#2,
    title=#3,
    base,
    #1
}
\newtcolorbox{schbox}[3][]{
    width=#2,
    title=#3,
    base,
    sharp corners,
    #1
}
\newtcolorbox{convbox}[3][]{
    width=#2,
    title=#3,
    base,
    sys,
    #1
}
\newtcolorbox{resbox}[3][]{
    width=#2,
    title=#3,
    base,
    sys,
    sharp corners,
    #1
}

\newcommand\N{\mathbb N}

\newcommand\R{\mathbb R}
\renewcommand{\eqref}[1]{\hyperref[#1]{(\ref*{#1})}}
\newcommand{\eqnref}[1]{\hyperref[#1]{Eq.~(\ref*{#1})}}
\newcommand{\eqnsref}[1]{Eqs.~\eqref{#1}}
\newcommand\aref[1]{\hyperref[#1]{Appendix~\ref*{#1}}}
\newcommand{\pref}[1]{\hyperref[#1]{page~\pageref*{#1}}} % page #

\newcommand\Hi{\mathcal H}  % Hilbert space
\newcommand{\lo}[1]{\mathcal{L}(\Hi_{#1})} % linear operator
\newcommand\epr{\ket{\Phi^+}} % EPR pair
\renewcommand\epr{\ket{\phi_0}} % EPR pair
\newcommand{\ketbra}[2]{{\lvert #1\rangle\!\langle #2\rvert}}

\newcommand{\proj}[1]{\ketbra{#1}{#1}}

\newcommand\x\times
\newcommand\df\coloneqq %\doteq  % define, e.g., :=, \coloneqq from mathtools package
\newcommand\e\varepsilon
\newcommand{\eps}{\varepsilon}
\newcommand\es\varnothing
\newcommand\then\longrightarrow
\newcommand\Then\Longrightarrow
\newcommand\pr[2]{\mathrm{Pr}^{#1}\left[#2\right]} % probability with experiment, removed negative space because redefining \left and \right solves this.
\newcommand\var\mathsf % for some variable names
\newcommand\vNextIndex{\var{nextIndex}}

\newcommand\vSeed{\var{seed}}
\newcommand\lab\mathsf % labels for inputs to resources
\newcommand\lAcc{\lab{acc}}
\newcommand\lRej{\lab{rej}}
\newcommand\lNewMsg{\lab{newMsg}}
\newcommand\lSend{\lab{send}}
\newcommand\lSkip{\lab{skip}}
\newcommand\lGetKey{\lab{getKey}}
\newcommand\lSeed{\lab{seed}}
\newcommand\lInput{\lab{input}}

\newcommand\lCheat{\lab{cheat}}
\newcommand\One{\mathds1} % identity matrix
\newcommand\mc\mathcal
\newcommand\oset[3][0ex]{\mathrel{\mathop{#3}\limits^{\vbox to#1{\kern-3\ex@\hbox{$\scriptstyle#2$}\vss}}}}
\newcommand\uar{\oset\${\vphantom{a}\smash\gets}} % uniform sampling
\newcommand\samp{\oset{\hspace{1mm}p_k}\longleftarrow} % sampling from p
\newcommand\lp{\left(}
\newcommand\rp{\right)}
\newcommand\bin{\{0,1\}}
\newcommand\ite[3]{#1\,\texttt?\,#2\texttt:#3}

\newcommand\Tr{\mathrm{Tr}}
\DeclareMathOperator\spl{spl}

\newcommand{\qauth}{\mathrm{q\text-auth}}
\newcommand{\qconf}{\mathrm{q\text-otp}}

\newcommand\key{\mathsf{KEY}}
\newcommand\prg{\mathsf{PRG}}
\newcommand\prf{\mathsf{PRF}}
\newcommand\urf{\mathsf{URF}}

\newcommand\qc{\mathsf{QC}} % quantum computer
\newcommand\iqc{\mathsf{IC}} % insecure channel
\newcommand\xqc{\mathsf{PMCC}} % confidential channel
\newcommand\nmqc{\mathsf{NMCC}} % non-malleable channel
\newcommand\sqc{\mathsf{SC}} % secure channel
\newcommand\osqc{\mathsf{OSC}} % ordered secure channel
\newcommand\iface{} %\mathsf % interfaces
\newcommand\iA{\iface A}
\newcommand\iB{\iface B}
\newcommand\iE{\iface E}

\newcommand\state\rho
\newcommand\reg[1]{^{#1}} % registers
\newcommand\rA{A}
\newcommand\rB{B}
\newcommand\rX{X}
\newcommand\rE{E}

\newcommand\rM{M}
\newcommand\rC{C}
\newcommand\cptp\mathcal % font for cptp maps

\makeatletter
\newcommand\ifmember[2]{% pseudo-code stuff
    \in@{#1}{#2}%
    \ifin@
    \expandafter\@firstoftwo
    \else
    \expandafter\@secondoftwo
    \fi
}
\makeatother
\newcommand\sch{\mathfrak S} % \sch = (\Gen,\Enc,\Dec)
\newcommand\sqes{SQES} % acronym for symmetric quantum
\newcommand\Gen{\mathtt{Gen}}
\newcommand\Enc{\mathtt{Enc}}
\newcommand\Dec{\mathtt{Dec}}
\newcommand\Key{\mc K}  % space for key
\newcommand\Rnd{\mc R} % space for randomness
\newcommand\Bot{\proj{\bot}} % \bot projector
\newcommand\msg{\varrho}
\newcommand\ctx{\sigma}
\newcommand\aux{\omega}

\newcommand\Out\Rightarrow
\newcommand\IND{\mathsf{IND}\text-}
\newcommand\AGM{\mathsf{AGM}\text-}
\newcommand\RRC{\mathsf{RRC}\text-}
\newcommand\RRO{\mathsf{RRO}\text-}

\newcommand\RCCA{\mathsf{RCCA}}
\newcommand\CCA[1]{\mathsf{CCA#1}}
\newcommand\QCPA{\mathsf{QCPA}}
\newcommand\QCCA[1]{\mathsf{QCCA#1}}
\newcommand\QAE{\mathsf{QAE}}
\newcommand\QCNF{\mathsf{CC\text-QCNF}}
\newcommand\QSEC{\mathsf{CC\text-QSEC}}
\newcommand\agm{\mathrm{agm}\text-}
\newcommand\rrc{\mathrm{rrc}\text-}
\newcommand\rro{\mathrm{rro}\text-}
\newcommand\qcca[1]{\mathrm{qcca#1}}
\newcommand\qccat[1]{\mathrm{qcca#1\text-test}}
\newcommand\qccaf[1]{\mathrm{qcca#1\text-fake}}
\newcommand\qccar[1]{\mathrm{qcca#1\text-real}}
\newcommand\qccai[1]{\mathrm{qcca#1\text-ideal}}
\newcommand\qae{\mathrm{qae}}
\newcommand\qaer{\mathrm{qae\text-real}}
\newcommand\qaei{\mathrm{qae\text-ideal}}
\newcommand\qenc{\mathrm{q\text-enc}}

\newcommand\Adv[3]{\mathbf{Adv}^{#1}_{#2,#3}}

\newcommand\sys\mathsf % font for systems
\newcommand\enc{\mathsf{enc}} % some converters
\newcommand\dec{\mathsf{dec}}
\newcommand\simul{\mathsf{sim}}
\newcommand\cmt[1]{\Comment\textit{\color{gray}#1}}

\newcommand\Each{{\bf each }}

\newcommand{\cB}{\mathcal{B}}
\newcommand{\cC}{\mathcal{C}}
\newcommand{\cD}{\mathcal{D}}
\newcommand{\cE}{\mathcal{E}}

\newcommand{\cH}{\mathcal{H}}

\newcommand{\cL}{\mathcal{L}}
\newcommand{\cM}{\mathcal{M}}

\newcommand{\bD}{\mathds{D}}

\newcommand{\bR}{\mathds{R}}
\newcommand{\bS}{\mathds{S}}

\theoremstyle{remark}
\newtheorem*{remark*}{Remark}
\begin{document} %%%%%%%%%%%%%%%%%%%%%%%%%%%%%%%%%%%%%%%%%%%%%%%%%%%%%%%%%%%%%%
    
    \mainmatter
    
    \title\tit
    \titlerunning\runtit
    \toctitle\tit
    \ifsub
    %\author{Anonymous Authors}
	\author{}    
    \authorrunning{Anonymous Authors}
	\institute{}    
    \else
    \author\aut
    \authorrunning\runaut
    \institute{\instit\\\email{\emails}}
    \tocauthor\aut
    \fi
    \maketitle

\begin{abstract}
Recent research in quantum cryptography has led to the development
of schemes that encrypt and authenticate quantum messages with
computational security. The security definitions used so far in the
literature are asymptotic, game\-/based, and not known to be
composable. We show how to define finite, composable, computational
security for secure quantum message transmission. The new
definitions do not involve any games or oracles, they are directly
operational: a scheme is secure if it transforms an insecure channel
and a shared key into an ideal secure channel from Alice to Bob,
i.e., one which only allows Eve to block messages and learn their
size, but not change them or read them. By modifying the ideal
channel to provide Eve with more or less capabilities, one gets an
array of different security notions. By design these transformations
are composable, resulting in composable security.

Crucially, the new definitions are \emph{finite}. Security does not
rely on the asymptotic hardness of a computational problem. Instead,
one proves a finite reduction: if an adversary can distinguish the
constructed (real) channel from the ideal one (for some fixed security
parameters), then she can solve a finite instance of some
computational problem. Such a finite statement is needed to make
security claims about concrete implementations.

We then prove that (slightly modified versions of) protocols proposed
in the literature satisfy these composable definitions. And finally,
we study the relations between some game\-/based definitions and our
composable ones. In particular, we look at notions of quantum
authenticated encryption and $\QCCA2$, and show that they suffer from
the same issues as their classical counterparts: they exclude certain
protocols which are arguably secure.
\end{abstract}

    \bibliographystyle{eprintalpha}
    
    \iffinal
    \section{Introduction}
\label{sec:intro}

At its core, a security definition is a set of mathematical
conditions, and a security proof consists in showing that these
conditions hold for a given protocol. Given various security
definitions, one may analyze which are stronger and weaker by proving
reductions or finding separating examples. This however does not tell
us which definitions one should use, since too weak definitions may
have security issues and too strong definitions may exclude protocols
that are arguably secure. For example, $\IND\CCA2$ is often considered
an unnecessarily strong security definition, since taking a scheme
which is $\IND\CCA2$ and appending a bit to the ciphertext results in
a new encryption scheme that is arguably as secure as the original
scheme, but does not satisfy $\IND\CCA2$
\cite{canetti2003relaxing,CMT13}.

In this work we take a more critical approach to defining security. We
ask what criteria a security definition needs to satisfy that are both
necessary and sufficient conditions to call a protocol ``secure''. We
then apply them to the problem of encrypting and authenticating
quantum messages with computational security in the symmetric\-/key
setting.

\subsection{A Security Desideratum}
\label{sec:intro.desideratum}

\paragraph{Operational security.} Common security definitions for
encryption and authentication found in the literature are
\emph{game\-/based}, i.e., they require that an adversary cannot win a
game such as guessing what message has been encrypted given access to
certain oracles, see, e.g., \cite{BDPR98} and \cite{KY06} for
comparisons of various such games in the public\-/key and
private\-/key settings, respectively. These have been adapted for
transmitting quantum messages: a definition for $\QCPA$ has been
proposed in \cite{BJ15}, $\QCCA1$ in \cite{alagic2016computational},
and $\QCCA2$ as well as notions of quantum unforgeability and quantum
authenticated encryption in \cite{alagic2018unforgeable}. These are
just some of the security games one can imagine \--- in the classical,
symmetric\-/key setting, \cite{KY06} analyzes 18 different security
notions. A natural question is then to ask which of these games are
the relevant ones, for which ones is it both necessary and sufficient
that an adversary cannot win them. And the general answer is: we do
not know.

Through such cryptographic protocols one wishes to prevent
an adversary from learning some part of a message or modifying a
message undetected. But it is generally unclear how such game\-/based
security definitions relate to these operational notions \--- we refer
to \cite{MRT12} for a more in-depth critique of game\-/based security.
Instead, one should directly define security
\emph{operationally}.\footnote{Note that once a game\-/based
  definition has been proven to capture operational notions such as
  confidentiality or authenticity (e.g., via a reduction), then the
  game\-/based criterion may become a benchmark for designing schemes
  with the desired security; see the discussion in
  \autoref{sec:intro.alternative}.} In this work we follow the
constructive paradigm of
\cite{maurer2011abstract,maurer2012constructive,maurer2016specifications},
and define a protocol to be secure if it constructs a channel with the
desired properties, e.g., only leaks the message size or only allows
the adversary to block the message, but not change it or insert new
messages.

\paragraph{Composable security.} A second drawback of the definitions
proposed so far in the literature for computational security of
quantum message transmission
\cite{BJ15,alagic2016computational,alagic2018unforgeable} is that they
are not (proven to be) \emph{composable}. A long history of work on
composable security has shown that analyzing a protocol in an isolated
setting does not imply that it is actually secure when one considers
the environment in which it is used. When performing such a composable
security analysis, one sometimes finds that the definitions used are
inappropriate but the protocols are actually secure like for quantum
key distribution~\cite{Ren05,BHLMO05,KRBM07}, that the definitions are
still secure (up to a loss of security parameter) like for delegated
quantum computation~\cite{DFPR14}, or that not only the definitions
but also the protocols are insecure like in relativistic and bounded
storage bit commitment and (biased) coin
tossing~\cite{vilasini2017relativisitic}.\footnote{Note that a
  negative result in a composable framework only proves that a
  protocol does not construct the desired ideal functionality. This
  does not exclude that the protocol may construct some other ideal
  functionality or may be secure given some additional set-up
  assumptions.} It is thus necessary for a protocol to be proven to
satisfy a composable security definition before it may be considered
(provably) secure and safely used in an arbitrary environment.

\paragraph{Finite security.} A third problem with the aforementioned
security definitions is that they are all \emph{asymptotic}. This
means that the protocols have a security parameter $k \in \N$ \---
formally, one considers a sequence of protocols $\{\Pi_k\}_{k \in \N}$
\--- and security is defined in the limit when $k \to \infty$. An
implementation of a protocol will however always be finite, e.g., the
honest players choose a specific parameter $k_0$ which they consider
to be sufficient and run $\Pi_{k_0}$. A security proof for
$k \to \infty$ does not tell us anything about security for any
specific parameter $k_0$ and thus does not tell us anything about the
security of $\Pi_{k_0}$, which is run by the honest players. To
resolve this issue, some works consider what is called \emph{concrete
  security} \cite{bellare1997concrete}, i.e., instead of hiding
parameters in $O$\=/notation, security bounds and reductions are given
explicitly. This is a first step at obtaining finite security, but it
still considers the security of a sequence $\{\Pi_k\}_{k \in \N}$
instead of security of the individual elements $\Pi_{k_0}$ in this
sequence. For example, one still considers adversaries that are
polynomial in $k$, simulators that must be efficient in $k$, and
errors that are negligible in $k$. But the security definition of some
$\Pi_{k_0}$ should not depend on any other elements in the sequence,
on how the sequence is defined or whether it is defined at all. Hence
notions such as poly-time, efficiency, or negligibility should not be
part of a security definition for some specific $\Pi_{k_0}$. We call
the security paradigm that analyzes individual elements $\Pi_{k_0}$
\emph{finite security}, and show in this work how to define it for
computational security of quantum message transmission.

\subsection{Overview of Results}

Our contributions are threefold. We first provide definitions for
encryption and authentication of quantum messages that satisfy the
desideratum expressed above. In particular, we show how to define
finite security in the computational case. In
\autoref{sec:intro.finite} below we explain the intuition behind this
security paradigm.

We then show that (slightly modified) protocols from the literature
\cite{alagic2016computational,alagic2018unforgeable} satisfy these
definitions. These protocols use the quantum one-time pad and quantum
information\-/theoretic authentication as subroutine
\cite{barnum2002qauth,portmann2017qauth}, but run them with keys that
are only computationally secure to encrypt multiple messages. We
explain the constructions and what is achieved in more detail in
\autoref{sec:intro.constructions}.

Now that we have security definitions that satisfy our desideratum, we
revisit some game\-/based definitions from the literature, and compare
them to our own notions of security. An overview of these results is
given in \autoref{sec:intro.games}.

\subsection{Finite Computational Security}
\label{sec:intro.finite}

In traditional asymptotic security, a cryptographic protocol is
parameterized by a single value $k \in \N$ \--- any other parameters
must be expressed as a function of $k$ \--- and one studies a sequence
of objects $\{\Pi_k\}_{k \in \N}$. In composable security, one uses
this to define a parameterized real world
$\bR = \{\sys R_k\}_{k \in \N}$ and ideal world
$\bS = \{\sys S_k\}_{k \in \N}$, and argues that no polynomial
distinguisher $\bD = \{\sys D_k\}_{k \in \N}$ can distinguish one from
the other with non\-/negligible advantage. At first glance the notions
of polynomial distinguishers and negligible functions might seem
essential, because an unbounded distinguisher can obviously
distinguish the two, and without a notion of negligibility, how can
one define what is a satisfactory bound on the distinguishability.

The latter problem is the simpler to address: instead of categorizing
distinguishability as black or white (negligible or not), we give
explicit bounds. The former issue is resolved by observing that we
never actually prove that the real and ideal world are
indistinguishable (except in the case of information\-/theoretic
security), since in most cases that would amount to solving a problem
such as $\mathsf{P} \neq \mathsf{NP}$. What one actually proves is a
\emph{reduction}, which is a finite statement, not an asymptotic
one. More precisely, one proves that if $\sys D_k$ can distinguish
$\sys R_k$ from $\sys S_k$ with advantage $p_k$, then some (explicit)
$\sys D'_k$ can solve some problem $W_k$ with probability $p'_k$ \---
if one believes that $W_k$ is asymptotically hard to solve, then this
implies that $\bD$ cannot distinguish $\bR$ from $\bS$.

A finite security statement stops after the reduction. We prove that
for any $k_0$ and any $\sys D_{k_0}$,
\begin{equation} \label{eq:finitesecurity}
d^{\sys D_{k_0}}(\sys R_{k_0}, \sys S_{k_0}) \leq f(D_{k_0}),
\end{equation}
where $d^{\sys D_{k_0}}(\cdot,\cdot)$ denotes the advantage $\sys
D_{k_0}$ has in distinguishing two given systems, and $f(\cdot)$ is
some arbitrary function, e.g., the probability that $\sys D'_{k_0}$ (which
is itself some function of $D_{k_0}$) can solve some problem $W_{k_0}$.

\autoref{eq:finitesecurity} does not require systems to be part of a
sequence with a single security parameter $k \in \N$. There may be no
security parameter at all, or multiple
parameters. Information\-/theoretic security corresponds to the
special case where one can prove that $f(\sys D_{k_0})$ is small for
all $\sys D_{k_0}$.

\subsection{Constructing Quantum Channels}
\label{sec:intro.constructions}

As mentioned in \autoref{sec:intro.desideratum}, we use the Abstract
and Constructive Cryptography (AC) framework of Maurer and Renner
\cite{maurer2011abstract,maurer2012constructive,maurer2016specifications}
in this work. To define the security of a message transmission
protocol, we need to first define the type of channel we wish to
achieve \--- for simplicity, we always consider channels going from
Alice to Bob.

The strongest channel we construct in this work is an ordered secure
quantum channel, $\osqc$, which allows Eve to decide which messages
that Alice sent will be delivered to Bob and which ones get discarded.
But it does not reveal any information about the messages (except
their size and number) to Eve and guarantees that the delivered
messages arrive in the same order in which they were sent. A somewhat
weaker channel, a secure channel $\sqc$, also allows Eve to block or
deliver each message, but additionally allows her to jumble their
order of arrival at Bob's.

Our first result shows that a modified version of a protocol from
\cite{alagic2018unforgeable} constructs the strongest channel,
$\osqc$, from an insecure channel and a short key that is used to
select a function from a pseudo\-/random family (PRF). Security holds
for any distinguisher that cannot distinguish the output of the PRF
from the output of a uniform function. We also show how one can
construct $\osqc$ from $\sqc$ by simply appending a counter to the
messages.

The two channels described above are labeled ``secure'', because they
are both confidential (Eve does not learn anything about the messages)
and authentic (Eve cannot change or insert any messages). If we are
willing to sacrifice authenticity, we can define weaker channels that
allow Eve to modify or insert messages in specific ways. We define a
non\-/malleable confidential channel, $\nmqc$ \--- which does not
allow Eve to change a message sent by Alice, but does allow her to
insert a message of her choice \--- and a Pauli\-/malleable channel,
$\xqc$ \--- which allows Eve to  apply
bit and phase flips to Alice's messages or insert a fully mixed state.

Our second construction modifies a protocol from
\cite{alagic2016computational} to construct $\xqc$ from an insecure
channel and a short key that is used to select a function from a
pseudo\-/random family (PRF). Here too, security holds for any
distinguisher that cannot distinguish the PRF from uniform.

\subsection{Comparison to Game-Based Definitions}
\label{sec:intro.games}

In the last part of this work, we relate existing game\-/based security
definitions for quantum encryption with our new proposed security
definitions phrased in constructive cryptography. More concretely, we
focus on the notions of \emph{quantum ciphertext indistinguishability
  under adaptive chosen-ciphertext attack} ($\QCCA2$) and
\emph{quantum authenticated encryption} ($\QAE$), both introduced in
\cite{alagic2018unforgeable}.

We first note that encryption schemes are defined to be stateless in
\cite{BJ15,alagic2016computational,alagic2018unforgeable} and the
proposed game\-/based definitions are tailored to such schemes. The
restricted class of encryption protocols analyzed can thus not
construct ordered channels, because the players need to remember tags
numbering the messages to be able to preserve this ordering. The
strongest notion of encryption from these works, namely $\QAE$, is
thus closest to constructing a $\sqc$. In fact, we show that $\QAE$ is
\emph{strictly} stronger than constructing a $\sqc$: a scheme
satisfying $\QAE$ constructs a $\sqc$, however there are (stateless)
schemes constructing a $\sqc$ that would be considered insecure by the
$\QAE$ game. These schemes are obtained in the same way as the ones
showing that classical $\IND\CCA2$ is unnecessarily strong: one starts
with a scheme satisfying $\QAE$ and appends a bit to the ciphertext,
resulting in a new scheme that still constructs a $\sqc$, but is not
$\QAE$-secure. Our proof shows that $\QAE$ may be seen as constructing
a $\sqc$ with a \emph{fixed} simulator that is hard\-/coded in the
game. A composable security definition only requires the existence of
a simulator, and the separation between the two notions is obtained by
considering schemes that can be proven secure using a different
simulator than the one hard\-/coded in the game.

For $\QCCA2$, we first propose an alternative game\-/based security
notion that captures the same intuition, but which we consider more
natural than the one suggested in \cite{alagic2018unforgeable}. In
particular, its classical analogue is easily shown to be equivalent to
a standard $\IND\CCA2$ notion, whereas the notion put forth in
\cite{alagic2018unforgeable}, when cast to a classical definition,
incurs a concrete constant factor loss when compared to $\IND\CCA2$,
and requires a complicated proof of this fact. We then show that for a
restricted class of protocols (which includes all the ones for which a
security proof is given in previous work), our new game-based notion
indeed implies that the protocol constructs a $\nmqc$. The same
separation holds here as well: $\QCCA2$ definitions are unnecessarily
strong, and exclude protocols that naturally construct a $\nmqc$. Note
that in the classical case, the $\IND\RCCA$
game~\cite{canetti2003relaxing} that was developed to avoid the
problems of $\IND\CCA2$ has been shown to be exactly equivalent to
constructing a classical non\-/malleable confidential channel in the
case of large message spaces~\cite{CMT13}.

\subsection{Alternative Security Notions}
\label{sec:intro.alternative}

Common security definitions often capture properties of (encryption)
schemes, e.g., let $M$ be a plaintext random variable, let $C$ be the
corresponding ciphertext, $H$ is the entropy function, $M'$ is the
received plaintext, and $\mathsf{accept}$ is the event that the
message is accepted by the receiver, then
\begin{equation} \label{eq:conf.auth} H(M|C) = H(M) \qquad \text{and}
  \qquad \Pr \left[M \neq M' \text{ and } \mathsf{accept}\right] \leq
  \eps\end{equation} are simple notions of confidentiality and
authenticity, respectively. But depending on how schemes satisfying
these equations are used \--- e.g., encrypt\-/then\-/authenticate or
authenticate\-/then\-/encrypt \--- one gets drastically different
results.\footnote{Encrypt\-/then\-/authenticate is always secure, but
  one can find examples of schemes satisfying \eqref{eq:conf.auth}
  following the authenticate\-/then\-/encrypt paradigm that are
  insecure~\cite{BN00,Kra01,MT10}.} The equations in
\eqref{eq:conf.auth} may be regarded as crucial security properties of
encryption schemes, but before schemes satisfying these may be safely
used, one needs to consider the context and prove what is actually
achieved by such constructs (in an operational sense).

The same applies to security definitions proposed for quantum key
distribution. The accessible
information\footnote{$I_{\text{acc}}(K;E) \coloneqq \max_\Gamma
  I(K;\Gamma(E))$, where $\rho_{KE}$ is the joint state of the secret
  key $K$ and the adversary's information $E$, and $\Gamma(E)$ is the
  random variable resulting from measuring the $E$ system with a POVM
  $\Gamma$.}  and the trace distance
criterion\footnote{$\| \rho_{KE} - \tau_K \otimes \rho_E\|$, where
  $\rho_{KE}$ is the joint state of the secret key $K$ and the
  adversary's information $E$ and $\tau_K$ is a fully mixed state.}
capture different properties of a secret key. If a scheme satisfying
the former is used with an insecure quantum channel, then the
resulting key could be insecure, but if the channel only allows the
adversary to measure and store classical information, then the key has
information\-/theoretic security~\cite{KRBM07,portmann2014qkd}. A
scheme satisfying the latter notion \--- the trace distance criterion
\--- constructs a secure key even when the quantum channel used is
completely insecure~\cite{Ren05,BHLMO05,portmann2014qkd}. Neither
criterion is a satisfactory security definition on its own, they both
require a further analysis to prove whether a protocol satisfying them
does indeed distribute a secure key. But now that this has been
done~\cite{BHLMO05,portmann2014qkd}, the trace distance criterion has
become a reference for what a quantum key distribution scheme must
satisfy~\cite{SBCDLP09,TL17}.

Previous works on the computational security of quantum message
transmission \cite{BJ15,alagic2016computational,alagic2018unforgeable}
as well as the new definition of $\QCCA2$ proposed on this paper may
be viewed in the same light. These game\-/based definitions capture
properties of encryption schemes. But before a scheme satisfying these
definitions may be safely used, one needs to analyze how the scheme is
used and what is achieved by it. The constructive definitions
introduced in this work and the reductions from the game\-/based
definitions do exactly this. As a result of this, $\QAE$ or $\QCCA2$
may be used as a benchmark for future schemes \--- though unlike the
trace distance criterion, they are only sufficient criteria, not
necessary ones.

\subsection{Other Related Work}
\label{sec:intro.related}

The desideratum expressed in \autoref{sec:intro.desideratum} is the
fruit of many different lines of research that go back to the late
90's. We give an incomplete overview of some of this work in this
section.

Composable security was introduced independently by Pfitzmann and
Waidner~\cite{PW00,PW01,BPW04,BPW07} and
Canetti~\cite{Can01,CDPW07,Can13}, who each defined their own
framework, dubbed \emph{reactive simulatability} and \emph{universal
  composability} (UC), respectively. Unruh adapted UC to the quantum
setting~\cite{unruh2010quantumUC}, whereas Maurer and Renner's AC
applies to any model of computation, classical or
quantum~\cite{maurer2011abstract}. Quantum UC may however not be used
for finite security without substantial modifications, since it
hard-codes asymptotic security in the framework: machines are defined
by sequences of operators $\left\{\mathcal{E}^{(k)}\right\}_{k}$,
where $k \in \N$ is a security parameter, and distinguishability
between networks of machines is then defined asymptotically in
$k$.\footnote{The object about which ones makes a security statement
  is quite different in an asymptotic and a finite framework. In the
  former it is an infinite sequence of behaviors (e.g., a
  \emph{machine} in UC), whereas in the later it is an element in such
  a sequence (the sequence itself is not necessarily
  well\-/defined). One thus composes different objects in the two
  models, and a composition theorem in one model does not immediately
  translate to a composition theorem in the other.}

Concrete security~\cite{bellare1997concrete} addresses the issues of
reductions and parameters being hidden in $O$\=/notation by requiring
them to be explicit. Theses works consider distinguishing advantages
(or game winning probabilities) as a function of the allowed
complexity or running time of the distinguisher, and aim at proving as
tight statements a possible.  In such an approach, one would have to
define a precise computational model. This, however, is avoided,
meaning that any model in a certain class of meaningful models is
considered equivalent. This unavoidably means that the security
statements are asymptotic, at least with an unspecified linear or
sublinear term. In contrast, the objects we consider, including
distinguishers, are discrete systems and are directly composed as
such, without need for considering a computational model for
implementing the systems.

In the classical case, a model of discrete systems that may be used
for finite security is \emph{random
  systems}~\cite{maurer2002indistinguishability,maurer2007indistinguishability}.
Generalizations to the quantum case have been proposed by Gutoski and
Watrous~\cite{GW07,Gut12} \--- and called \emph{quantum strategies}
\--- by Chiribella, D'Ariano and Perinotti~\cite{chiribella2009combs}
\--- called \emph{quantum combs} \--- and by
Hardy~\cite{Har11,Har12,Har15} \--- \emph{operator tensors}. A model
for discrete quantum systems that can additionally model time and
superpositions of causal structures is the \emph{causal boxes}
framework~\cite{portmann2017causalboxes}.

None of the previous works on computational security of quantum
message transmission satisfy any of the three criteria outlined in
\autoref{sec:intro.desideratum}. These criteria are however standard
by now for quantum key distribution~\cite{portmann2014qkd,TL17}. In
the classical case, they have also been used for computational
security, e.g., \cite{MRT12,CMT13}.

\subsection{Structure of this Paper}
\label{sec:intro.structure}

In \autoref{sec:ac} we introduce the elements needed from
AC~\cite{maurer2011abstract,maurer2012constructive,maurer2016specifications},
and from the discrete system model with which we instantiate AC,
namely quantum combs~\cite{chiribella2009combs}. This allows us to
define the notion of a finite construction of a resource (e.g., a
secure channel) from another resource (e.g., an insecure channel and a
key). In \autoref{sec:constructions} we first define the channels and
other resources needed in this work. Then we give protocols and prove
that they construct various confidential and secure channels, as
outlined in \autoref{sec:intro.constructions}. Finally, in
\autoref{sec:relations} we compare our security definitions to some
game\-/based ones from the literature~\cite{alagic2018unforgeable} and
prove the results described in \autoref{sec:intro.games}.

%%% Local Variables:
%%% TeX-master: "qccFull"
%%% End:
    \section{Abstract \& Constructive Cryptography}
\label{sec:ac}

In this section we give a brief introduction to the Abstract and
Constructive Cryptography (AC) framework. This framework views
cryptography as a resource theory. Players may share certain resources
\--- e.g., secret key, an authentic channel, a public\-/key
infrastructure, common reference strings, etc. \--- and use these to
construct other resources \--- e.g., an authentic channel, a secure
channel, secret key, a bit commitment resource, an idealization of a
multipartite function, etc. A protocol is thus a map between
resources. We give an overview of this in \autoref{sec:ac.theory} and
refer to
\cite{maurer2011abstract,maurer2012constructive,portmann2014qkd,maurer2016specifications}
for further reading.

To illustrate how this framework is used, we model known examples from
the literature in \autoref{sec:example.it} and
\autoref{sec:example.computational}, namely one\-/time authentication
and encryption of quantum messages in the cases of both
information\-/theoretic security and computational security,
respectively. In \aref{app:notation} we define the (quantum) notation
that we use throughout this paper.

The resource theory from \autoref{sec:ac.theory} does not depend on
the model of systems considered \--- it may be instantiated with
classical or quantum systems, synchronous or asynchronous
communication, sequential or timed scheduling. In this work we
consider asynchronous, sequential, quantum communication. This type of
system has been studied by Gutoski and Watrous~\cite{GW07,Gut12},
Chiribella, D'Ariano and Perinotti~\cite{chiribella2009combs}, and
Hardy~\cite{Har11,Har12,Har15}, and we refer to it in the following
using the term from \cite{chiribella2009combs}, namely \emph{quantum
  combs}. Quantum combs are a generalization of \emph{random
  systems}~\cite{maurer2002indistinguishability,maurer2007indistinguishability}
to quantum information theory. In \autoref{sec:ac.instantiation} we
give an brief overview of quantum combs, and refer to the literature
above for more details.

\subsection{A Resource Theory}
\label{sec:ac.theory}

As mentioned at the beginning of this section, the AC framework views
cryptography as a resource theory in which a protocol is a
transformation between resources. More precisely, a protocol $\pi$
uses some resource $\sys R$ (the \emph{assumed} resource) to construct
some other resource $\sys S$ (the \emph{constructed} resource) within
$\eps$, where $\eps$ may be thought of as the error of the
construction. We denote this \begin{equation} \label{eq:construction}
  \sys R \xrightarrow{\pi,\eps} \sys S.\end{equation} A formal
definition of \eqnref{eq:construction} is given below in
\autoref{def:security}. Such a security statement is
\emph{composable}, because if $\pi_1$ constructs $\sys S$ from
$\sys R$ within $\eps_1$ and $\pi_2$ constructs $\sys T$ from $\sys S$
within $\eps_2$, the composition of the two protocols, $\pi_2\pi_1$,
constructs $\sys T$ from $\sys R$ within $\eps_1+\eps_2$, i.e.,
\begin{equation} \label{eq:ac.composition1}
\left.\begin{aligned}
& \sys R \xrightarrow{\pi_1,\eps_1} \sys S \\
& \sys S \xrightarrow{\pi_2,\eps_2} \sys T
\end{aligned}\right\} \implies \sys R \xrightarrow{\pi_2\pi_1,\eps_1+\eps_2} \sys T.
\end{equation}
Furthermore, if $\pi$ constructs $\sys S$ from
$\sys R$ within $\eps$, then when some new resource $\sys U$ is
accessible to the players, the construction still holds, albeit with a
new error $\eps'$, since the computational power of $\sys U$ may be
used by the dishonest parties to improve their attack, i.e.,
\begin{equation} \label{eq:ac.composition2}
\sys R \xrightarrow{\pi,\eps} \sys S \implies \left[ \sys R , \sys U \right]
\xrightarrow{\pi,\eps'} \left[ \sys S , \sys U\right].
\end{equation}
A proof of these two statements may be found in
\aref{app:composition}.

In order to formalize \eqnref{eq:construction}, we need to define the
different elements, namely resources $\sys R,\sys S$, a protocol $\pi$
and the measure of error yielding $\eps$. We do this briefly in the
following, and refer to %\aref{app:ac} and
\cite{maurer2011abstract,maurer2012constructive,portmann2014qkd,maurer2016specifications}
for more details.

\paragraph{Resources.} In a setting with $N$ parties, a
\emph{resource} is an interactive system with $N$ interfaces \--- in
the three party setting, the interfaces are usually denoted $A$, $B$,
and $E$, for Alice, Bob, and Eve \--- allowing different parties to
interact with their interface, i.e., provide inputs and receive
outputs. In the following sections, resources are instantiated with
quantum combs \cite{chiribella2009combs}, but the security definition
is independent of the instantiation, so we defer discussing quantum
combs to \autoref{sec:ac.instantiation}. The only property of
resources that we need in this section is that many resources
$\sys R_1,\dotsc,\sys R_n$ taken together form a new resource, which
we write $\left[\sys R_1,\dotsc,\sys R_n\right]$. A party interacting
at their interface of such a composed resource has simultaneous access
to the corresponding interfaces of the different $\sys R_i$, and may
provide inputs or receive outputs at any of these.

An example of a system with multiple resources is given in
\autoref{sec:example.it}, where we model the security of one-time
authentication and encryption of quantum messages. The scheme, which
is illustrated in \autoref{fig:auth.real}, involves four resources,
$\key^\mu$, $\iqc^{1,n}$, $\qc^{1,m,n}_A$, and $\qc^{1,m,n}_B$, drawn
with square corners. In such a figure, we always put Alice on the
left, Bob on the right, and Eve below. Each party has access to the
inputs and outputs (illustrated by arrows) at their respective sides
of the drawing (at their interfaces). The exact definition of these
resources is provided in \autoref{sec:constructions}, after quantum
combs have been introduced.

\paragraph{Converters.} A \emph{converter} may be thought of as local
operations performed by some party. In \autoref{fig:auth.real}
there are two converters, $\pi_A^\qauth$ and $\pi_B^\qauth$, that
correspond to Alice's and Bob's part of the encryption and
authentication protocol. Converters are always drawn with rounded
corners. Each converter is associated to one interface, written as
subscript, which denotes the party applying the converter. A protocol
is then a tuple of converters, one for each honest party, e.g.,
$\pi^\qauth_{AB} = (\pi_A^\qauth,\pi_B^\qauth)$.

Formally, a converter is defined as a function mapping a resource to
another resource.  Composition of converters is defined as the
composition of the corresponding functions, and we usually write
$\pi'\pi$ instead of $\pi' \circ \pi$. Furthermore, converters acting
at different interfaces always commute, i.e.,
$\pi_A\pi_B = \pi_B\pi_A = \pi_{AB}$. The system drawn in
\autoref{fig:auth.real} is thus a resource given by
\[\sys R = \pi^\qauth_{AB}\left[\key^\mu,\iqc^{1,n},\qc^{1,m,n}_A,\qc^{1,m,n}_B\right].\]

\begin{remark}[Explicit memory and computational requirements]
\label{rem:computation}
In the spirit of making explicit statements \--- as opposed to hiding
parameters in $O$\-/notation or poly\-/time statements \--- in the
following we instantiate converters with systems that are memory-less
and incapable of performing any operation on their inputs, i.e., they
only forward messages between different resources and their outside
interface. This forces any memory and computation capabilities of
parties to be explicitly modeled as resources, e.g., the ``quantum
computer'' resources $\qc^{1,m,n}_A$ and $\qc^{1,m,n}_B$ in
\autoref{fig:auth.real}, which are required for performing the
encryption and decryption operations. We note that one could
equivalently ``absorb'' $\qc^{1,m,n}_A$ into $\pi^\qauth_A$ and $\qc^{1,m,n}_B$
into $\pi^\qauth_B$, which would result in exactly the same security
statement, but not have the computational requirements spelled
out. %\todomarginChris{Add something about closure?}
\end{remark}

\paragraph{Pseudo-metric.} The final component needed before we can
formalize \eqnref{eq:construction} is a pseudo\-/metric,\footnote{A
  pseudo\-/metric $d(\cdot,\cdot)$ is a function which is
  non\-/negative, symmetric, satisfies the triangle inequality, and
  for any $\sys R$, $d(\sys R, \sys R)=0$. If additionally
  $d(\sys R,\sys S) = 0 \implies \sys R = \sys S$, then
  $d(\cdot,\cdot)$ is a metric.} which is used to define the error
$\eps$. We do this in the standard way by defining a
\emph{distinguisher} $\sys D$ to be a system that interacts with a
resource, and outputs a bit. Let $\sys D[\sys R]$ be the random
variable corresponding to the distinguisher's output when interacting
with $\sys R$. Then the functions
\begin{align*}
\Delta^{\sys D}(\sys R, \sys S) & \df \left| \pr{}{\sys D [\sys R] = 0} - \pr{}{\sys D[\sys S] = 0} \right| \\%\label{eq:distance.1}
\intertext{and}
d^{\cD}(\sys R, \sys S) & \df \sup_{\sys D \in \cD} \Delta^{\sys D}(\sys R, \sys S) %\label{eq:distance.2}
\end{align*}
are pseudo\-/metrics for any set of distinguishers $\cD$. We define
the error of a construction using one particular set $\cD$, namely the
set of distinguishers obtained from some distinguisher $\sys D$ by adding
or removing converters between $\sys D$ and the measured resources. As
stated in \autoref{rem:computation}, converters have no memory or
computational power, so this results in a class of ``equivalent''
distinguishers. Thus, for any distinguisher $\sys D$, we define the class
\begin{equation} \label{eq:distinguisher.class}
\cB(\sys D) \coloneqq \left\{ \sys D' \middle| \exists \alpha \text{ such that }
    \sys D\alpha = \sys D' \text{ or } \sys D'\alpha = \sys D \right\},
\end{equation} where $\Delta^{\sys D\alpha}(\sys R,\sys S) = \Delta^{\sys D}(\alpha
\sys R,\alpha \sys S)$. Abusing somewhat notation, we often write $\sys D$
instead of $\cB(\sys D)$. In the following, $d^{\sys D}(\cdot,\cdot)$
always refers to the pseudo\-/metric using the class of distinguishers
generated from $\sys D$ as in \eqnref{eq:distinguisher.class}. 

\paragraph{Cryptographic Security.} We now formalize the notion of
(secure) resource construction in the three party setting, with honest
Alice and Bob and dishonest Eve.

\begin{definition}[Cryptographic security \cite{maurer2011abstract}]
\label{def:security}
Let $\eps$ be a function from distinguishers to real numbers. We say
that a protocol $\pi_{AB} = (\pi_A,\pi_B)$ constructs a resource
$\sys S$ from a resource $\sys R$ within $\eps$ if there exists a
converter $\simul_E$ (called a \emph{simulator}) such that for all $\sys D$,
\[ d^{\sys D}(\pi_{AB} \sys R,\simul_E \sys S) \leq \eps(\sys D).\] If this
holds, then we write
\[\sys R \xrightarrow{\pi,\eps} \sys S.\] When the
resources $\sys R,\sys S$ are clear from the context, we say that
$\pi$ is $\eps$\-/secure.
\end{definition}

$\pi_{AB} \sys R$ is often referred to as the \emph{real} system, and
$\simul_E \sys S$ as the \emph{ideal} one. We emphasis that an ideal (or
\emph{constructed}) resource $\sys S$ will be used as the real (or
\emph{assumed}) resource in the next construction, so the terms
\emph{real} and \emph{ideal} are relative.

\begin{remark}[Efficiency \& negligibility]
\label{rem:efficiency}
Note that \autoref{def:security} does not contain any notion of
efficiency such as poly\-/time, nor does it mention that the error
must be negligible. This is because concepts like efficiency and
negligibility are only defined \emph{asymptotically}. As explained in
\autoref{sec:intro.finite}, we define \emph{finite} security in this
work. In this setting, \emph{efficient} and \emph{negligible} are
ill\-/defined, and cannot be used \--- nor are they needed. Asymptotic
security statements may be recovered by considering sequences of
finite security statements, and taking the limit. This is further
explained in \autoref{sec:example.computational}.
\end{remark}

\subsection{Finite Information-Theoretic Security of One-Time
  Quantum Message Authentication and Encryption}
\label{sec:example.it}

\begin{figure}[tb]
\begin{centering}

\begin{tikzpicture}[
resourceLong/.style={draw,thick,minimum width=3.5cm,minimum height=1cm},
resource/.style={draw,thick,minimum width=1cm,minimum height=1cm},
sArrow/.style={->,>=stealth,thick},
sLine/.style={-,thick},
protocol/.style={draw,thick,minimum width=1.6cm,minimum height=4cm,rounded corners},
pnode/.style={minimum width=1cm,minimum height=1cm}]

\small

\def\t{4.6} % 1.75+.5+1.6+.75
\def\a{3.05} % 1.75+.5+1.6/2
\def\v{1.5}
\def\w{.45}
\def\x{1.25}
\def\z{2.75} %1/2+.5+1+.75

\node[resourceLong] (channel) at (0,-\v) {};
\node[yshift=-1.5,above right] at (channel.north west) {\footnotesize
  Insecure channel $\iqc^{1,n}$};
\node[resource] (qa) at (-\x,0) {};
\node[yshift=-1.5,above] at (qa.north) {\footnotesize $\qc^{1,m,n}_A$};
\node[resource] (qb) at (\x,0) {};
\node[yshift=-1.5,above] at (qb.north) {\footnotesize $\qc^{1,m,n}_B$};
\node[resourceLong] (key) at (0,\v) {};
\node[yshift=-1.5,above right] at (key.north west) {\footnotesize
  Secret key $\key^\mu$};
\node[protocol] (protA) at (-\a,0) {};
\node[yshift=-1.5,above right] at (protA.north west) {\footnotesize $\pi^\qauth_A$};
\node[pnode] (a1) at (-\a,\v) {};
\node[pnode] (a2) at (-\a,0) {};
\node[pnode] (a3) at (-\a,-\v) {};
\node[protocol] (protB) at (\a,0) {};
\node[yshift=-1.5,above left] at (protB.north east) {\footnotesize $\pi^\qauth_B$};
\node[pnode] (b1) at (\a,\v) {};
\node[pnode] (b2) at (\a,0) {};
\node[pnode] (b3) at (\a,-\v) {};

%\node (aliceUp) at (-\t,\v) {};
\node (aliceMiddle) at (-\t,0) {};
%\node (aliceDown) at (-\t,-\v) {};
%\node (bobUp) at (\t,\v) {};
\node (bobMiddle) at (\t,0) {};
%\node (bobDown) at (\t,-\v) {};
\node (eveLeft) at (-\w,-\z) {};
\node (eveRight) at (\w,-\z) {};

\draw[sArrow] (aliceMiddle) to node[auto,pos=.4] {$\rho^M$} (a2);
\draw[sArrow] (a3) to (eveLeft |- a3) to node[auto,swap,pos=.75,xshift=2] {$\sigma^C$} (eveLeft);
\draw[sArrow] (eveRight) to node[auto,swap,pos=.25,xshift=-2] {$\tilde{\sigma}^C$} (eveRight |- b3) to (b3);
\draw[sArrow] (b2) to node[auto,pos=1] {$\tilde{\rho}^M,\bot$} (bobMiddle);

% \draw[sArrow] (aliceUp) to node[auto,pos=.3] {\footnotesize req.} (a1);
% \draw[sArrow] (bobUp) to node[auto,swap,pos=.3] {\footnotesize req.} (b1);
% \draw[sArrow] (a2) to node[auto,swap,pos=.8] {$k',\bot$} (aliceMiddle);
% \draw[sArrow] (b2) to node[auto,pos=.6] {$k'$} (bobMiddle);

\node[draw] (key) at (0,\v) {key};
\draw[sArrow,bend left=15] (key) to node[auto,swap,pos=.85,yshift=-2] {$k$} (a1);
\draw[sArrow,bend right=15] (key) to node[auto,pos=.85,yshift=-2] {$k$} (b1);
\node[pnode] (kLeft) at (-\x+.3,\v) {};
\node[pnode] (kRight) at (\x-.3,\v) {};
\draw[sArrow,bend left=25] (a1) to node[auto,pos=.5] {\footnotesize req.} (kLeft);
\draw[sArrow,bend right=25] (b1) to node[auto,swap,pos=.5] {\footnotesize req.} (kRight);

\node[pnode] (qLeft) at (-\x+.3,0) {};
\node[pnode] (qRight) at (\x-.3,0) {};
\draw[sArrow,bend left=25] (a2) to (qLeft);
\draw[sArrow,bend left=25] (qLeft) to (a2);
\draw[sArrow,bend right=25] (b2) to (qRight);
\draw[sArrow,bend right=25] (qRight) to (b2);
\end{tikzpicture}

\end{centering}
\caption[Real system for quantum
authentication]{\label{fig:auth.real}The real system for quantum
  authentication consists of Alice's and Bob's parts of the protocol,
  $\pi^\qauth_A$ and $\pi^\qauth_B$, and the following four resources:
  a shared secret key $\key^\mu$, a (one-time use) insecure quantum
  channel $\iqc^{1,n}$, and Alice's and Bob's local (one-time use) quantum
  computers $\qc^{1,m,n}_A$ and $\qc^{1,m,n}_B$. Upon receiving a
  message $\rho^M$ at its outer interface, $\pi^\qauth_A$ stores it in
  $\qc^{1,m,n}_A$, then requests a key $k$ from $\key^\mu$ that is
  passed on to $\qc^{1,m,n}_A$, which performs the encryption. Finally
  the ciphertext $\sigma^C$ received from $\qc^{1,m,n}_A$ is input to
  the channel $\iqc^{1,n}$. Upon receiving a ciphertext $\tilde{\sigma}^C$ from
  the channel $\iqc^{1,n}$, $\pi^{\qauth}_B$ stores it in $\qc^{1,m,n}_B$,
  then requests a key $k$ from $\key^\mu$ that is passed on to
  $\qc^{1,m,n}_B$, which performs the decryption. Finally, the result
  of this decryption \--- either the message $\tilde{\rho}^M$ or an error
  message $\bot$ \--- is output at the outer interface of
  $\pi^{\qauth}_B$.}
\end{figure}
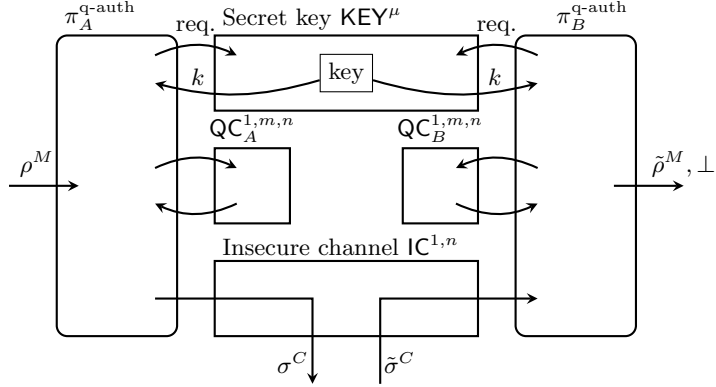

To illustrate \autoref{def:security} we model the security of
one-time encryption and authentication of quantum messages \--- taken
from \cite{portmann2017qauth} \--- which will be used as a subroutine
in \autoref{sec:constructions}. In this scheme, Alice wishes to
send a quantum message to Bob in such a way that an eavesdropper, Eve,
can neither read it nor change it. To do this, they share a uniform
secret key, unknown to Eve, and an insecure channel, which allows Eve
to intercept the message from Alice and send one of her own to
Bob.\footnote{Eve does not have to wait for Alice to use the insecure
  channel to send her own message to Bob. She may first send him a
  chosen ciphertext, and only later intercept the ciphertext generated
  by Alice.} These two resources are denoted $\key^\mu$ and
$\iqc^{1,n}$, and are illustrated in
\autoref{fig:auth.real}. $\mu$ is the length of the secret
key, the $1$ as superscript of $\iqc^{1,n}$ denotes that the channel
is used (at most) once and $n$ is the size (in qubits) of the message
that can be transmitted. Formal definitions of these are provided in
\autoref{sec:constructions.channels}. Two more resources are
needed for Alice and Bob to run their protocol, namely devices that
can perform the encryption and decryption operations. These are
denoted $\qc^{1,m,n}_A$ and $\qc^{1,m,n}_B$ \--- the superscript $1$
specifies that they can each perform one computation, $m$ and $n$
specify the size of the plaintext and ciphertext, i.e., a message of
$m$ qubits is encrypted into a ciphertext of $n$ qubits. Alice's part
of the protocol, $\pi^\qauth_A$, requests $\qc^{1,m,n}_A$ to perform
the encryption and inputs the resulting ciphertext to the channel
$\iqc^{1,n}$. Bob's part of the protocol, $\pi^{1,m,n}_B$, uses
$\qc^{1,m,n}_B$ to perform the decryption, and outputs the
result. This is described in more detail in the caption of
\autoref{fig:auth.real}.

\begin{figure}[tb]
\begin{centering}

\begin{tikzpicture}[
resource/.style={draw,thick,minimum width=3.2cm,minimum height=1.6cm},
rnode/.style={minimum width=.8cm,minimum height=1cm},
sArrow/.style={->,>=stealth,thick},
sLine/.style={-,thick},
sLeak/.style={->,>=stealth,thick,dashed},
simulator/.style={draw,thick,minimum width=7.9cm,minimum height=1cm,rounded corners}]

\small

\def\t{2.6} %1.6+1
\def\v{.3}
\def\w{1.9} %.8+.6+.5
\def\s{.4}
\def\x{5.5}
\def\z{2.9} %.8+.5+1+.6

\node[simulator] (sim) at (\x/2,-\w) {};
\node[xshift=-1.5,below right] at (sim.north east) {\footnotesize $\simul^\qauth_E$};

\node[resource] (ch) at (0,0) {};
\node[yshift=-1.5,above right] at (ch.north west) {\footnotesize
  Secure channel $\sqc^{1,m}$};
\node (alice) at (-\t,0) {};
\node (bob) at (\t,0) {};

\node[resource] (qc) at (\x,0) {}; %{$\epr^{MR}$};
\node[yshift=-1.5,above right] at (qc.north west) {\footnotesize
  Q.\ Computer $\qc^{2,m,n}_E$};
\node[rnode] (qc1) at (\x-.4,-.3) {};
\node[rnode] (sim1) at (\x-.4,-\w) {};
\node[rnode] (qc2) at (\x+.4,-.3) {};
\node[rnode] (sim2) at (\x+.4,-\w) {};

\node (up) at (0,\v) {};
\node (down) at (0,-\v) {};

\node (leakPoint) at (-\s,0) {};
\node (eveLeft) at (\x/2-\s,-\z) {};
\node (eveRight) at (\x/2+\s,-\z) {};

\node at (eveLeft |- sim) {};

\draw[sLine] (alice) to node[auto,pos=.15] {$\rho^M$} (0,0) to node (handle) {} (330:2*\s);
\draw[sArrow] (2*\s,0) to node[auto,pos=.85] {$\rho^M,\bot$} (bob);
\draw[sLeak] (leakPoint.center) to (leakPoint |- sim.north);
\draw[sArrow] (handle.center |- sim.north) to node[auto,swap,pos=.15,xshift=-1.5] {$0,1$} (handle.center);

\draw[sArrow] (eveLeft |- sim.south) to node[auto,pos=.5,xshift=-1] {$\sigma^C$} (eveLeft.center);
\draw[sArrow] (eveRight.center) to node[auto,pos=.5,xshift=-1,swap] {$\tilde{\sigma}^C$} (eveRight |- sim.south);
\draw[sArrow, bend left=15] (qc2) to (sim2);
\draw[sArrow, bend left=15] (sim1) to (qc1);

\end{tikzpicture}

\end{centering}
\caption[Ideal system for quantum
authentication]{\label{fig:auth.ideal}The ideal quantum authentication
  system consists of a (one-time use) secure channel $\sqc^{1,m}$,
  Eve's quantum computer $\qc^{2,m,n}_E$ (that may perform one
  encryption and one decryption), and the simulator
  $\sigma^\qauth_E$. Typically, upon receiving a notification from
  $\sqc^{1,m}$ that Alice sent a message, $\simul^\qauth_E$ will ask
  $\qc^{2,m,n}_E$ to generate a ciphertext $\sigma^C$, which it then
  outputs at Eve's interface. Upon receiving a ciphertext
  $\tilde{\sigma}^C$ at its outer interface, it will ask
  $\qc^{2,m,n}_E$ to check its validity, and forward the result to
  $\sqc^{1,m}$, which then either releases Alice's message $\rho^M$ or
  outputs an error $\bot$ \--- if $\simul^\qauth_E$ receives a
  ciphertext $\tilde{\sigma}^C$ before any message $\rho^M$ is input
  at Alice's interface, then the ciphertext is declared invalid and
  $\sqc^{1,m}$ always outputs $\bot$.}
\end{figure}
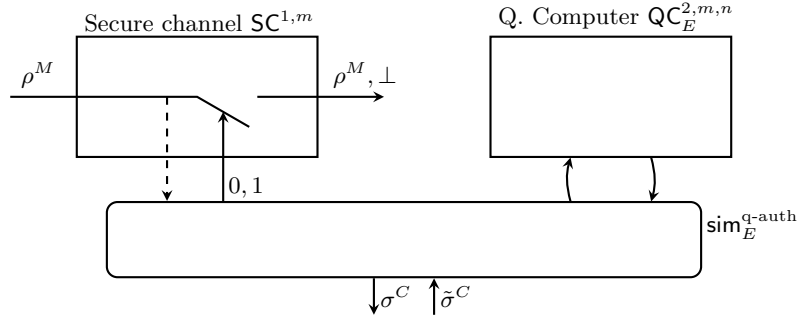

We now need to design the ideal system. The goal of this protocol is
to construct a secure channel, i.e., one which does not leak any
information about the message to Eve \--- except the message length
\--- and does not allow her to change the message either, nor insert
one of her own. In the real world, Eve may either prevent any
ciphertext from reaching Bob or jumble the ciphertext resulting in Bob
receiving garbage (and outputting an error), so this must also be
reflected in the ideal world. Thus, the ideal resource constructed
\--- the secure channel $\sqc^{1,m}$ illustrated in
\autoref{fig:auth.ideal} \--- notifies Eve when a message is
input at Alice's interface (denoted by the dashed
arrow),\footnote{Since the construction is parameterized by the
  message length $m$, there is no need to leak this information at
  Eve's interface. In a construction where $m$ is unknown to Eve, a
  secure channel may additionally leak (an upper bound on) the message
  size to Eve.} and waits for a bit at Eve's interface that tells it
whether to output the original message at Bob's interface or an error
message $\bot$.\footnote{If Eve does not provide this bit, the
  channel does not output anything.} The superscript $1$ specifies
that this channel may be used once and $m$ denotes the size of the
message it transmits.

If such a protocol is used as a subroutine in a larger context, Eve
could use the computational power of Alice and Bob to her advantage,
e.g., if she needs the encryption or decryption operations to be
applied to some quantum state, she can provide this state to the
honest players and intercept their output. Such a construction thus
provides some computation power to Eve, namely, it also constructs a
resource $\qc^{2,m,n}_E$ which can perform one encryption and one
decryption for a randomly chosen key.\footnote{For some protocols this
  may be providing Eve with unnecessary power, e.g., if the ciphertext
  looks like a full mixed state, it may be sufficient for
  $\qc^{2,m,n}_E$ to generate a purification of such a state and not
  need the ability to perform an encryption.}$^{\text{,}}$\footnote{In
  an asymptotic setting one would not care about this, since the
  protocol is efficient this would vanish in some poly\-/time
  statement. But in a finite setting where we keep track of all
  computation, one should model this explicitly. See also
  \autoref{rem:computation}.}

% For the security proof to go through, we additionally require
% $\qc^{1,m,n}_E$ to be able to generate $m$ EPR pairs and measure the
% same system in the Bell basis.\footnote{Technically, this corresponds
%   to the computational power needed by the simulator.}

The security definition, \autoref{def:security}, requires there
to exist a simulator such that the real and ideal worlds are
indistinguishable. The exact simulator $\simul^{\qauth}_E$ will depend
on the protocol analyzed, but \--- as described in more detail in the
caption of \autoref{fig:auth.ideal} \--- it will typically ask
$\qc^{2,m,n}_E$ to generate a ciphertext to simulate that of Alice,
then ask $\qc^{2,m,n}_E$ to perform a decryption to check the validity
of the ciphertext received from Eve. The result of this decryption
will decide whether Alice's message is released to Bob or not.

For a specific protocol $\pi^{\qauth}_{AB}$, a security proof in the
information\-/theoretic setting consists in showing that this is
secure for all distinguishers, i.e., one must find a
$\simul^{\qauth}_E$ and some \emph{constant function} $\eps$ such that the real and ideal
systems (\autoref{fig:auth.real} and \ref{fig:auth.ideal}) are
$\eps$-close, e.g., \autoref{lem:qauth} here below.

\begin{lemma}{\cite{portmann2017qauth}}
\label{lem:qauth}
Let $\pi^{\qauth}_{AB}$ be any member of a family of protocols that
encrypt the message using (weak) purity testing
codes~\cite{barnum2002qauth,portmann2017qauth} from a set $\cC$ of
size $\log |\cC| = \nu$ with error $\eps \in \R$. Then there exists a
$\simul^{\qauth}_E$ such that for any $\sys D$,
\begin{multline*}
  d^{\sys D}\left(\pi^{\qauth}_{AB}\left[\key^{\nu+m+n},\iqc^{1,n},\qc^{1,m,n}_A,\qc^{1,m,n}_B\right],\simul^{\qauth}_E\left[\sqc^{1,m},\qc^{2,m,n}_E\right]\right)
  \\\leq \max\{2^{m-n},\eps\},\end{multline*} i.e.,
\[\left[\key^{\nu+2m+n},\iqc^{1,n},\qc^{1,m,n}_A,\qc^{1,m,n}_B\right]
  \xrightarrow{\pi^{\qauth}_{AB},\max\{2^{m-n},\eps\}}
  \left[\sqc^{1,m},\qc^{2,m,n}_E\right].\]
\end{lemma}
Note that since we only consider constant functions $\eps$ in
information\-/theoretic security, we write $\eps$ both for the
function and the value $\eps(\sys D) \in \R$.

Plugging in a specific instantiation of a (weak) purity testing code
with $\nu = 3n$ and $\eps = 2^{m-n}$ \--- the unitary $2$\-/design
construction~\cite{portmann2017qauth} \--- one then obtains the
following.

\begin{corollary}{\cite{portmann2017qauth}}
\label{cor:qauth}
Let $\pi^{\qauth}_{AB}$ be the unitary $2$\-/design construction, then
\[\left[\key^{m+4n},\iqc^{1,n},\qc^{1,m,n}_A,\qc^{1,m,n}_B\right]
  \xrightarrow{\pi^{\qauth}_{AB},2^{m-n}}
  \left[\sqc^{1,m},\qc^{2,m,n}_E\right].\]
\end{corollary}

One may easily obtain an asymptotic security statement
from \autoref{cor:qauth} by taking a sequence of constructions
parameterized by, e.g., $r=n-m$, and taking the limit as
$r \to \infty$.

\subsection{Finite Computational Security of One-Time Quantum Message
  Authentication and Encryption}
\label{sec:example.computational}

In \autoref{sec:example.it} we gave an example of a construction
that has \emph{information\-/theoretic} security: the distance between
the real and ideal systems is bounded for every distinguisher, i.e.,
$\eps$ is a constant function. If one considers \emph{computational}
security, then one is not interested in security in all contexts, but
only when the distinguisher cannot solve problems considered to be
``hard''. A security proof is then a reduction to one of these hard
problems. In the case of finite security we consider reductions to
finite problems.

\begin{figure}[tb]
\begin{centering}

\begin{tikzpicture}[
resourceLong/.style={draw,thick,minimum width=3.5cm,minimum height=1cm},
resource/.style={draw,thick,minimum width=1cm,minimum height=1cm},
sArrow/.style={->,>=stealth,thick},
sLine/.style={-,thick},
protocol/.style={draw,thick,minimum width=1.6cm,minimum height=4cm,rounded corners},
pnode/.style={minimum width=1cm,minimum height=1cm}]

\small

\def\t{4.9} % 1.75+.5+1.6+1.05
\def\a{3.05} % 1.75+.5+1.6/2
\def\v{1.5}
\def\w{.45}
\def\x{1.25}
\def\z{3.05} %1/2+.5+1+1.05

\node[resourceLong] (channel) at (0,-\v) {};
\node[yshift=-1.5,above right] at (channel.north west) {\footnotesize
  Insecure channel $\iqc^{1,n}$};
\node[resource] (qa) at (-\x,0) {};
\node[yshift=-1.5,above] at (qa.north) {\footnotesize $\qc^{1,m,n}_A$};
\node[resource] (qb) at (\x,0) {};
\node[yshift=-1.5,above] at (qb.north) {\footnotesize $\qc^{1,m,n}_B$};
\node[resourceLong] (key) at (0,\v) {};
\node[yshift=-1.5,above right] at (key.north west) {\footnotesize
  Secret key $\prg^{r,\mu}$};
\node[protocol] (protA) at (-\a,0) {};
\node[yshift=-1.5,above right] at (protA.north west) {\footnotesize $\pi^\qauth_A$};
\node[pnode] (a1) at (-\a,\v) {};
\node[pnode] (a2) at (-\a,0) {};
\node[pnode] (a3) at (-\a,-\v) {};
\node[protocol] (protB) at (\a,0) {};
\node[yshift=-1.5,above left] at (protB.north east) {\footnotesize $\pi^\qauth_B$};
\node[pnode] (b1) at (\a,\v) {};
\node[pnode] (b2) at (\a,0) {};
\node[pnode] (b3) at (\a,-\v) {};

%\node (aliceUp) at (-\t,\v) {};
\node (aliceMiddle) at (-\t,0) {};
%\node (aliceDown) at (-\t,-\v) {};
%\node (bobUp) at (\t,\v) {};
\node (bobMiddle) at (\t,0) {};
%\node (bobDown) at (\t,-\v) {};
\node (eveLeft) at (-\w,-\z) {};
\node (eveRight) at (\w,-\z) {};

\draw[sArrow] (aliceMiddle) to node[auto,pos=.3] {$\rho^M$} (a2);
\draw[sArrow] (a3) to (eveLeft |- a3) to node[auto,swap,pos=.8,xshift=2] {$\sigma^C$} (eveLeft);
\draw[sArrow] (eveRight) to node[auto,swap,pos=.2,xshift=-2] {$\tilde{\sigma}^C$} (eveRight |- b3) to (b3);
\draw[sArrow] (b2) to node[auto,pos=.9] {$\tilde{\rho}^M,\bot$} (bobMiddle);

% \draw[sArrow] (aliceUp) to node[auto,pos=.3] {\footnotesize req.} (a1);
% \draw[sArrow] (bobUp) to node[auto,swap,pos=.3] {\footnotesize req.} (b1);
% \draw[sArrow] (a2) to node[auto,swap,pos=.8] {$k',\bot$} (aliceMiddle);
% \draw[sArrow] (b2) to node[auto,pos=.6] {$k'$} (bobMiddle);

\node[draw] (key) at (0,\v) {key};
\draw[sArrow,bend left=15] (key) to node[auto,swap,pos=.5,yshift=-2] {$k$} (a1);
\draw[sArrow,bend right=15] (key) to node[auto,pos=.5,yshift=-2] {$k$} (b1);
\node[pnode] (kLeft) at (-\x+.3,\v) {};
\node[pnode] (kRight) at (\x-.3,\v) {};
\draw[sArrow,bend left=25] (a1) to node[auto,pos=.5] {\footnotesize req.} (kLeft);
\draw[sArrow,bend right=25] (b1) to node[auto,swap,pos=.5] {\footnotesize req.} (kRight);

\node[pnode] (qLeft) at (-\x+.3,0) {};
\node[pnode] (qRight) at (\x-.3,0) {};
\draw[sArrow,bend left=25] (a2) to (qLeft);
\draw[sArrow,bend left=25] (qLeft) to (a2);
\draw[sArrow,bend right=25] (b2) to (qRight);
\draw[sArrow,bend right=25] (qRight) to (b2);

\draw[thick,dashed] (-4.05,2.5) -- (-2,2.5) -- (-2,.75) -- (2,.75) --
(2,2.5) -- (4.05,2.5) -- (4.05,-2.2) -- (-4.05,-2.2) -- cycle;
\node at (4.3,2.25) {$\sys C$};
\end{tikzpicture}

\end{centering}
\caption[Real system for quantum authentication with
PRG]{\label{fig:auth.prg}The real system for quantum authentication
  with a PRG. It is identical to \autoref{fig:auth.real},
  except for the uniform key resource $\key^\mu$ which has been
  replaced by a pseudo\-/random key resource $\prg^{r,\mu}$ which
  produces a key of length $\mu$ from a uniform seed of length
  $r$. The system within the dashed line is labeled $\sys C$, and is
  relevant in the reduction.}
\end{figure}
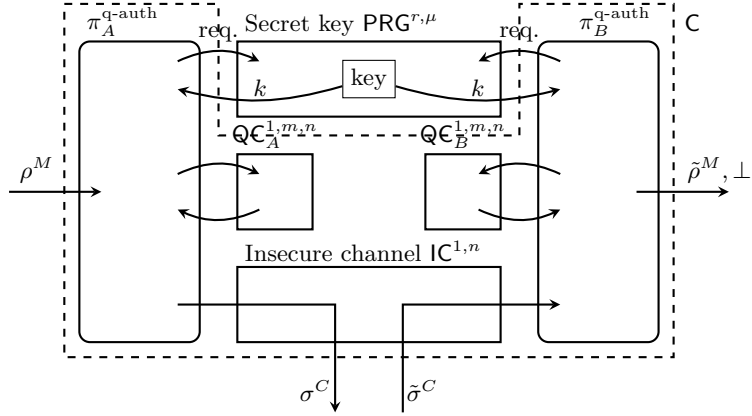

To illustrate this we will model a computational version of the
quantum authentication and encryption protocol from
\autoref{sec:example.it}. In fact, we will analyze exactly the
same protocol, but just replace the uniform key resource $\key^\mu$
with pseudo\-/random key resource $\prg^{r,\mu}$, where $r$ is the
length of the (uniform) seed of the pseudo\-/random generator (PRG)
\--- assumed to be hard\-/coded in the PRG \--- and $\mu$ is the
length of the key it produces. The real world now looks exactly like
\autoref{fig:auth.real}, except for the swap of the key for a
PRG, which we have drawn in \autoref{fig:auth.prg}.

The ideal world is identical to \autoref{fig:auth.ideal} as our goal
is still to construct the same ideal channel resource
$\sqc^{1,m}$. What changes is the error function $\eps$ which will be
small for some distinguishers and large for others. More precisely, we
show that if $\sys D$ can distinguish the real from the ideal systems,
then there exists an (explicit) distinguisher $\sys D'$ that can
distinguish $\prg^{r,\mu}$ from $\key^\mu$.

\begin{lemma}
\label{lem:qauth.computational}
Let $\pi^\qauth_{AB}$, $\simul^\qauth_E$ and $\eps$ be such that for
any $\sys D$,
\[d^{\sys D}\left(\pi^{\qauth}_{AB}\left[\key^{\mu},\iqc^{1,n},\qc^{1,m,n}_A,\qc^{1,m,n}_B\right],\simul^{\qauth}_E\left[\sqc^{1,m},\qc^{2,m,n}_E\right]\right)
  \leq \eps(\sys D).\] Then for any $\sys D$ we have
\begin{multline} \label{eq:qauth.computational.1}
  d^{\sys D}\left(\pi^{\qauth}_{AB}\left[\prg^{r,\mu},\iqc^{1,n},\qc^{1,m,n}_A,\qc^{1,m,n}_B\right],\simul^{\qauth}_E\left[\sqc^{1,m},\qc^{2,m,n}_E\right]\right) \\
\leq d^{\sys{DC}}\left(\prg^{r,\mu},\key^{\mu}\right) + \eps(\sys D),
\end{multline}
where $\sys C$ is the system drawn in dashed lines in
\autoref{fig:auth.prg}, and $\sys{DC}$ is a distinguisher
obtained by composing $\sys D$ and $\sys C$ in the obvious way.
\end{lemma}

\begin{proof} Follows from the triangle inequality and 
  $d^{\sys D}(\sys{CR},\sys{CS}) \leq
  d^{\sys{DC}}(R,S)$. %\todoChris{prove this}
\end{proof}

\begin{remark}
\label{rem:finite}
The security statement in \eqnref{eq:qauth.computational.1} is a
\emph{reduction}. It states that if it is hard to distinguish
$\prg^{r,\mu}$ from $\key^\mu$, then the scheme must be secure. Note
that this is a \emph{finite} reduction. We do not need to care whether
$\prg^{r,\mu}$ and $\key^\mu$ are asymptotically
indistinguishable,\footnote{Asymptotic indistinguishability means that
  one sets $\mu = p(r)$ for some polynomial $p$ and takes the limit as
  $r \to \infty$.} but whether the finite instantiation for some fixed
$r$ and $\mu$ is indistinguishable \--- which is the relevant
statement for an implementation. Unlike the asymptotic case, where one
does not need to define the reduction system $\sys C$ or error $\eps$
precisely but just make a statement about their limit, here, the exact
reduction $\sys C$ and parameter $\eps$ are crucial for the statement
to be complete.
\end{remark}

Suppose now that one has a bound $\eps'$ on the distinguishability of
the PRG from a uniform key. Then one may plug this into
\autoref{lem:qauth.computational} to get a (direct) security statement
for the construction.\footnote{\autoref{cor:qauth.computational} is a
  special case of the composition theorem: security of the encryption
  scheme with a PRG follows the security of the PRG and the security
  of the encryption scheme.}

\begin{corollary}
\label{cor:qauth.computational}
Let $\eps'$ be a function such that for all $\sys D$,
$d^{\sys D}(\prg^{r,\mu},\key^\mu) \leq \eps'(\sys D)$. Then
\begin{equation}
\label{eq:qauth.computational.2}
\left[\prg^{r,\mu},\iqc^{1,n},\qc^{1,m,n}_A,\qc^{1,m,n}_B\right]\xrightarrow{\pi^{\qauth}_{AB},\eps(\cdot)+\eps'(
  \cdot \sys C)}\left[\sqc^{1,m},\qc^{2,m,n}_E\right],\end{equation}
where $\eps'(\cdot \sys C)$ is the function $\sys D \mapsto
\eps'(\sys{DC})$.
\end{corollary}

\paragraph{Asymptotic security.} \eqnsref{eq:qauth.computational.1} and
\eqref{eq:qauth.computational.2} are finite security statements for
the protocol $\pi^{\qauth}_{AB}$. Asymptotic statements may be
obtained as a corollary if one takes sequences of finite statements
and takes the limit. The (generic) construction presented so far has
multiple parameters, namely $\eps,m,n,r,\mu$. To define a sequence and
take the limit, one has to simplify this. We do so by
plugging in specific constructions. Suppose that we have a PRG with
$\mu = p(r)$ for some polynomial $p$, and we take the unitary
$2$\-/design construction from \autoref{cor:qauth}. Then
\eqnref{eq:qauth.computational.1} becomes
\begin{equation*}
    \label{eq:qauth.computational.3}
    d^{\sys
    D^{m,r}}\left(\pi^{\qauth}_{AB} \sys R^{m,r},
    \simul^{\qauth}_E \sys S^{m,r} \right) %\\
  \leq
  d^{\sys{D}^{m,r}\sys{C}^{m,r}}\left(\prg^{r,p(r)},\key^{p(r)}\right)
  + 2^{m-p(r)},
\end{equation*}
where we have set
$\sys R^{m,r} \df
\left[\prg^{r,p(r)},\iqc^{1,p(r)},\qc^{1,m,p(r)}_A,\qc^{1,m,p(r)}_B\right]$
and $\sys S^{m,r} \df \left[\sqc^{1,m},\qc^{2,m,p(r)}_E\right]$, and
parameterized $\sys D$ and $\sys C$ by $m$ and $r$ for
consistency. Let us view $m$ as fixed and take sequences with
$r \in \N$.\footnote{One may alternatively set $m = q(r)$ for some appropriate
polynomial $q$.}

Since $\sys C^{m,r}$ is poly\-/time \--- $\sys C^{m,r}$ basically runs
the protocol, which is efficient \--- the error $2^{m-p(r)}$ is
negligible and we believe that
$d^{\tilde{\sys D}^{m,r}}\left(\prg^{r,p(r)},\key^{p(r)}\right)$
is negligible for any poly\-/time $\tilde{\sys D}^{m,r}$ (in
particular for $\tilde{\sys D}^{m,r} = \sys{D}^{m,r}\sys{C}^{m,r}$),
then we can recover the standard asymptotic statement from
\eqnref{eq:qauth.computational.3}. More precisely, it then follows
that for any poly\-/time $\sys D^{m,r}$,
\begin{equation}
\label{eq:qauth.computational.4}
\Delta^{\sys
    D^{m,r}}\left(\pi^{\qauth}_{AB} \sys R^{m,r},
    \simul^{\qauth}_E \sys S^{m,r}\right)
  \leq \eps_{m,r}
\end{equation}
for some negligible $\eps_{m,r}$.

%\todoChris{Add comment about how this might be too strong. Explicit
%reduction. See Rogaway's ``formalizing Human Ignorance'' paper.}

\subsection{Instantiating AC}
\label{sec:ac.instantiation}

We consider \emph{sequential scheduling} in this work, i.e., a system
receives a message, then produces an output, receives another input,
generates another output, etc. Here, the messages sent and received
may be quantum, but the party to whom a message is sent is given by
classical information.\footnote{An alternative model of systems that
  allows quantum scheduling \--- messages can be in a superposition of
  sent and not sent, or a superposition of sent to player $A$ and sent
  to player $B$ \--- is the Causal Boxes
  model~\cite{portmann2017causalboxes}, which can additionally capture
  relativistic causal constraints
  \cite{vilasini2017relativisitic}. For the current work however,
  (classical) sequential scheduling is sufficient, so we use the
  simpler quantum combs~\cite{chiribella2009combs}.}

\paragraph{Quantum Combs.} 
Let $\Hi_A$ be the input Hilbert space of a system $\sys R$, $\Hi_B$
be the output Hilbert space and $\Hi_M$ be some internal memory.
$\sys R$ is fully defined by a completely positive, trace\-/preserving
(CPTP) map
\begin{equation} \label{eq:sys.1} \cE : \lo{AM} \to
  \lo{BM},\end{equation} which specifies what output is produced and
how the internal memory is updated as a function of the input and
current internal memory. Each time an input $\rho \in \lo{A}$ is
received (which may be part of a larger entangled system), the map
$\cE$ is applied, and the output $\sigma \in \lo{B}$ is sent to the
appropriate party.

This is however a redundant description of a system, since different
maps $\cE$ may produce exactly the same output behavior. The works on
quantum combs \cite{GW07,Gut12,chiribella2009combs,Har11,Har12,Har15}
define such a system for $n$ inputs as a CPTP map
\begin{equation} \label{eq:sys.2} \Phi^n : \cL\left(\Hi^{\otimes
      n}_A\right) \to \cL\left(\Hi^{\otimes n}_B\right)\end{equation}
subject to a causality constraint that enforces that an output in
position $i$ can only depend on inputs in positions $\leq i$. Such a
system may have various descriptions, e.g., as the (unique) Choi
matrix corresponding to the map $\Phi^n$, as a map $\cE$ as in
\eqnref{eq:sys.1}, or as pseudo\-/code that describes the behavior of
the system.

In this work we often use pseudo\-/code to describe resources. Though
each time it corresponds to a specific quantum comb that is uniquely
defined by a map as in \eqnref{eq:sys.2}.

\paragraph{Resources as Combs.}
As stated at the beginning of \autoref{sec:ac.instantiation}, in our
model the sender and receiver of messages are given by classical
information. Formally, this means the input and output Hilbert space
of a resource has the form $\cH_A = \bigoplus_i \cH_{A_i}$, and the
first operation of a resource is a projection on the different
subspaces $A_i$ to determine this classical information. For example,
the different $i$ correspond to the different interfaces of a resource
or different ports at an interface.

When two resources $\sys R$ and $\sys S$ with input (or output) spaces
$\cH_A$ and $\cH_B$ accessible in parallel, the resulting resource
$\sys T = \left[\sys R, \sys S\right]$ is a new quantum comb with
input (respectively, output) space given by $\cH_A \oplus \cH_B$, and
the initial measurement determines which sub\-/resource ($\sys R$ or
$\sys S$) receives and processes the input. Similarly, plugging a
converter $\alpha$ into a resource $\sys R$, results in a new resource
$\sys U = \alpha \sys R$, whose corresponding quantum comb is computed
from the comb of $\sys R$ and the pairs of ports connected by
$\alpha$.

More generally, the works on quantum combs
\cite{GW07,Gut12,chiribella2009combs,Har11,Har12,Har15} show how such
objects may be composed to form networks \--- which are again objects
of the same type, namely quantum combs \--- and how one defines a
(distinguisher\-/based) pseudo\-/metric on combs. We refer to them for
more details.

%%% Local Variables:
%%% TeX-master: "qccFull"
%%% End:
    \section{Constructing Quantum Cryptographic Channels}
\label{sec:constructions}

In \autoref{sec:constructions.keys} and
\autoref{sec:constructions.channels} we formalize the resources used
in our constructions.  Then, starting from the insecure quantum
channel $\iqc$, a shared secret key $\key$ and local pseudo random
function $\prf$, we show how to construct (1) the ordered secure
quantum channel $\osqc$ in \autoref{sec:constructions.osc} and (2) the
Pauli-malleable confidential quantum channel $\xqc$ in
\autoref{sec:constructions.xcc}.  A construction of the ordered secure
quantum channel $\osqc$ from one which is secure but not ordered
($\sqc$) is also presented in \aref{app:osc}.

\subsection{Key Resources}	
\label{sec:constructions.keys}

A (shared) secret key resource corresponds to a system that provides a
key $k$ to the honest players, but nothing to the adversary.

\begin{definition}[Symmetric (Classical) Key $\key$]
    The resource $\key$ is associated with a probability distribution
    $P_K$ for (classical) key space $\Key$. A key $k \in \Key$ is
    drawn according to $P_K$ and stored in the resource.
    \begin{itemize}
        \item {\bf Interface $\iA$:} On input $\lGetKey$, $k$ is
          output at interface $\iA$.
        \item {\bf Interface $\iB$:} On input $\lGetKey$, $k$ is
          output at interface $\iB$.
        \item {\bf Interface $\iE$:} Inactive.
    \end{itemize}
  \end{definition}

  In the computational setting, instead of sharing a long key, players
  often share a short key which is used as seed in a local key
  expansion scheme. On such key expansion scheme which we use in this
  work is a so\-/called \emph{pseudo random function}. It is
  essentially a family of functions which looks random.

\begin{definition}[Pseudo Random Function $\prf^{r,\nu,\mu}$]
  The resource $\prf^{r,\nu,\mu}$ is associated to a family of
  functions
  $\left\{f_k:\{0,1\}^{\nu} \to \{0,1\}^{\mu}\middle| k \in
    \{0,1\}^r\right\}$ and has an internal variable $\vSeed$ of length
  $r$. The functions in the family have input length $\nu$ and output
  length $\mu$. The resource is local to one party only. 
  Let this party's interface be labeled $\iface X$. % We do not specify the input length in the parameter as it will be clear in the context. Since the resource will be requested with at most $\ell$ different inputs, the domain size $2^\nu$ is guaranteed to be larger or equal to $\ell$.
    \begin{itemize}
        \item {\bf Interface $\iface X$:}
        	\begin{itemize} 
        	\item On input $\lSeed(s)$, set variable $\vSeed$ to $s$.
        	\item On input $\lInput(x)$, output $f_\vSeed(x)$ at interface $\iface X$.
        	\end{itemize}
    \end{itemize}
\end{definition}

The above definition of a $\prf$ does not contain any criterion for what
it means to ``look random''. This is defined in a second step as
distinguishability from a uniform random function.

\begin{definition}[Uniform Random Function $\urf^{\nu, \mu}$]
  The resource $\urf^{\nu, \mu}$ picks a function $f$ from all
  functions $\{0,1\}^{\nu} \to \{0,1\}^{\mu}$ uniformly at random. 
    \begin{itemize}
        \item {\bf Interface $\iface A$:} On input $\lInput(x)$, output $f(x)$ at interface $\iA$.
        \item {\bf Interface $\iface B$:} On input $\lInput(x)$, output $f(x)$ at interface $\iB$.
        \item {\bf Interface $\iface E$:} Inactive.
    \end{itemize}
\end{definition}

Let $\pi^\prf$ be the trivial protocol which uses a (short) shared key
(from a $\key$ resource) and plugs it as seed in a $\prf$ resource,
and let $\eps^\prf(\sys D)$ be the advantage the distinguisher $\sys
D$ has in distinguishing such a construction from a $\urf$, i.e., for all $\sys D$
\[d^{\sys D}(\pi^\prf[\key^r, \prf_A^{r,\nu,\mu}, \prf_B^{r,\nu,\mu}],
  \urf^{\nu,\mu}) \leq \epsilon^\prf(\sys D),\] where
$d^{\sys D}(\cdot, \cdot)$ is the distinguisher pseudo\-/metric as
defined in \autoref{sec:ac.theory}. In terms of AC construction, this means that
\begin{equation}\label{eq:prf2urf}
\left[\key^r, \prf_A^{r,\nu,\mu}, \prf_B^{r,\nu,\mu}\right] 
\xrightarrow{\pi^{\prf},\epsilon^\prf}
  \urf^{\nu,\mu} .
\end{equation}
Concrete constructions of PRFs proven secure in the presence of
quantum adversaries may be found in \cite{Zha12pqPRF}.
% Note that in the finite setting we do not categorize functions
% $\epsilon^\prf$ in good (e.g., negligible for efficient
% quantum distinguishers, e.g. \cite{Zha12pqPRF}) and bad, but instead provide the function explicitly.

\subsection{Channel resources}
\label{sec:constructions.channels}

We consider three-party channels in this work: the sending party Alice
has access to interface $\iA$, the receiving party Bob to interface
$\iB$, and the adversary Eve to interface $\iE$.  We model all our
channels in the following way: upon an input at interface $\iA$, an
output is generated at interface $\iE$, while upon an input at
interface $\iE$, an output is generated at interface $\iB$.  Moreover,
we consider multi-message channels parameterized by $\ell$, that is,
Alice and Eve can provide at most $\ell$ inputs at their respective
interfaces. These inputs can be entangled with each
other. We model quantum channels, therefore inputs and outputs
to and from the channels' interfaces are quantum systems. The channels
are also parameterized by $m$, the size of each message in qubits.

In the following we introduce the formal description of the channels considered in this work by specifying the behavior they assume upon inputs at their $\iA$ and $\iE$ interfaces.
First, we consider the weakest possible channel, that is, the \emph{insecure} one, which gives full control to the adversary Eve. Eve receives all the message that Alice inputs to the channel. Bob receives all the messages that Eve inputs to the channel. 

\begin{definition}[Insecure Quantum Channel $\iqc^{\ell,m}$]
    \begin{itemize}
    \item {\bf Interface $\iA$:} On receiving an input system in some
      state $\state$, perform an identity map and output the same
      system at interface $\iE$.
    \item {\bf Interface $\iE$:} On receiving an input system in some
      state $\state'$, perform an identity map and output the same
      system at interface $\iB$.
    \end{itemize}
    Interface $\iA$ and $\iE$ will receive at most $\ell$ inputs and
    ignore the rest. The quantum systems input at interface $\iA$ and
    $\iE$ and output at interface $\iB$ have length $m$ in qubits.
\end{definition}

Next, we enhance the insecure channel by providing some form of
confidentiality on the states input by Alice.  More precisely, we
allow Eve to only get a notification that a new message has arrived in
interface $\iA$, but still, Eve will retain the capability to
\emph{modify} each input $\state^{\rA_i}$ (held in register $\rA_i$).

% For this, the channel keeps an internal register $\rC$ with state $\state^\rC$ and defines a CPTP map $\cL(\Hi_{\rA_i\rC}) \to \cL(\Hi_{\rC})$ which is applied at each new input at registers $\rA_i$ and $\rC$, whereas Eve not only inputs a state $\state^\rE$ (held in register $\rE$), but also a CPTP map $\mE: \cL(\Hi_{\rC \rE}) \to \cL(\Hi_{\rC\rB})$, which is applied on the registers $\rC$ and $\rE$ each time before the state $\state^\rB$ (held in register $\rB$) is output at interface $\iB$. 
% \begin{definition}[Confidential Quantum Channel $\cqc^{\ell,m}$]
%     The channel keeps registers $\rA_1, \rA_2, \dots, \rA_\ell$, $\rC$, $\rB$, and $\rE$, and defines a mapping $\mC:\cL(\Hi_{\rA_i\rC})\to\cL(\Hi_{\rC})$.
%     \begin{itemize}
%         \item {\bf Interface $\iA$:} On $i$-th input state $\state^{\rA_i}$,held in register $\rA_i$, output $\lNewMsg$ at interface $\iE$, and update register $\rC$ by applying $\mC$ to registers $\rA_i\,\rC$. 
%         \item {\bf Interface $\iE$:} On input the map-state pair $(\mE,\state^\rE)$, where $\mE: \cL(\Hi_{\rC \rE}) \to \cL(\Hi_{\rC\rB})$ is a CPTP map and $\state^\rE$ is held in register $\rE$, apply $\mE$ to registers $\rC$ and $\rE$ and then output $\state^\rB$ held in register $\rB$ at interface $\iB$. 
%     \end{itemize}
%     Interface $\iA$ and $\iE$ will receive at most $\ell$ inputs and ignore the rest. The state input at interface $\iA$ and output at interface $\iB$ always have length $m$ in qubits.
% \end{definition}

Here, one may consider different ways in which Eve is allowed to
modify the messages. The first channel we consider grants Eve the
power to insert fully mixed states on the channel, as well as
performing Pauli operators (bit flips and phase flips) on Alice's
message and decide when each message gets delivered. This is modeled
by keeping registers $\rA_i$ for each new input at interface $\iA$,
and allowing Eve to input indices specifying which register should be
modified and output at interface $\iB$. Along with the index, Eve also
inputs a string of length $2m$, indicating on which qubits of the
message to apply Pauli operators. If Eve wants a fully mixed state to
be output at Bob's, she inputs $\bot$ at her interface and the channel
generates the corresponding
state. %Recall that by the no-cloning theorem, Alice's messages cannot be reused by Eve.
\begin{definition}[Pauli-Malleable Confidential Quantum Channel $\xqc^{\ell,m}$]
\label{def:xqc}
	The channel keeps registers $\rA_1, \rA_2, \dots, \rA_\ell$, initially set to $\bot$.
    \begin{itemize}
        \item {\bf Interface $\iA$:} Upon receiving the $i$-th input
          in some state $\rho$, this system is stored in register $\rA_i$, and
          $\lNewMsg$ is output at interface $\iE$.
        \item {\bf Interface $\iE$:} 
        \begin{itemize}
        \item On input $(j, k)\in [l] \times \{0,1\}^{2m}$, output
          system in state $P_k\rho^{\rA_j}P_k$ at interface $\iB$,
          where $\rho^{\rA_j}$ is the state of the system held in
          register $\rA_j$ and $P_k$ is the Pauli operator defined by
          the string $k$ \--- see \aref{app:notation}. If the tuple is
          invalid or $\rho^{\rA_j}$ is $\bot$, the input is considered
          as $\bot$. After the output, the state in register $\rA_j$
          becomes $\bot$.
        \item On input $\bot$, output a fully mixed stated
          $\frac{1}{2^m}I_{2^m}$ at interface $\iB$. 
       \end{itemize}
    \end{itemize}
    Interface $\iA$ and $\iE$ will receive at most $\ell$ inputs and
    ignore the rest. The quantum systems input at interface $\iA$ and
    output at interface $\iB$ always have length $m$ in qubits.
\end{definition}

Another type of confidential channel we consider is obtained by
removing Eve's capability to modify Alice's messages, while giving her
the ability to \emph{inject} any system (instead of only systems in
the fully mixed state).
% With no power to modify the input of Alice in channel $\nmqc$, Eve can only  inject a message, or blindly control the delivery of the messages input by Alice. Basically, Eve's capabilities now are reduced to either inject, block or reorder Alice's messages.

\begin{definition}[Non-Malleable Confidential Quantum Channel $\nmqc^{\ell,m}$]
	The channel keeps registers $\rA_1, \rA_2, \dots, \rA_\ell$, initially set to $\bot$.
    \begin{itemize}
        \item {\bf Interface $\iA$:} Upon receiving the $i$-th input
          in some state $\rho$, this system is stored in register $\rA_i$, and
          $\lNewMsg$ is output at interface $\iE$.
        \item {\bf Interface $\iE$:} 
        	\begin{itemize}
        		\item On receiving an input system in some
                          state $\state'$, perform an identity map and
                          output the same system at interface $\iB$.
        		\item On input index $j\in[\ell]$, output the
                          system in state $\state^{\rA_j}$ held in
                          register $\rA_j$ at interface $\iB$. After
                          the output, the state of register $\rA_j$
                          becomes $\bot$.
			\end{itemize}        	 
    \end{itemize}
    Interface $\iA$ and $\iE$ will receive at most $\ell$ inputs and
    ignore the rest. The quantum systems input at interface $\iA$ and
    output at interface $\iB$ always have length $m$ in qubits.
\end{definition}

% Note that in channel $\nmqc$ Eve cannot modify the message but can do injection, while in channel $\xqc$ Eve can only modify via bit flip or phase flip but cannot do injection. Comparing to channel $\cqc$, Eve are limited in different ways in $\nmqc$ and in $\xqc$. The ability of Eve is not comparable in these two channels. 

The next property to consider is authenticity: recall that in the
quantum setting, authenticity implies confidentiality, thus it does
not make sense to consider a ``non-confidential authentic channel'',
since a state cannot be cloned to be given to both Bob and Eve. An
authentic channel will automatically also be a confidential
one \cite{barnum2002qauth}. Therefore, as a next channel we directly consider the
\emph{secure} one \--- by secure we mean both authentic and
confidential. Eve only knows a new message has arrived but cannot
read, modify, nor inject messages. Eve still has the power to block and
reorder Alice's message.

\begin{definition}[Secure Quantum Channel $\sqc^{\ell,m}$]
	The channel keeps registers $\rA_1, \rA_2, \dots, \rA_\ell$, initially set to $\bot$.
    \begin{itemize}
        \item {\bf Interface $\iA$:} Upon receiving the $i$-th input
          in some state $\rho$, this system is stored in register $\rA_i$, and
          $\lNewMsg$ is output at interface $\iE$.
        \item {\bf Interface $\iE$:} On input index $j\in[\ell]$,
          output the system in state $\state^{\rA_j}$ held in register
          $\rA_j$ at interface $\iB$. After the output, the state in
          register $\rA_j$ becomes $\bot$.
    \end{itemize}
    Interface $\iA$ and $\iE$ will receive at most $\ell$ inputs and
    ignore the rest. The quantum systems input at interface $\iA$ and
    output at interface $\iB$ always have length $m$ in qubits.
\end{definition}

Finally, we consider an even stronger version of the secure channel which preserves the \emph{order} of the transmitted messages.
In particular, the adversary now only retains the power to delete messages, but cannot change the order in which they are transmitted.
This is enforced by replacing the capability to input indices by the ability of only inputting either $\lSend$ or $\lSkip$. %Note that also for this channel repetition of messages by Eve is not possible.

\begin{definition}[Ordered Secure Quantum Channel $\osqc^{\ell,m}$]
	The channel keeps registers $\rA_1, \rA_2, \dots, \rA_\ell$, initially set to $\bot$.
    \begin{itemize}
    \item {\bf Interface $\iA$:} Upon receiving the $i$-th input
          in some state $\rho$, this system is stored in register $\rA_i$, and
          $\lNewMsg$ is output at interface $\iE$.
    \item {\bf Interface $\iE$:} On $i$-th input $\lSend$ or $\lSkip$:
      If the input is $\lSend$, output the system in state
      $\state^{\rA_i}$ held in register $\rA_i$ at interface $\iB$. If
      the input is $\lSkip$, then output $\bot$ at interface
      $\iB$. After the output, the state in register $\rA_i$ becomes
      $\bot$.
    \end{itemize}
    Interface $\iA$ and $\iE$ will receive at most $\ell$ inputs and ignore the rest. The quantum systems input at interface $\iA$ and output at interface $\iB$ always have length $m$ in qubits.
\end{definition}

\subsection{Constructing an Ordered Secure Quantum Channel}
\label{sec:constructions.osc}

As stated in \autoref{lem:qauth} from \autoref{sec:example.it}, there
is a construction of one time secure quantum channel from one time
insecure quantum channel resource and a uniform key resource within
$\epsilon^\qauth$, i.e.
\[\left[\iqc^{1,n}, \key^{\mu}, \qc^{1,m,n}_A,\qc^{1,m,n}_B\right] 
  \xrightarrow{\pi^\qauth_{AB},\epsilon^\qauth}
  \left[\sqc^{1,m},\qc^{2,m,n}_E\right].\] Here, $\iqc$, $\sqc$ and
$\key$ are channel and key resources, as defined above. $\qc_{A/B/E}$
denote a resource that does quantum computation for Alice, Bob or Eve,
and allows them to perform encryption and decryption operations (we
informally refer to such resources as \emph{quantum computers} in the
following). These appear in the construction statement since for
finite security one makes all computational operations explicit \---
see \autoref{sec:ac} for more details.

We denote the encoding and decoding CPTP maps in this construction by
$\enc^\qauth:\Key\x\cL(\Hi_A) \rightarrow \cL(\Hi_C)$ and
$\dec^\qauth:\Key\x\cL(\Hi_{\tilde{C}}) \rightarrow \cL(\Hi_B \oplus
\proj{\bot})$.  We also denote by $\mathcal{E}$ the CPTP map that
always discards the state and replaces it with error state
$\proj{\bot}$. In this section, we build on top of these encoding and
decoding maps to construct a multi-message ordered secure quantum
channel from a multi-message insecure quantum channel, with a shared
uniform random function resource $\urf^{\log\ell,\mu}$.  The real
system is drawn in \autoref{fig:ic.osc.real} and the components are
described in \autoref{fig:ic.osc.protocol}. The ideal system one
wishes to build is depicted in \autoref{fig:ic.osc.ideal}.

\begin{figure}[tb]
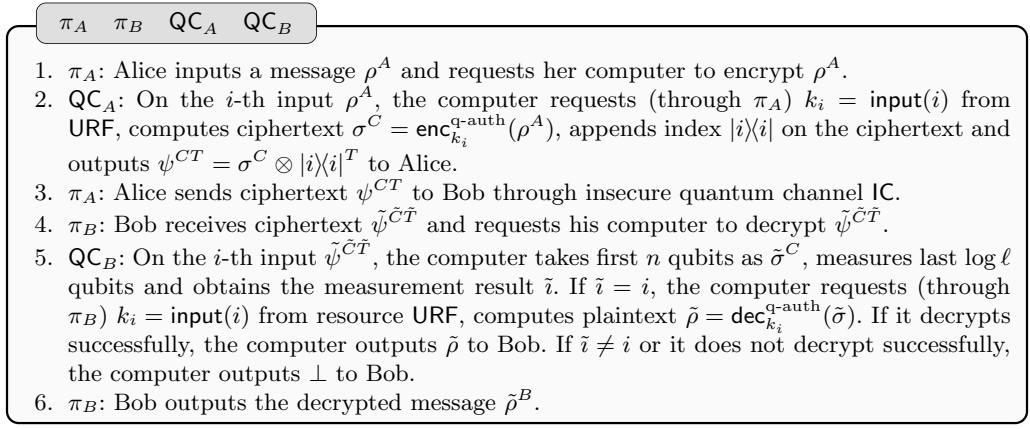

    \centering
    \begin{convbox}{\textwidth}{$\pi_A \quad \pi_B \quad \qc_A \quad \qc_B$}
		\begin{enumerate}
			\item $\pi_A$: Alice inputs a message $\state^{\rA}$ and requests her computer to encrypt $\state^{\rA}$.
			\item $\qc_A$: On the $i$-th input $\state^{\rA}$, the computer requests (through $\pi_A$) $k_i = \lInput (i)$ from $\urf$, computes ciphertext $\sigma^C = \enc^\qauth_{k_i}(\state^{A})$, appends index $\proj{i}$ on the ciphertext and outputs $\psi^{CT}= \sigma^C \otimes \proj{i}^T$ to Alice.
			\item $\pi_A$: Alice sends ciphertext $\psi^{CT}$ to Bob through insecure quantum channel $\iqc$.
			\item $\pi_B$: Bob receives ciphertext $\tilde{\psi}^{\tilde{C}\tilde{T}}$ and requests his computer to decrypt $\tilde{\psi}^{\tilde{C}\tilde{T}}$.
			\item $\qc_B$: On the $i$-th input
                          $\tilde{\psi}^{\tilde{C}\tilde{T}}$, the computer takes first
                          $n$ qubits as  $\tilde{\sigma}^C$, measures
                          last $\log\ell$ qubits and obtains the
                          measurement result $\tilde{\imath}$. If
                          $\tilde{\imath} = i$, the computer requests (through $\pi_B$)
                          $k_i = \lInput (i)$ from resource $\urf$,
                          computes plaintext $\tilde{\state} =
                          \dec^\qauth_{k_i}(\tilde{\sigma})$. 
                          If it decrypts successfully,
                          the computer outputs $\tilde{\state}$ to
                          Bob. If $\tilde{\imath} \neq i$ or it
                          does not decrypt successfully, the computer outputs $\bot$ to Bob.
          \item $\pi_B$: Bob outputs the decrypted message $\tilde{\rho}^B$.
		\end{enumerate}
  
    \end{convbox}
    \caption{Converters and computing resources to construct
      $\osqc^{\ell,m}$ from $\iqc^{\ell,n+\log\ell}$.  $\qc^{\ell, m, n+\log\ell}_A$ and $\qc^{\ell, m, n+\log\ell}_B$ will be queried $\ell$ times. The plaintext has length $m$ and the ciphertext has length $n + \log \ell$. $\urf^{\log\ell, \mu}$ has input length $\log\ell$ and output length $\mu$.} 
    \label{fig:ic.osc.protocol}
  \end{figure}

  \begin{figure}[tb]

    \begin{centering}

      \begin{tikzpicture}[
resourceLong/.style={draw,thick,minimum width=3cm,minimum height=1cm},
resource/.style={draw,thick,minimum width=1cm,minimum height=1cm},
sArrow/.style={->,>=stealth,thick},
sLine/.style={-,thick},
protocol/.style={draw,thick,minimum width=1cm,minimum height=4cm,rounded corners},
pnode/.style={minimum width=1cm,minimum height=1cm}]

\small

\def\t{4.9} % 1.75+.5+1.6+1.05
\def\a{3.05} % 1.75+.5+1.6/2
\def\v{1.5}
\def\w{.45}
\def\x{1}
\def\z{3.05} %1/2+.5+1+1.05

\node[resourceLong] (channel) at (0,-\v) {};
\node at (channel) {\footnotesize $\iqc$};
\node[resource] (qa) at (-\x,0) {};
\node at (qa) {\footnotesize $\qc_A$};
\node[resource] (qb) at (\x,0) {};
\node at (qb) {\footnotesize $\qc_B$};
\node[resourceLong] (key) at (0,\v) {};
\node at (key) {\footnotesize $\urf$};
\node[protocol] (protA) at (-\a,0) {};
\node[below right] at (protA.north west) {\footnotesize $\pi_A$};
\node[pnode] (a1) at (-\a,\v) {};
\node[pnode] (a2) at (-\a,0) {};
\node[pnode] (a3) at (-\a,-\v) {};
\node[protocol] (protB) at (\a,0) {};
\node[below left] at (protB.north east) {\footnotesize $\pi_B$};
\node[pnode] (b1) at (\a,\v) {};
\node[pnode] (b2) at (\a,0) {};
\node[pnode] (b3) at (\a,-\v) {};

%\node (aliceUp) at (-\t,\v) {};
\node (aliceMiddle) at (-\t,0) {};
%\node (aliceDown) at (-\t,-\v) {};
%\node (bobUp) at (\t,\v) {};
\node (bobMiddle) at (\t,0) {};
%\node (bobDown) at (\t,-\v) {};
\node (eveLeft) at (-\w,-\z) {};
\node (eveRight) at (\w,-\z) {};

\draw[sArrow] (aliceMiddle) to node[auto,pos=.3] {$\rho^{\rA}$} (a2);
\draw[sArrow] (a3) to (eveLeft |- a3) to node[auto,swap,pos=.8,xshift=2] {$\psi^{CT}$} (eveLeft);
\draw[sArrow] (eveRight) to node[auto,swap,pos=.2,xshift=-2] {$\tilde{\psi}^{\tilde{C}\tilde{T}}$} (eveRight |- b3) to (b3);
\draw[sArrow] (b2) to node[auto,pos=.9] {$\tilde{\rho}^{\rB}/\bot$} (bobMiddle);

%\node[draw] (key) at (0,\v) {key};
\node[pnode] (kLeft) at (-\x,\v) {};
\node[pnode] (kRight) at (\x,\v) {};
\draw[sArrow,bend left=15] (kLeft) to node[auto,pos=.5] {$k_i$} (a1);
\draw[sArrow,bend right=15] (kRight) to node[auto,swap, pos=.5] {$k_i$} (b1);
\draw[sArrow,bend left=25] (a1) to node[auto,pos=.5] {\footnotesize $i$} (kLeft);
\draw[sArrow,bend right=25] (b1) to node[auto,swap,pos=.5] {\footnotesize $i$} (kRight);

\node[pnode] (qLeft) at (-\x,0) {};
\node[pnode] (qRight) at (\x,0) {};
\draw[sArrow,bend left=25] (a2) to node[auto,pos=.5, yshift=-3]{$\rho^{\rA}, k_i$} (qLeft);
\draw[sArrow,bend left=25] (qLeft) to node[auto,pos=.5]{$\psi^{CT}$} (a2);
\draw[sArrow,bend right=25] (b2) to  node[auto, swap, pos=.5,yshift=-3] {$\tilde{\psi}^{\tilde{C}\tilde{T}},k_i$}  (qRight);
\draw[sArrow,bend right=25] (qRight) to node[auto, swap, pos=.5] {$\tilde{\rho}^{\rB}/\bot$} (b2);

\end{tikzpicture}

\end{centering}
\caption[Constructing $\osqc$ from $\iqc$: Real System]{\label{fig:ic.osc.real} 
The real system consisting of the shared resources $\iqc^{\ell,n+\log
\ell}$ and $\urf^{\log \ell,\mu}$, Alice and Bob's computing resources
$\qc^{\ell,m,n+\log \ell}_A $$\qc^{\ell,m,n+\log \ell}_B$, and the
protocol converters $\pi_A$ and $\pi_B$.}
\end{figure}
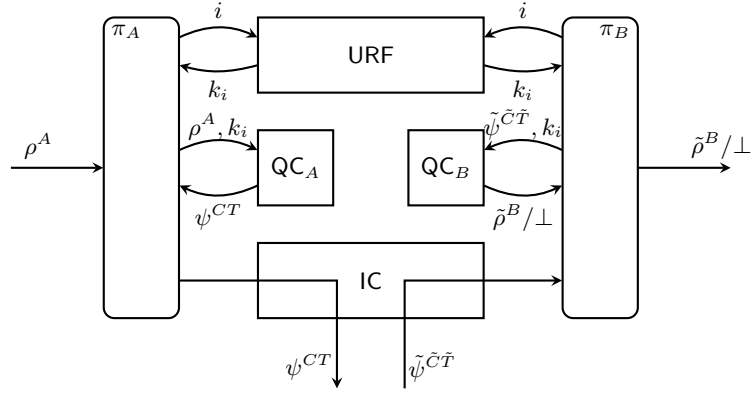

  \begin{figure}[tb]

\begin{centering}

\begin{tikzpicture}[
resource/.style={draw,thick,minimum width=3.5cm,minimum height=1.6cm},
sresource/.style={draw,thick,minimum width=0.4cm,minimum height=0.4cm},
rnode/.style={minimum width=.8cm,minimum height=1cm},
sArrow/.style={->,>=stealth,thick},
sLine/.style={-,thick},
sLeak/.style={->,>=stealth,thick,dashed},
simulator/.style={draw,thick,minimum width=8.5cm,minimum height=1cm,rounded corners}]

\small

\def\t{2.6} %1.6+1
\def\v{.3}
\def\w{1.9} %.8+.6+.5
\def\s{.4}
\def\x{5.5}
\def\z{2.9} %.8+.5+1+.6

\node[simulator] (sim) at (\x/2,-\w) {};
\node[xshift=-1.5,below right] at (sim.south west) {\footnotesize $\simul_E$};

\node[resource] (ch) at (0,0) {};
\node[yshift=-1.5,above right] at (ch.north west) {\footnotesize $\osqc$};
\node (alice) at (-\t,0) {};
\node (bob) at (\t,0) {};

\node[resource] (qc) at (\x,0) {}; %{$\epr^{MR}$};
\node[yshift=-1.5,above right] at (qc.north west) {\footnotesize $\qc_E$};
\node[rnode] (qc1) at (\x-.9,-.3) {};
\node[rnode] (qc2) at (\x+.9,-.3) {};
\node[rnode] (qc3) at (\x-1,-.3) {};
\node[rnode] (qc4) at (\x+1,-.3) {};

\node[rnode] (sim1) at (\x-.9,-\w) {};
\node[rnode] (sim2) at (\x+.9,-\w) {};
\node[rnode] (sim3) at (\x-1,-\w) {};
\node[rnode] (sim4) at (\x+1,-\w) {};

\def\h{-.2}
\node at (\x, \h+0.5) {$\qc_i^\qauth$} ;
\node[sresource] (qcauth1) at (\x-1,\h) {};
\node[sresource] (qcauth2) at (\x-0.5,\h) {};
\node[sresource] (qcauth3) at (\x,\h) {};
\node at (\x+0.5, \h) {$\ldots$};
\node[sresource] (qcauth4) at (\x+1,\h) {};

\node (up) at (0,\v) {};
\node (down) at (0,-\v) {};

\node (leakPoint) at (-\s,0) {};
\node (eveLeft) at (\x/2-\s,-\z) {};
\node (eveRight) at (\x/2+\s,-\z) {};

\node at (eveLeft |- sim) {};

\draw[sLine] (alice) to node[auto,pos=.15] {$\rho^{\rA_i}$} (0,0) to node (handle) {} (330:2*\s);
\draw[sArrow] (2*\s,0) to node[auto,pos=.85] {$\rho^{\rA_i}/\bot$} (bob);

\draw[sArrow] (eveLeft |- sim.south) to node[auto,pos=.5,xshift=-1] {$\psi^C$} (eveLeft.center);
\draw[sArrow] (eveRight.center) to node[auto,pos=.5,xshift=-1,swap] {$\tilde{\psi}^C$} (eveRight |- sim.south);

\draw[sLeak] (leakPoint.center) to node[auto, swap, pos=.85]{$\lNewMsg$} (leakPoint |- sim.north); 
\draw[sArrow] (handle.center |- sim.north) to node[auto,swap,pos=.15] {$\lSkip/\lSend$} (handle.center);
\draw[sArrow, bend left=15] (sim2) to node[auto, pos=.45]{$\tilde{\psi}^C$} (qc2);
\draw[sArrow, bend left=15] (qc1) to node[auto, pos=.55]{${\psi}^C$} (sim1);
\draw[sArrow, bend left=15] (sim3) to node[auto, pos=.45]{$\lNewMsg$} (qc3);
\draw[sArrow, bend left=15] (qc4) to node[auto, pos=.55]{$\lSkip/\lSend$} (sim4);

\end{tikzpicture}

\end{centering}
\caption[Constructing $\osqc$ from $\iqc$: Ideal
System]{\label{fig:ic.osc.ideal} The ideal system consisting of
  $\osqc^{\ell,m}, \qc_E^{2\ell,m,n+\log\ell}$ and
  $\simul_E$. $\qc_E^{2\ell,m,n+\log\ell}$ makes use of $\ell$
  instances of $\qc^\qauth$ internally, while $\simul_E$ only receives
  and forwards messages and does no computation.}
\end{figure}

\begin{theorem} \label{thm:ic.osqc} Let
  $\pi_{AB}=(\pi_A, \pi_B), \qc^{\ell, m, n+\log\ell}_A
  ,\qc^{\ell, m, n+\log\ell}_B$ and $\urf^{\log\ell,\mu}$denote converters and computing
  resources as described in  \autoref{fig:ic.osc.protocol},  corresponding to Alice and Bob both applying the following CPTP maps with increasing index $i$:  
$$\Lambda_i^{A \rightarrow CT}(\cdot) = \enc^\qauth_{k_i}(\cdot) \otimes \proj{i}^T$$
$$\Lambda_i^{\tilde{C}\tilde{T} \rightarrow B}(\cdot) = \dec^{\qauth}_{k_i} \lp (I^{\tilde{C}} \otimes \bra{i}^{\tilde{T}})(\cdot)(I^{\tilde{C}} \otimes \ket{i}^{\tilde{T}}) \rp
+ \mathcal{E}\lp \bar P_i\reg {\tilde{T}} (\cdot) \bar P_i\reg {\tilde{T}}) \rp,$$
    where $\bar P_i = I-\proj{i}$, and $k_i$ is the output of $\urf^{\log\ell,\mu}$ with input $i$. Let $\qc_E^{2\ell, m, n+\log\ell}$ be
  the computing resource of Eve capable of doing $\ell$ encryption
  operations and $\ell$ decryption operations.  Let
  $\epsilon^\qauth$ be the upper bound on the distinguishing advantage of the one
  time secure quantum channel construction. Then, 
  %$\pi_{AB}$ constructs an ordered secure quantum
  %channel $\osqc^{\ell, m}$ from a insecure quantum channel resource
  %$\iqc^{\ell,n+\log\ell}$, a key resource $\key^{\mu\ell}$ within
  %$\ell\epsilon^\qauth$.  i.e.,
\begin{multline*}
\left[\iqc^{\ell,n+\log\ell}, \urf^{\log\ell,\mu}, \qc^{\ell,m,n+log\ell}_A, \qc^{\ell,m,n+log\ell}_B \right] \\
  \xrightarrow{\pi_{AB},\ell\epsilon^{\qauth}}
  \left[\osqc^{\ell, m}, \qc^{2\ell,m,n+\log\ell}_E \right].
\end{multline*}
\end{theorem}

\begin{proof}
    In order to prove the statement, we need to find a simulator $\simul_E$ and describe what Eve's quantum computer $\qc_E^{2\ell, m,n+\log\ell}$ does such that real system, as illustrated in \autoref{fig:ic.osc.real} and the ideal system, as illustrated in \autoref{fig:ic.osc.ideal} are indistinguishable within $\ell\epsilon^\qauth$, where $\epsilon^\qauth$ is the constant upper bound on the error
    of the one message authentication scheme.

    The simulator $\simul_E$ works as following. On input $\lNewMsg$ from the ideal resource $\osqc^{\ell,m}$, it requests $\qc_E^{2\ell, m, n+\log\ell}$ to provide a quantum state $\psi$ and outputs $\psi$ at interface $\iE$.  On input $\tilde{\psi}$ at interface $\iE$, it requests $\qc_E^{2\ell, m, n+\log\ell}$ to provide $\lSend$ or $\lSkip$ signal.  The signal is then forwarded to ideal channel $\osqc^{\ell,m}$.
    
    Eve's computing resource $\qc_E^{2\ell,m,n+\log\ell}$ will use $\qc_E^{2,m,n}$ of the one time secure quantum channel construction as a subroutine. We refer to this subroutine computer resource as $\qc^\qauth$. $\qc_E^{2\ell,m,n+\log\ell}$ works as following,
\begin{itemize}
\item On the $i$-th input $\lNewMsg$, the computer instantiates a new $\qc^\qauth_i$, sends $\lNewMsg$ to $\qc^\qauth_i$.  After getting the ciphertext $\sigma$ from $\qc_i^\qauth$, the computer appends index $i$ to $\sigma$ and output $\psi = \sigma \otimes \proj{i}$ to $\simul_E$. After output $\psi$, the computer updates internal counter $\vNextIndex=i$. 
\item On input $\tilde{\psi}$, the computer takes the first $n$ qubits as ciphertext $\tilde{\sigma}$, measures the last $\log\ell$ qubits and obtains the measurement result $\tilde{\imath}$. If $\vNextIndex \neq \tilde{\imath}$, the computer outputs $\lSkip$ to $\simul_E$. If $\vNextIndex=\tilde{\imath}$, the computer forwards ciphertext $\tilde{\sigma}$ to $\qc^\qauth_{\tilde{\imath}}$. If $\qc_{\tilde{\imath}}^\qauth$ outputs $\lAcc$, the computer outputs $\lSend$ to $\simul_E$. If $\qc_{\tilde{\imath}}^\qauth$ outputs $\lRej$, the computer outputs $\lSkip$ to $\simul_E$.
\end{itemize}    
We define $\sys R^{\ell} \df \pi_{AB}\left[\iqc^{\ell,n+\log\ell},\urf^{\log\ell, \mu}, \qc^{\ell,m,n+\log\ell}_A, \qc^{\ell,m,n+\log\ell}_B\right]$ and 
    $\sys S^{\ell} \df \simul_E\left[\osqc^{\ell,m}, \qc_{E}^{2\ell,m,n+\log\ell}\right]$.  Note that we can consider the channel of $\ell$ messages as $\ell$ different 1 message channels with an obvious converter $\alpha$ that directs the message according to its tag into the corresponding channel, i.e., $\iqc^{\ell,n+\log\ell} \equiv \alpha\left[\iqc^{1,n}_1, \iqc^{1,n}_2, \ldots, \iqc^{1,n}_\ell\right]$. Similarly $\urf^{\log\ell, \mu}$ can be consider as $\ell$ different $\key^\mu$ resources and channel $\osqc^{\ell,m}$ and  quantum computers $\qc^{\ell,m,n+\log\ell}_A$, $\qc^{\ell,m,n+\log\ell}_B$, $\qc_{E}^{2\ell,m,n+\log\ell}$ can all be considered as $\ell$ different 1 message resources with the same converter $\alpha$. We further define $\sys R_i \df \pi_{AB}^{\qauth}\left[\iqc^{1,n}_i,\key^\mu_i, \qc^{1,m,n}_{A_i}, \qc^{1,m,n}_{B_i}\right]$ and $\sys S_i = \simul^\qauth\left[\osqc^{1,m}_i, \qc_{E_i}^{2,m,n}\right]$. We can  see that $\sys R^\ell \equiv \alpha\left[\sys R_1, \sys R_2, \ldots, \sys R_\ell\right]$ and $\sys S^\ell \equiv \alpha\left[\sys S_1, \sys S_2, \ldots, \sys S_\ell\right]$ since for each message the protocol
    $\pi_{AB}$ and $\simul_E$ works
    exactly the same as $\pi_{AB}^\qauth$ and $\simul^\qauth$ in the one time secure quantum channel
    construction. For any distinguisher $\sys D$,
    \begin{align*}%\label{eq:1tol.reduction}
     \Delta^{\sys D}\left(\sys R^\ell, \sys S^\ell\right) 
	& =      \Delta^{\sys D\alpha}\left( \left[\sys R_1, \sys R_2, \ldots, \sys R_\ell\right], \left[\sys S_1, \sys S_2, \ldots, \sys S_\ell\right]\right) \nonumber \\
	& \leq \sum_{i=1}^\ell \Delta^{\sys D\alpha} \left( \left[\sys S_1, \ldots, \sys S_{i-1}, \sys R_i, \sys R_{i+1}, \ldots, \sys R_\ell\right], \right. \nonumber \\
	& \qquad \qquad \qquad \qquad \qquad \left. \left[\sys S_1, \ldots, \sys S_{i-1}, \sys S_i, \sys R_{i+1}, \ldots, \sys R_\ell\right]\right)\nonumber \\
    & \leq \sum_{i=1}^\ell \Delta^{\sys {D\alpha C}_i}\left(\sys R_i, \sys S_i\right)  \nonumber \\
    & \leq  \sum^\ell_{i = 1} \epsilon^\qauth(\sys {D\alpha C}_i) \leq \ell\epsilon^{\qauth}.
    \end{align*}

    This analysis follows from
    a standard hybrid argument.  For a distinguisher $\sys D$ to distinguish in the
    multi-message case,  the new distinguishers $\sys {D\alpha C_i}$,  attached by a converter $\alpha$ and $C_i = \left[\sys S_1, \ldots, \sys S_{i-1}, \cdot , \sys R_{i+1}, \ldots, \sys R_\ell\right]$,  has to distinguish in one of the one
    message cases. The final inequality follows from the
    information\-/theoretic bound in the one message case, where we
    have used $\epsilon^{\qauth}$ both for the function and its upper bound.
        % This concludes the proof.
\end{proof}

\begin{remark*}
  \autoref{thm:ic.osqc} is meaningful only if the protocol is also
  \emph{correct}, i.e., if the distinguisher always puts back the same
  ciphertext on the insecure channel in the right order, then Bob
  always successfully decrypts. This follows trivially from the
  correctness of the underlying quantum authentication protocol, so we
  omit a formal discussion of it.
\end{remark*}

Suppose now that one has a $\prf$ resource and a bound $\epsilon^\prf$ satisfying \eqnref{eq:prf2urf},
that is, indistinguishable from $\urf$ within $\epsilon^\prf$, the following corollary follows trivially from the composition theorem \autoref{thm:composition} in \aref{app:composition}.
\begin{corollary}
\begin{multline*}
\left[\iqc^{\ell,n+\log\ell}, \key^r, \prf_A^{r,\log\ell,\mu}, \prf_B^{r,\log\ell,\mu}, \qc^{\ell,m,n+log\ell}_A, \qc^{\ell,m,n+log\ell}_B \right] \\
  \xrightarrow{\pi'_{AB},\epsilon}
  \left[\osqc^{\ell, m}, \qc^{2\ell,m,n+\log\ell}_E \right],
\end{multline*}
where $\pi'_{AB} =(\pi_{AB}, \pi^{\prf}), \epsilon(\sys D) = \epsilon^\prf(\sys D \sys C)+\ell\epsilon^{\qauth}$ and $\sys C$ is the system including $\pi_{AB}$, $\iqc^{\ell,n+\log\ell}$, $\qc^{\ell,m,n+log\ell}_A$, $\qc^{\ell,m,n+log\ell}_B$.

\end{corollary}
\subsection{Constructing a Pauli-Malleable Confidential Quantum Channel}
\label{sec:constructions.xcc}

In this section, we construct a Pauli-malleable confidential quantum
channel $\xqc^{\ell,m}$ from an insecure quantum channel
$\iqc^{\ell,m+\nu}$.  In the Pauli-malleable confidential channel, the
adversary can only get a notification of a new message arriving but
has no access to the message. The adversary has the ability to block,
reorder and modify the message via Pauli operators (bit flip and phase flip),
as well as ask the channel to output a fully mixed state at Bob's
interface, as defined in \autoref{def:xqc}.

Now we present the protocol in the multi\-/message case, described in
\autoref{fig:ic.xqc.protocol}. % The protocol is similar to that of the
% one message case described in \autoref{fig:ic.xqc.one.protocol} in
% \aref{app:example.otp} with a uniform random function resource 
% $\urf^{\nu,2m}$. 
In the protocol, Alice's computer will generate a new random string
$x$ of length $\nu$ for each message different from previous random
strings and input it to $\urf^{\nu,2m}$, a key $k$ is returned by
$\urf^{\nu,2m}$ , the Pauli-operator $P_k$ is applied to the message
and $x$ is appended to the ciphertext.  Bob's computer will do the
measurement on the last $\nu$ qubits to get $\tilde{x}$, which is
input to $\urf^{\nu,2m}$, from which $\tilde{k}$ is obtained and
finally the Pauli operator $P_{\tilde{k}}$ is applied to the
ciphertext. The real system is drawn in \autoref{fig:ic.xqc.real}, and
the ideal system we construct is illustrated in
\autoref{fig:ic.xqc.ideal}.

\begin{figure}[tb]
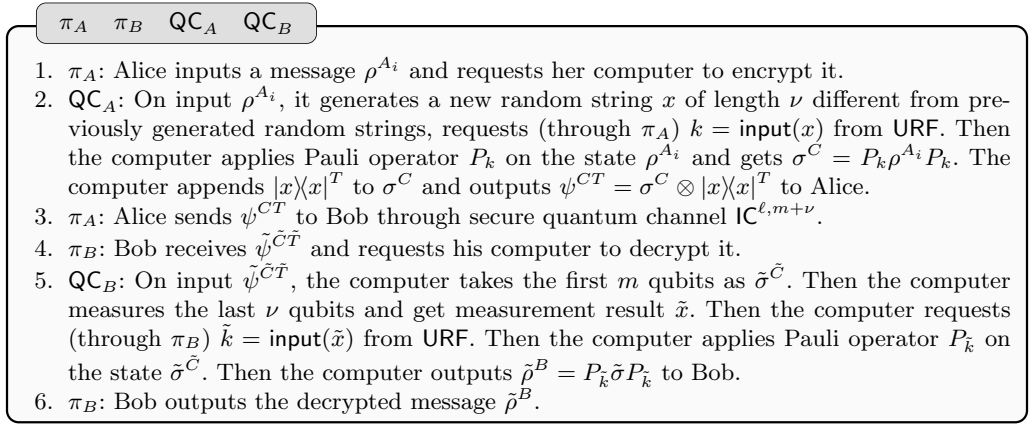

    \centering
    \begin{convbox}{\textwidth}{$\pi_A \quad \pi_B \quad \qc_A \quad \qc_B$}
		\begin{enumerate}
			\item $\pi_A$: Alice inputs a message $\state^{\rA_i}$ and requests her computer to encrypt it.
			\item $\qc_A$: On input
                          $\state^{\rA_i}$, it generates a new random string
                          $x$ of length $\nu$ different from
                          previously generated random strings, requests (through $\pi_A$) $k=\lInput(x)$ from $\urf$. Then the computer applies Pauli operator $P_k$ on the state $\state^{\rA_i}$  and gets $\sigma^C =P_k\state^{\rA_i} P_k $. The computer appends $\proj{x}^T$ to $\sigma^C$ and outputs $\psi^{CT}=\sigma^C  \otimes \proj{x}^T$ to Alice.
			\item $\pi_A$: Alice sends $\psi^{CT}$ to Bob through secure quantum channel $\iqc^{\ell,m+\nu}$.
			\item $\pi_B$: Bob receives $\tilde{\psi}^{\tilde{C}\tilde{T}}$ and requests his computer to decrypt it.
			\item $\qc_B$:  On input $\tilde{\psi}^{\tilde{C}\tilde{T}}$, the computer takes the first $m$ qubits as $\tilde{\sigma}^{\tilde{C}}$. Then the computer measures the last $\nu$ qubits and get measurement result $\tilde{x}$.  Then the computer requests (through $\pi_B$) $\tilde{k}=\lInput(\tilde{x})$ from $\urf$. Then the computer applies Pauli operator $P_{\tilde{k}}$ on the state $\tilde{\sigma}^{\tilde{C}}$.  Then the computer outputs $\tilde{\rho}^B =  P_{\tilde{k}} \tilde{\sigma} P_{\tilde{k}} $ to Bob.
			 \item $\pi_B$: Bob outputs the decrypted message $\tilde{\rho}^B$.
		\end{enumerate}
  
    \end{convbox}
    \caption{Converters and computer resources to construct $\xqc^{\ell,m}$ from $\iqc^{\ell,m+\nu}$.  $\qc^{\ell, m, m+\nu}_A$ and $\qc^{\ell, m, m+\nu}_B$ will be queried $\ell$ times. The plaintext has length $m$ and ciphertext has length $m+\nu$. $\urf^{\nu, 2m}$ has input length $\nu$ and output length $2m$.}
    \label{fig:ic.xqc.protocol}
\end{figure}

\begin{figure}[tb]
  \begin{centering}
\begin{tikzpicture}[
resourceLong/.style={draw,thick,minimum width=3cm,minimum height=1cm},
resource/.style={draw,thick,minimum width=1cm,minimum height=1cm},
sArrow/.style={->,>=stealth,thick},
sLine/.style={-,thick},
protocol/.style={draw,thick,minimum width=1cm,minimum height=4cm,rounded corners},
pnode/.style={minimum width=1cm,minimum height=1cm}]

\small

\def\t{4.9} % 1.75+.5+1.6+1.05
\def\a{3.05} % 1.75+.5+1.6/2
\def\v{1.5}
\def\w{.45}
\def\x{1}
\def\z{3.05} %1/2+.5+1+1.05

\node[resourceLong] (channel) at (0,-\v) {};
\node at (channel) {\footnotesize $\iqc$};
\node[resource] (qa) at (-\x,0) {};
\node at (qa) {\footnotesize $\qc_A$};
\node[resource] (qb) at (\x,0) {};
\node at (qb) {\footnotesize $\qc_B$};
\node[resourceLong] (key) at (0,\v) {};
\node at (key) {\footnotesize $\urf$};
\node[protocol] (protA) at (-\a,0) {};
\node[below right] at (protA.north west) {\footnotesize $\pi_A$};
\node[pnode] (a1) at (-\a,\v) {};
\node[pnode] (a2) at (-\a,0) {};
\node[pnode] (a3) at (-\a,-\v) {};
\node[protocol] (protB) at (\a,0) {};
\node[below left] at (protB.north east) {\footnotesize $\pi_B$};
\node[pnode] (b1) at (\a,\v) {};
\node[pnode] (b2) at (\a,0) {};
\node[pnode] (b3) at (\a,-\v) {};

%\node (aliceUp) at (-\t,\v) {};
\node (aliceMiddle) at (-\t,0) {};
%\node (aliceDown) at (-\t,-\v) {};
%\node (bobUp) at (\t,\v) {};
\node (bobMiddle) at (\t,0) {};
%\node (bobDown) at (\t,-\v) {};
\node (eveLeft) at (-\w,-\z) {};
\node (eveRight) at (\w,-\z) {};

\draw[sArrow] (aliceMiddle) to node[auto,pos=.3] {$\rho^{\rA_i}$} (a2);
\draw[sArrow] (a3) to (eveLeft |- a3) to node[auto,swap,pos=.8,xshift=2] {$\psi^{CT}$} (eveLeft);
\draw[sArrow] (eveRight) to node[auto,swap,pos=.2,xshift=-2] {$\tilde{\psi}^{\tilde{C}\tilde{T}}$} (eveRight |- b3) to (b3);
\draw[sArrow] (b2) to node[auto,pos=.9] {$\tilde{\rho}^{\rB}$} (bobMiddle);

%\node[draw] (key) at (0,\v) {key};
\node[pnode] (kLeft) at (-\x,\v) {};
\node[pnode] (kRight) at (\x,\v) {};
\draw[sArrow,bend left=15] (kLeft) to node[auto,pos=.5] {$k$} (a1);
\draw[sArrow,bend right=15] (kRight) to node[auto,swap, pos=.5] {$\tilde{k}$} (b1);
\draw[sArrow,bend left=25] (a1) to node[auto,pos=.5] {\footnotesize $x$} (kLeft);
\draw[sArrow,bend right=25] (b1) to node[auto,swap,pos=.5] {\footnotesize $\tilde{x}$} (kRight);

\node[pnode] (qLeft) at (-\x,0) {};
\node[pnode] (qRight) at (\x,0) {};
\draw[sArrow,bend left=25] (a2) to node[auto,pos=.5, yshift=-3]{$\rho^{\rA_i}, k$} (qLeft);
\draw[sArrow,bend left=25] (qLeft) to node[auto,pos=.5]{$\psi^{CT}$} (a2);
\draw[sArrow,bend right=25] (b2) to  node[auto, swap, pos=.5,yshift=-3] {$\tilde{\psi}^{\tilde{C}\tilde{T}},\tilde{k}$}  (qRight);
\draw[sArrow,bend right=25] (qRight) to node[auto, swap, pos=.5] {$\tilde{\rho}^{\rB}$} (b2);

\end{tikzpicture}

\end{centering}
    \caption{The real system consisting of shared resources
      $\iqc^{\ell,m+\nu}$ and $\urf^{\nu,2m}$, Alice and
      Bob's computing resources $\qc^{\ell, m, m+\nu}_A$ and
      $\qc^{\ell, m, m+\nu}_B$, and the protocol converters $\pi_A$
      and $\pi_B$.}
    \label{fig:ic.xqc.real}
  \end{figure}
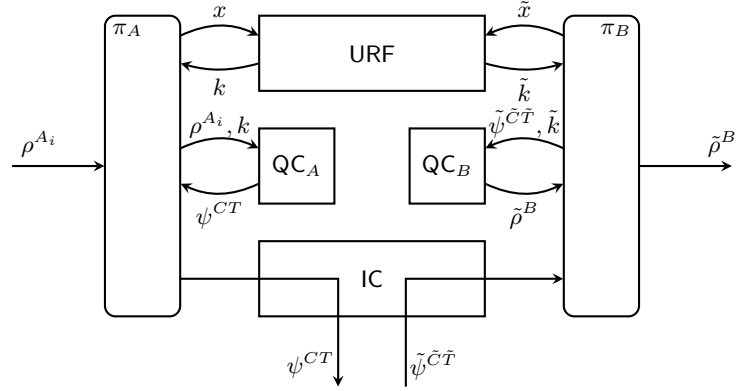

  \begin{figure}[tb]
\begin{centering}

\begin{tikzpicture}[
resource/.style={draw,thick,minimum width=3.5cm,minimum height=1.6cm},
sresource/.style={draw,thick,minimum width=0.4cm,minimum height=0.4cm},
rnode/.style={minimum width=.8cm,minimum height=1cm},
sArrow/.style={->,>=stealth,thick},
sLine/.style={-,thick},
sLeak/.style={->,>=stealth,thick,dashed},
simulator/.style={draw,thick,minimum width=8.5cm,minimum height=1cm,rounded corners}]

\small

\def\t{2.6} %1.6+1
\def\v{.3}
\def\w{1.9} %.8+.6+.5
\def\s{.4}
\def\x{5.5}
\def\z{2.9} %.8+.5+1+.6

\node[simulator] (sim) at (\x/2,-\w) {};
\node[xshift=-1.5,below right] at (sim.south west) {\footnotesize $\simul_E$};

\node[resource] (ch) at (0,0) {};
\node[yshift=-1.5,above right] at (ch.north west) {\footnotesize $\xqc$};
\node (alice) at (-\t,0) {};
\node (bob) at (\t,0) {};

\node[resource] (qc) at (\x,0) {}; %{$\epr^{MR}$};
\node[yshift=-1.5,above right] at (qc.north west) {\footnotesize $\qc_E$};
\node[rnode] (qc1) at (\x-.9,-.3) {};
\node[rnode] (qc2) at (\x+.9,-.3) {};
\node[rnode] (qc3) at (\x-1,-.3) {};
\node[rnode] (qc4) at (\x+1,-.3) {};

\node[rnode] (sim1) at (\x-.9,-\w) {};
\node[rnode] (sim2) at (\x+.9,-\w) {};
\node[rnode] (sim3) at (\x-1,-\w) {};
\node[rnode] (sim4) at (\x+1,-\w) {};

\def\h{-.2}
\node at (\x, \h+0.5) {$\qc_i^\qconf$} ;
\node[sresource] (qcauth1) at (\x-1,\h) {};
\node[sresource] (qcauth2) at (\x-0.5,\h) {};
\node[sresource] (qcauth3) at (\x,\h) {};
\node at (\x+0.5, \h) {$\ldots$};
\node[sresource] (qcauth4) at (\x+1,\h) {};

\node (up) at (0,\v) {};
\node (down) at (0,-\v) {};

\node (leakPoint) at (-\s,0) {};
\node (eveLeft) at (\x/2-\s,-\z) {};
\node (eveRight) at (\x/2+\s,-\z) {};

\node at (eveLeft |- sim) {};

\node[sresource] (pauli) at (0.5,0) {P};
\draw[sLine] (alice) to node[auto,pos=.15] {$\rho^{\rA_i}$} (0,0) to  (pauli);
\draw[sArrow] (pauli |- sim.north) to node[auto,swap,pos=.15] {$(j, k)/\bot$} (pauli);
\draw[sArrow] (0.7,0) to node[auto,pos=1] {$\frac{1}{2^m}I_{2^m}$} (bob);
\node[above, yshift = 9] at (bob.north) {$P_k\rho^{\rA_j}P_k/$};

\draw[sArrow] (eveLeft |- sim.south) to node[auto,pos=.5,xshift=-1] {$\psi^C$} (eveLeft.center);
\draw[sArrow] (eveRight.center) to node[auto,pos=.5,xshift=-1,swap] {$\tilde{\psi}^C$} (eveRight |- sim.south);

\draw[sLeak] (leakPoint.center) to node[auto, swap, pos=.85]{$\lNewMsg$} (leakPoint |- sim.north); 
\draw[sArrow, bend left=15] (sim2) to node[auto, pos=.45]{$\tilde{\psi}^C$} (qc2);
\draw[sArrow, bend left=15] (qc1) to node[auto, pos=.55]{${\psi}^C$} (sim1);
\draw[sArrow, bend left=15] (sim3) to node[auto, pos=.45]{$\lNewMsg$} (qc3);
\draw[sArrow, bend left=15] (qc4) to node[auto, pos=.55]{$(j, k)/\bot$} (sim4);

\end{tikzpicture}

\end{centering}
\caption[Ideal system for constructing $\xqc$ from $\iqc$]{\label{fig:ic.xqc.ideal}
The ideal system consisting of $\xqc^{\ell,m}, \qc_E^{2\ell,m,m+\nu}$
and $\simul_E$. $\qc_E^{2\ell,m,m+\nu}$ makes use of $\ell$ instances of $\qc^\qconf$ internally while  $\simul_E$ only receives and forwards messages and does no computation.
}
\end{figure}

\begin{theorem} \label{thm:ic.xqc}
Let $\pi_{AB}=(\pi_A, \pi_B), \qc^{\ell, m, m+\nu}_A$ and $\qc^{\ell, m, m+\nu}_B$ denote converters and computing resources, described in \autoref{fig:ic.xqc.protocol}, corresponding to Alice and Bob applying the following CPTP maps,
\begin{align*}
  \Lambda^{A_i \rightarrow CT}(\cdot) & = \frac{1}{2^\nu}\sum_x
  P_{k_x}(\cdot)P_{k_x} \otimes \proj{x}^T \\
\Lambda^{\tilde{C}\tilde{T} \rightarrow B}(\cdot) & = \sum_x  (P_{k_x} \otimes \bra{x}^{\tilde{T}})(\cdot)(P_{k_x} \otimes \ket{x}^{\tilde{T}}) ,\end{align*}
    where $k_x$ is the output of $\urf^{\nu,2m}$ with input $x$. Let $\qc_E^{2\ell, m, m+\nu}$ be the computing resource of Eve capable of doing $\ell$ encryption operations and $\ell$ decryption operations. Then $\pi_{AB}$ constructs a Pauli-malleable confidential quantum channel $\xqc^{\ell, m}$ from an insecure quantum channel resource $\iqc^{\ell,m+\nu}$, a shared uniform random function resource $\urf^{\nu, 2m}$ within $\ell^2 \cdot 2^{-\nu} $, i.e., 
\begin{equation*}
\left[\iqc^{\ell,m+\nu}, \urf^{\nu,2m}, \qc^{\ell,m,m+\nu}_A, \qc^{\ell,m,m+\nu}_B \right] %\\
  \xrightarrow{\pi_{AB},\ell^2 2^{-\nu}}
  \left[\xqc^{\ell, m}, \qc^{2\ell,m,m+\nu}_E \right].
\end{equation*}
\end{theorem}

In \aref{app:example.otp} we analyze the one\-/message case, which is
then used in the following proof.

\begin{proof}
    We present the simulator $\simul_E$ and $\qc_E^{2\ell,m,m+\nu}$ in
    the multi-message case and prove that ${\sys R}^{\ell}$ cannot be
    distinguished  from $\sys S^{\ell}$ with significant advantage for
    any distinguisher $\sys D$, where we define the real system
    illustrated in \autoref{fig:ic.xqc.real} as ${\sys R}^{\ell} \df
    \pi_{AB}\left[\iqc^{\ell,m+\nu},\urf^{\nu,2m},
      \qc^{\ell,m,m+\nu}_A, \qc^{\ell,m,m+\nu}_B \right]$, and the
    ideal system drawn in \autoref{fig:ic.xqc.ideal} is $\sys S^{\ell} \df \simul_E\left[\xqc^{\ell,m}, \qc_E^{2\ell,m,m+\nu}\right]$.
    
    The simulator $\simul_E$ works exactly the same as the one message
    case analyzed in \aref{app:example.otp}. On input $\lNewMsg$ from the ideal resource $\xqc^{\ell,m}$, it requests $\qc_E^{2\ell, m, m+\nu}$ to provide a quantum state $\psi$ and outputs $\psi$ at interface $\iE$.  On input $\tilde{\psi}$ at interface $\iE$, it requests $\qc_E^{2\ell, m, m+\nu}$ to provide a tuple $(i,k)$ (or $\bot$).  The tuple (or $\bot$) is then forwarded to ideal channel $\xqc^{\ell,m}$.		
    Eve's computing resource $\qc_E^{2\ell,m, m+\nu}$ will use the computing resource $\qc^\qconf$ of the one message Pauli-malleable confidential quantum channel construction as a subroutine. $\qc_E^{2\ell,m,m+\nu}$ works as following. 
\begin{itemize}
\item On the $i$-th input $\lNewMsg$ from $\simul_E$, the computer instantiates a new $\qc^\qconf_i$ instance with index $i$ and sends $\lNewMsg$ to $\qc^\qconf_i$.  After getting the state $\sigma$ from $\qc_i^\qconf$, the computer generates a random $x$ of length $\nu$,  stores tuple $(i,x)$ internally, appends $x$ to state $\sigma$ and outputs $\psi = \sigma \otimes \proj{x}$ to $\simul_E$.
\item On input $\tilde{\psi}$ from interface $\iE$, the computer takes the first $m$ qubits of $\tilde{\psi}$ as state $\tilde{\sigma}$, measures the last $\nu$ qubits of $\tilde{\psi}$ and obtains the measurement result $\tilde{x}$. If $\tilde{x}$ is not stored internally, the computer outputs $\bot$ to $\simul_E$. If $\tilde{x}$ is store internally, the computer retrieves tuple $(\tilde{\imath}, \tilde{x})$,  then it sends state $\tilde{\sigma}$ to $\qc^\qconf_{\tilde{\imath}}$.  After getting output $\tilde{k}$ (or $\bot$) from $\qc^\qconf_{\tilde{\imath}}$, the computer sends tuple $(\tilde{\imath}, \tilde{k})$ (or $\bot$) to $\simul_E$.  
\end{itemize}    
    Let us consider $\overline{\sys R}^{\ell} \df  \pi_{AB}\left[\overline{\iqc}^{\ell,m+\nu},\urf^{\nu,2m}, \qc^{\ell,m,m+\nu}_A, \qc^{\ell,m,m+\nu}_B \right]$ and $\overline{\sys S}^{\ell} \df \simul_E\left[\overline{\xqc}^{\ell,m}, \qc_E^{2\ell,m,m+\nu}\right]$ where the channels with overline $\overline{\iqc}$ and $\overline{\xqc}$ will reject any input from interface $\iA$ if it contains the same tag $x$ as previous input from interface $\iE$. Therefore, $\overline{\sys R}^\ell$ and $\sys R^\ell$,  $\overline{\sys S}^\ell$ and $\sys S^\ell$ are equivalent to any distinguisher up to the point that the distinguisher successfully guesses the string $x$ that is randomly chosen by Alice's computer.   In this case, in the real world the distinguisher inputs $\sigma \otimes \proj{x}$ at interface $\iE$, state $\sigma$ is flipped by some Pauli operator and is output at interface $\iB$. The distinguisher then forwards the output from interface $\iB$ to interface $\iA$. Alice's computer happens to choose the same $x$ for this input, thus flips the state back with the same Pauli operator and outputs same message $\sigma \otimes \proj{x}$ at interface $\iE$. While in the ideal world,  a fully mixed state with an appended tag will be output at interface $\iE$.   %\todomarginJiamin{improve?}
    
    Let $x_1, x_2, \ldots, x_\ell$ be the random string of length $\nu$ randomly chosen by Alice,  the probability that distinguisher guesses one successfully in $\ell$ times is 
    $$1-(1- \frac{\ell}{2^\nu})^\ell \leq \frac{\ell^2}{2^\nu}.$$ 
    
    Therefore, for any distinguisher $\sys D$, we have     
    \begin{equation}
    d^{\sys D} \left({\sys R}^{\ell},  \sys S^{\ell}\right) 
    \leq d^{\sys D} \left({\overline{\sys R}}^{\ell},  {\overline{\sys S}}^{\ell}\right)  + \frac{\ell^2}{2^\nu}.   
    \end{equation}
    %\begin{align} \label{eq: xqc.many_one_reduction}
    %d^{\sys D} \left({\overline{\sys R}}^{\ell},  \overline{\sys S}^{\ell}\right) 
   %& \leq \sum_{i=1}^\ell d^{\sys {DC}_i} \left(\overline{\sys R_i^{1}},  \overline{\sys S_i^{1}}\right)  \nonumber \\
   %&= \ell \cdot 0 = 0 
   % \end{align}
   
    Following from a standard hybrid argument, if a distinguisher $\sys D$ can distinguish $\overline{\sys R}^\ell$ and $\overline{\sys S}^\ell$,
    it can be used to construct a distinguisher in the one message scenario, distinguishing $\overline{R}^1 = \pi_{AB}\left[\overline{\iqc}^{1,m},\key^{2m},	 \qc^{1,m,m}_A, \qc^{1,m,m}_B\right]$ from $\overline{S}^1  = \simul^\qconf\left[\overline{\xqc}^{1,m}, \qc^\qconf\right]$. Yet from 
    \autoref{lem:ic.xqc.one} we know that $\overline{R}^1$ and $\overline{S}^1$  are equivalent. Therefore for any distinguisher $\sys D$, $d^{\sys D} \left({\overline{\sys R}}^{\ell},  {\overline{\sys S}}^{\ell}\right) = 0 $.
    
	Overall, we have proved that for any distinguisher $\sys D$,
	$$d^{\sys D} \left({\sys R}^{\ell},  \sys S^{\ell}\right) \leq    \frac{\ell^2}{2^\nu}    ,$$    
	which concludes the construction proof.
\end{proof} 

\begin{remark*}
The protocol given in \autoref{thm:ic.xqc} also has to satisfy
correctness, i.e., when the distinguisher always puts back the same
state Bob should decrypt correctly. One can easily see that this
holds, since in the real world, the state will be flipped on Alice's side and be flipped back on Bob side, thus the distinguisher will get the same state back at interface $\iB$.  % In the ideal world, the simulator prepares a maximum entangled state. It gets measured on Bell basis and the result is $0$. Thus the state will not be flipped by the ideal channel and distinguisher will get the same state back at interface $\iB$ as well.
\end{remark*}

Suppose now that one has a $\prf$ resource and a bound $\epsilon^\prf$ satisfying \eqnref{eq:prf2urf},
that is, indistinguishable from $\urf$ within $\epsilon^\prf$, the following corollary follows trivially from the composition theorem \autoref{thm:composition} in \aref{app:composition}.    
\begin{corollary}
\begin{multline*}
\left[\iqc^{\ell,m+\nu}, \key^{r}, \prf^{r,\nu,2m}, \prf^{r,\nu,2m},\qc^{\ell,m,m+\nu}_A, \qc^{\ell,m,m+\nu}_B \right] \\
  \xrightarrow{\pi'_{AB},\epsilon}
  \left[\xqc^{\ell, m}, \qc^{2\ell,m,m+\nu}_E \right].
\end{multline*}
where $\pi'_{AB} =(\pi_{AB}, \pi^{\prf}), \epsilon(\sys D) = \epsilon^\prf(\sys D \sys C)+\ell^2  2^{-\nu}$ and $\sys C$ is the system including $
  \pi_{AB}, \iqc^{\ell,m+\nu}, \qc^{\ell,m,m+\nu}_A, \qc^{\ell,m,m+\nu}_B$.

\end{corollary}

%%% Local Variables:
%%% TeX-master: "qccFull" 
%%% End:

    \section{Relations to Game-Based Security Definitions}
\label{sec:relations}

In this section we explore the relations between our constructive security definitions and two game based security definitions for (specific protocols making use of) symmetric quantum encryption schemes, both introduced in \cite{alagic2018unforgeable}.
The two notions we consider are those of \emph{quantum ciphertexts indistinguishability under adaptive chosen-ciphertext attack} ($\AGM\QCCA2$) and \emph{quantum authenticated encryption} ($\QAE$).
Both definitions are inspired by classical security notions which intrinsically require the ability to copy data, which in \cite{alagic2018unforgeable} were successfully translated into quantum analogue by circumventing the no-cloning theorem.

We will first show that $\QAE$ security exactly implies the constructive cryptography security notion of \emph{constructing a secure channel from an insecure one and a shared secret key}, which we call $\QSEC$ (but is actually stronger, and thus we also show a separation).
Secondly, we will relate the $\AGM\QCCA2$ security definition to the constructive cryptography security notion of \emph{constructing a confidential channel from an insecure one and a shared secret key}, which we call $\QCNF$, but the implication will be less direct.
In fact, we introduce two new (intermediate) game-based security definitions, $\RRC\QCCA2$ and $\RRO\QCCA2$, and show that:
\begin{enumerate}
    \item The classical versions of $\AGM\QCCA2$ and $\RRC\QCCA2$ are asymptotically equivalent;
    \item For a restricted class of schemes, $\RRC\QCCA2$ implies $\RRO\QCCA2$ (they are actually equivalent);
    \item $\RRO\QCCA2$ implies $\QCNF$ (but is actually stronger).
\end{enumerate}
We leave open the question whether it is possible to generalize (2.) to general schemes.
Throughout this section we will assume that the plaintext and the ciphertext spaces comprise elements of the same length, an thus ignore the corresponding superscripts for channels and quantum computers.

\subsection{Background and Notation}

%We remark that we adopt a finite security approach also in this section, and we therefore adapt the game-based definitions from \cite{alagic2018unforgeable} to match our treatment, and this also includes the very definition of a symmetric quantum encryption scheme, as originally introduced in \cite{alagic2016computational}.
%Note that for Hilbert space $\Hi_M$, we denote the class of positive, Hermitian, trace-one linear operators on $\Hi_M$ as $\Msg$.
%
%\begin{definition}[Symmetric Quantum Encryption Scheme $\sqes$]
%    \label{def:sqes}
%    A \emph{symmetric quantum encryption scheme ($\sqes$)} $\sch\df(\Gen,\Enc,\Dec)$ for (classical) key space $\Key$, message space $\Msg$ for quantum register $M$, and ciphertext space $\Ctx$ for quantum register $C$, is a collection of three functions:
%    \begin{itemize}
%        \item A \emph{key-generation function} $\Gen$ inducing a distribution over $\Key$, which outputs a \emph{symmetric key} $k\in\Key$.
%        We write $k\gets\Gen()$.
%        \item An \emph{encryption function} $\Enc:\Key\x\Msg\to\Ctx$, which takes as input a key $k\in\Key$, a mixed state $\msg\reg M\in\Msg$ as message, and outputs a mixed state $\ctx\reg C\in\Ctx$ as ciphertext.
%        We write $\ctx\reg C\gets\Enc_k(\msg\reg M)$.
%        \item A \emph{decryption function} $\Dec:\Key\x\Ctx\to\Msgb$, which takes as input a key $k\in\Key$ and a mixed state $\ctx\reg C\in\Ctx$ as ciphertext, and outputs either a mixed state $\msg\reg M\in\Msg$ or $\Bot$.
%        We write $\msg\reg M\gets\Dec_k(\ctx\reg C)$.
%    \end{itemize}
%\end{definition}

In \cite{barnum2002qauth}, a characterization of any \emph{symmetric quantum encryption schemes} (\sqes) was given, which states that encryption works by attaching some (possibly) key-dependent auxiliary state, and applying a unitary operator, and decryption undoes the unitary, and then checks whether the support of the state in the auxiliary register has changed.
%Note that upon encryption, some randomness from a \emph{(classical) randomness space} $\Rnd$ is needed (for a \sqes\ to be at least $\CPA{}$ secure).
%We remark that here we only consider schemes where the auxiliary state is \emph{key-independent} (but might still depend on the randomness), since all the concrete schemes considered in \cite{alagic2018unforgeable} are of this form.
%We also restrict our attention to \sqes s that upon encryption also explicitly append the randomness to the ciphertext. %\todomarginFabio{Motivate?}
Thus, as pointed out in \cite{alagic2018unforgeable}, for key-generation function $\Gen$ (inducing a probability distribution over some key-space $\Key$), encryption function $\Enc$, and decryption function $\Dec$, we can characterize a \sqes\ $\sch\df(\Gen,\Enc,\Dec)$ as follows.
\begin{lemma}[\protect{\cite[Corollary 1]{alagic2018unforgeable}}]
    Let $\sch=(\Gen,\Enc,\Dec)$ be a \sqes.
    Then for every $k\in\Key$ there exists a probability distribution $p_k:\Rnd\to[0,1]$ and a family of quantum states $\{\ket{\psi_{k,r}}\reg T\}_{r\in\Rnd}$, with $\Pi_{k,r}\reg T\df\ketbra{\psi_{k,r}}{\psi_{k,r}}\reg T$, such that:
    \begin{itemize}
        \item $\Enc_k(\msg\reg M)\df V_k\lp\msg\reg M\otimes\Pi_{k,r}\reg T\rp V_k^\dag$, where $r$ is sampled according to $p_k$;
        \item $\Dec_k(\ctx\reg C)\df\Tr_T\lp P_{\aux_k}\reg T(V_k^\dag\ctx\reg CV_k)P_{\aux_k}\reg T\rp+\hat D_k\lp\bar P_{\aux_k}\reg T(V_k^\dag\ctx\reg CV_k)\bar P_{\aux_k}\reg T\rp$;
    \end{itemize}
    where $P_{\aux_k}\reg T$ and $\bar P_{\aux_k}\reg T$ are the orthogonal projectors onto the support of
    \[\aux_k\reg T\df\sum_{r\in\Rnd}p_k(r)\cdot\Pi_{k,r}\reg T=\sum_{r\in\Rnd}p_k(r)\cdot\ketbra{\psi_{k,r}}{\psi_{k,r}}\reg T.\]
\end{lemma}
%\begin{lemma}[Corollary 1 in \cite{alagic2018unforgeable}, restated and adapted]
%    Let $\sch=(\Gen,\Enc,\Dec)$ be a \sqes\ with key-independent auxiliary register, and which appends explicitly the randomness $r\in\Rnd$ to the ciphertext.
%    Then there exists a probability distribution $p:\Rnd\to[0,1]$ and a family of quantum states $\{\ket{\psi_r}\reg T\}_{r\in\Rnd}$, with $\Pi_r\reg T\df\ketbra{\psi_r}{\psi_r}\reg T$, such that:
%    \begin{itemize}
%        \item $\Enc_k(\msg\reg M)\df V_k\lp\msg\reg M\otimes\Pi_r\reg T\rp V_k^\dag\otimes\ketbra rr$, where $r$ is sampled according to $p$;
%        \item $\Dec_k(\ctx\reg C\otimes\ketbra rr)\df\Tr_T\lp P_\aux(V_k^\dag\ctx\reg CV_k)P_\aux\rp+\hat D_k\lp\bar P_\aux(V_k^\dag\ctx\reg CV_k)\bar P_\aux\rp$;
%    \end{itemize}
%    where $P_\aux$ and $\bar P_\aux$ are the orthogonal projectors onto the support of
%    \[\aux\reg T\df\sum_{r\in\Rnd}p_r\,\Pi_r\reg T=\sum_{r\in\Rnd}p_r\,\ketbra{\psi_r}{\psi_r}\reg T.\]
%\end{lemma}

For a \sqes\ $\sch$, we define a security notion $\mathsf{XXX}$ in
terms of the advantage $\Adv{\mathrm{xxx}}\sch{\sys D}$ of a
distinguisher $\sys D$ in solving some (usually distinction) problem
involving $\sch$.  In the asymptotic setting, security of $\sch$
according to notion $\mathsf{XXX}$ should be interpreted as
$\Adv{\mathrm{xxx}}\sch{\sys D}$ being negligible for every $\sys D$
from some class $\bD$ of distinguishers (usually, efficient
distinguishers). Following the finite security approach, here we are
just interested in relating advantages of different notions, making
use of black-box reductions.  Therefore, for a second notion
$\mathsf{YYY}$, we say that \emph{$\mathsf{XXX}$ (security) implies
  $\mathsf{YYY}$ (security)} if and only if
$\Adv{\mathrm{yyy}}\sch{\sys D}\leq c\cdot\Adv{\mathrm{xxx}}\sch{\sys D\sys
  C}$, for some $c\geq1$, where $\sys C$ denotes the black-box reduction that uses the
distinguisher $\sys D$ for $\mathsf{YYY}$ to make a new distinguisher
$\sys D\sys C$ for $\mathsf{XXX}$.

When describing experiments involving interaction between a distinguisher\footnote{We understand the distinguisher $\sys D$ as stateful, which can therefore be invoked multiple times (without making explicit the various updated states).} $\sys D$ and a game system $\sys G$, we use pseudo-code from $\sys G$'s perspective, that is, the ${\bf return}$ statement indicates what is output by the latter.
Note that this implies that for distinction problems we always make the game system output the bit output by the distinguisher.
In this case we use the expression $\sys D[\sys G]$ to denote the bit output by $\sys D$ after interacting with $\sys G$.
On the other hand, if the output bit is decided by $\sys G$ (as is the case for the $\AGM\QCCA2$ definition, restated in \aref{app:agm}, which is \emph{not} a distinction problem), we use the expression $\sys G[\sys D]$.
Moreover, we use both expressions not only for the returned value, but also for denoting the whole random experiments.
When specifying that a distinguisher $\sys D$ has access to a list of oracles, e.g. ${\bf O}_1(\cdot)$ and ${\bf O}_2(\cdot)$, we write $x\gets\sys D^{{\bf O}_1(\cdot),{\bf O}_2(\cdot)}$, where the variable $x$ holds the value output by $\sys D$ after the interaction with the oracles.
% Note that due to the no-cloning theorem, it is not possible for $\sys D$ to query an oracle twice on the same input.
We denote the application of a two-outcome projective measurement,
e.g. $\{P_{\aux_k}\reg T,\One-P_{\aux_k}\reg T\}$, as
$\{P_{\aux_k}\reg T,\One-P_{\aux_k}\reg T\}\Out b$, where
$b\in\{0,1\}$ is the result of the measurement (we associate $0$ to
the the first outcome and $1$ to the second).  The state $\epr$ is the
EPR pair (one of the Bell state), to which we associate the
two-outcome projective measurement $\{\Pi_+,\One-\Pi_+\}$.
Furthermore, by $XY\gets\epr$ we mean that the EPR pair has been
prepared on registers $XY$, and we use $\tau\reg X$ as a shorthand for
the reduced state in register $X$, that is, half of a
maximally-entangled state.  Finally, we use sans-serif font for
boolean labels, e.g. $\lCheat$, and by the expression $\ite\lCheat xy$
we mean the value $x$ if $\lCheat$ is true ($1$), and the value $y$
otherwise (false, $0$).

\subsection{Relating $\QAE$ and $\QSEC$}

In this section we first present the quantum authenticated encryption security definition introduced in \cite{alagic2018unforgeable}, and then show that it directly implies our constructive security notion $\QSEC$ of constructing a secure channel from an insecure one and a shared secret key.

\subsubsection{$\QAE$ Security Definition (\cite{alagic2018unforgeable}).}

We begin by restating what it means for a \sqes\ $\sch$ to be secure in the $\QAE$ sense according to \cite{alagic2018unforgeable}.
On a high level, a distinguisher $\sys D$ must not be able to distinguish between two scenarios: in the first (the real one), it has access to regular encryption and decryption oracles, whereas in the second (the ideal one), it has access to an encryption oracle which replaces its queried plaintexts by random ones (half of a maximally-entangled state), and a decryption oracle that normally decrypts ciphertext not returned by the encryption oracle, but answers with the originally queried plaintexts otherwise (thus not really performing correct decryption).
Note that this security notion, as shown in \cite{alagic2018unforgeable}, when phrased classically is equivalent to the canonical notion of authenticated encryption (dubbed $\IND\CCA3$ by Shrimpton in \cite{shrimpton2004characterization}).
The only difference with the latter, is that the decryption oracle returns $\bot$ when queried on ciphertexts previously returned by the encryption oracle.
But crucially, this detail is what would not make it possible to adapt $\IND\CCA3$ into a quantum definition: returning $\bot$ would require the game to \emph{copy data} (store the ciphertexts returned by the encryption oracle, and then compare them to each query to the decryption oracle), which is not allowed in general in the quantum world.
Nevertheless, the formulation of $\QAE$ introduced in \cite{alagic2018unforgeable} works quantumly because, intuitively, \emph{``it is possible to compare random states generated as half of a maximally-entangled state''}: the trick consists of first ignoring (but storing) each plaintext submitted by the adversary to the encryption oracle, and then, for each plaintext, prepare an EPR pair $\epr$, encrypt just half of it, and store the other half (as well as the involved randomness) together with the original plaintext submitted by the distinguisher; then the decryption oracle normally decrypts each ciphertext, and subsequently applies a projective measurement on the support of $\epr$ to the obtained plaintext against each stored half, and the associated original plaintext can thus be easily retrieved.
We now restate the definition from \cite{alagic2018unforgeable} (Definition 10 therein), adapted to our notation, and in the concrete setting (as opposed to the asymptotic one).

\begin{definition}[$\QAE$ Security \cite{alagic2018unforgeable}]
    \label{def:qae}
    For \sqes\ $\sch\df(\Gen,\Enc,\Dec)$ (implicit in all defined systems) we define the \emph{$\QAE$-advantage} of $\sch$ for distinguisher $\sys D$ as
    \[\Adv\qae\sch{\sys D}\df\pr{}{\sys D[\sys G^\qaer]=1}-\pr{}{\sys D[\sys G^\qaei]=1},\]
    where the interactions of $\sys D$ with game systems $\sys G^\qaer$ and $\sys G^\qaei$ are defined in \autoref{fig:qae}.
\end{definition}
\begin{figure}[tb]
    \begin{tcbraster}[raster columns=1, sidebyside, lefthand width=.28\textwidth]
        \begin{algobox}{\textwidth}{Experiments ${\sys D[\sys G^\qaer]}$ and ${\sys D[\sys G^\qaei]}$ for \sqes\ ${\sch\df(\Gen,\Enc,\Dec)}$}
            \begin{algorithmic}
                \State $k\gets\Gen()$
                \State\Return $\sys D^{\Enc_k(\cdot),\Dec_k(\cdot)}$
            \end{algorithmic}
            \tcblower
            \begin{algorithmic}
                \State $k\gets\Gen()$
                \State $\cM\gets\es$
                \State\Return $\sys D^{{\bf Enc}(\cdot),{\bf Dec}(\cdot)}$
                \vspace{1mm}\hrule\vspace{1mm}
                \Orac[Enc]{$\msg\reg M$}
                    \State $\hat r\samp\Rnd$
                    \State $\hat M\tilde M\gets\epr$
                    \State $\hat\ctx\reg{\hat C}\gets V_k(\hat{\msg}\reg{\hat M}\otimes\Pi_{k,\hat r}\reg T)V_k^\dag$\cmt{Ignore $\msg\reg M$}
                    \State $\cM\gets\cM\cup\{(\hat r,\tilde M,M)\}$
                    \State\Return $\hat\ctx\reg{\hat C}$
                \EndOrac
                    \vspace{1mm}\hrule\vspace{1mm}
                \Orac[Dec]{$\ctx\reg C$}
                    \State $\hat MT\gets V_k^\dag\ctx\reg CV_k$
                    \For{\Each$(\hat r,\tilde M,M)\in\mc M$}
                        \If{$\{\Pi_{k,\hat r},\One-\Pi_{k,\hat r}\}(\aux\reg T)\Out0$}
                            \If{$\{\Pi_+,\One-\Pi_+\}(\varphi\reg{\hat M\tilde M})\Out0$}
                                \State\Return $\msg\reg M$
                            \EndIf
                        \EndIf
                    \EndFor
                    \State\Return $\Bot$
                \EndOrac
            \end{algorithmic}
        \end{algobox}
    \end{tcbraster}
    \caption{$\QAE$ security games $\sys G^\qaer$ ({\bf left}) and $\sys G^\qaei$ ({\bf right}).}
    \label{fig:qae}
\end{figure}

\subsubsection{$\QAE$ Implies $\QSEC$.}

Here we denote by $\sys G^{\qaer,\ell}$ and $\sys G^{\qaei,\ell}$ the games $\sys G^{\qaer}$ and $\sys G^{\qaei}$ where the distinguisher is allowed to make at most $\ell$ queries to each oracle (and analogously for $\Adv{\qae,\ell}\sch{\sys D}$).

\begin{figure}[tb]
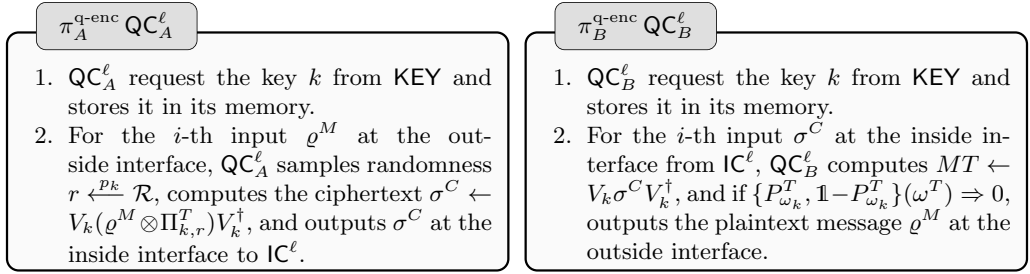

    \begin{tcbraster}[raster columns=2, raster equal height=rows]
        \begin{convbox}{\textwidth}{$\pi^\qenc_A\,\qc^\ell_A$}
            \begin{enumerate}
                \item $\qc^\ell_A$ request the key $k$ from $\key$ and stores it in its memory.
                \item For the $i$-th input $\msg\reg M$ at the outside interface, $\qc^\ell_A$ samples randomness $r\samp\Rnd$, computes the ciphertext $\ctx\reg C\gets V_k(\msg\reg M\otimes\Pi_{k,r}\reg T)V_k^\dag$, and outputs $\ctx\reg C$ at the inside interface to $\iqc^\ell$.
            \end{enumerate}
        \end{convbox}
        \begin{convbox}{\textwidth}{$\pi^\qenc_B\,\qc^\ell_B$}
            \begin{enumerate}
                \item $\qc^\ell_B$ request the key $k$ from $\key$ and stores it in its memory.
                \item For the $i$-th input $\ctx\reg C$ at the inside interface from $\iqc^\ell$, $\qc^\ell_B$ computes $MT\gets V_k\ctx\reg CV_k^\dag$, and if $\{P_{\aux_k}\reg T,\One-P_{\aux_k}\reg T\}(\aux\reg T)\Out0$, outputs the plaintext message $\msg\reg M$ at the outside interface.
            \end{enumerate}
        \end{convbox}
    \end{tcbraster}
    \caption{Encryption and decryption protocols.}
    \label{fig:enc-dec}
\end{figure}
\newcommand\solidrule[1][1cm]{\rule[0.5ex]{#1}{.4pt}}
\newcommand\dashedrule{\mbox{\solidrule[2mm]\hspace{2mm}\solidrule[2mm]\hspace{2mm}\solidrule[2mm]}}
\begin{figure}[tb]
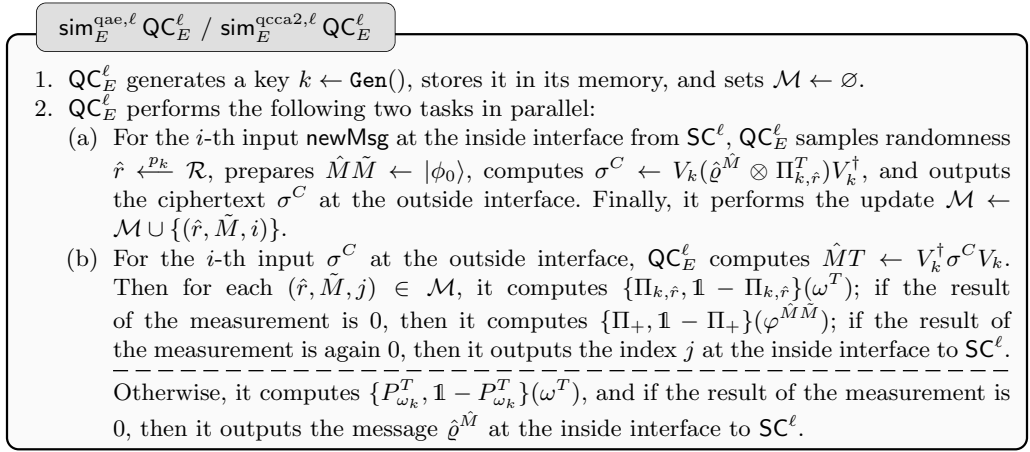

    \centering
    \begin{convbox}{\textwidth}{$\simul^{\qae,\ell}_E\,\qc^\ell_E$ / $\simul^{\qcca2,\ell}_E\,\qc^\ell_E$}
        \begin{enumerate}
            \item $\qc^\ell_E$ generates a key $k\gets\Gen()$, stores it in its memory, and sets $\cM\gets\es$.
            \item $\qc^\ell_E$ performs the following two tasks in parallel:
            \begin{enumerate}
                \item For the $i$-th input $\lNewMsg$ at the inside interface from $\sqc^\ell$, $\qc^\ell_E$ samples randomness $\hat r\samp\Rnd$, prepares $\hat M\tilde M\gets\epr$, computes $\ctx\reg C\gets V_k(\hat\msg\reg{\hat M}\otimes\Pi_{k,\hat r}\reg T)V_k^\dag$, and outputs the ciphertext $\ctx\reg C$ at the outside interface.
                Finally, it performs the update $\cM\gets\cM\cup\{(\hat r,\tilde M,i)\}$.
                \item For the $i$-th input $\ctx\reg C$ at the outside interface, $\qc^\ell_E$ computes $\hat MT\gets V_k^\dag\ctx\reg CV_k$.
                Then for each $(\hat r,\tilde M,j)\in\mc M$, it computes $\{\Pi_{k,\hat r},\One-\Pi_{k,\hat r}\}(\aux\reg T)$; if the result of the measurement is $0$, then it computes $\{\Pi_+,\One-\Pi_+\}(\varphi\reg{\hat M\tilde M})$; if the result of the measurement is again $0$, then it outputs the index $j$ at the inside interface to $\sqc^\ell$.\\[-2mm]
                \hdashrule{12cm}{0.5pt}{2mm 1mm}\\
                Otherwise, it computes $\{P_{\aux_k}\reg T,\One-P_{\aux_k}\reg T\}(\aux\reg T)$, and if the result of the measurement is $0$, then it outputs the message $\hat\msg\reg{\hat M}$ at the inside interface to $\sqc^\ell$. %, and $\hat D(\hat\msg\reg{\hat M}\otimes\aux\reg T)$ otherwise.
            \end{enumerate}
        \end{enumerate}
    \end{convbox}
    \caption{$\QAE$ (until the dashed line) and $\QCCA2$ (until the end) simulators.}
    \label{fig:sim}
\end{figure}

\begin{theorem}
    \label{thm:qae}
    Let $\sch\df(\Gen,\Enc,\Dec)$ be a \sqes\ (implicit in all defined systems).
    Then with protocol $\pi^\qenc_{AB}=(\pi^\qenc_A,\pi^\qenc_B)$ making use of quantum computers $\qc^\ell_A$ and $\qc^\ell_B$ as defined in \autoref{fig:enc-dec}, simulator $\simul^\qae_E$ making use of quantum computer $\qc^\ell_E$ as defined in \autoref{fig:sim} (until the dashed line), and (trivial) reduction system $\sys C$ as specified in the proof, for any distinguisher $\sys D$ we have
    \[\Delta^{\sys D}(\pi^\qenc_{AB}\,[\key,\iqc^\ell,\qc^\ell_A,\qc^\ell_B],\simul^\qae_E\,[\sqc^\ell,\qc^\ell_E])\leq\Adv{\qae,\ell}\sch{\sys D\sys C}.\]
\end{theorem}
\begin{proof}
    Let $\sys R\df\pi^\qenc_{AB}\,[\key,\iqc^\ell,\qc^\ell_A,\qc^\ell_B]$ and $\sys S\df\simul^\qae_E\,[\sqc^\ell,\qc^\ell_E]$.
    We need to provide a reduction system $\sys C$ so that distinguishing between the games $\sys G^{\qaer,\ell}$ and $\sys G^{\qaei,\ell}$ can be reduced to distinguishing between resources $\sys R$ and $\sys S$.
    Let $\sys D$ be a distinguisher for $\sys R$ and $\sys S$.
    Then we define an adversary $\sys D'=\sys D\sys C$ with access to real/ideal encryption oracle $\mathbf{Enc(\cdot)}$ and real/ideal decryption oracle $\mathbf{Dec(\cdot)}$ provided by $\sys G^{\qaer,\ell}$/$\sys G^{\qaei,\ell}$, where $\sys C$ behaves as follows towards $\sys D$:
    \begin{itemize}
        \item {\bf Interface $\iA$:} On input message $\msg$, output $\mathbf{Enc}(\msg)$ at interface $\iE$.
        \item {\bf Interface $\iE$:} On input state $\ctx$, output $\mathbf{Dec}(\ctx)$ at interface $\iB$.
    \end{itemize}
    Note that we trivially have that $\sys C\sys G^{\qaer,\ell}\equiv\sys R$ and $\sys C\sys G^{\qaei,\ell}\equiv\sys S$, hence
    \begin{align*}
    \Adv{\qae,\ell}\sch{\sys D\sys C}
    &=\pr{}{\sys D\sys C[\sys G^{\qaer,\ell}]=1}-\pr{}{\sys D\sys C[\sys G^{\qaei,\ell}]=1}\\
    &=\pr{}{\sys D[\sys C\sys G^{\qaer,\ell}]=1}-\pr{}{\sys D[\sys C\sys G^{\qaei,\ell}]=1}\\
    &=\pr{}{\sys D[\sys R]=1}-\pr{}{\sys D[\sys S]=1}\\
    &=\Delta^{\sys D}(\sys R,\sys S).%\qedhere
    \end{align*}
    This concludes the proof.
\end{proof}

%\ifsub\else
\begin{corollary}
    \label{cor:qae}
    With $\e(\sys D):=\sup_{\sys D'\in\mathcal B(\sys D)}\Adv{\qae,\ell}\sch{\sys D'}$, we have
    \begin{equation*}
        \left[\key,\iqc^\ell,\qc^\ell_A,\qc^\ell_B\right]
        \xrightarrow{\pi^{\qenc}_{AB},\eps}
        \left[\sqc^\ell,\qc^\ell_E\right],
    \end{equation*}
    where the class $\mathcal B(\sys D)$ is defined in \eqnref{eq:distinguisher.class}, \autoref{sec:ac.theory}.
\end{corollary}
%\fi

\subsubsection{$\QAE$ is Stronger than $\QSEC$.}

We remark that even though $\QAE$ implies $\QSEC$, the converse is not true.
In particular, we find that $\QAE$ is an (unnecessarily) stronger notion than $\QSEC$.
We can in fact show that there are \sqes s that satisfy $\QSEC$, but not $\QAE$.
Following \cite{canetti2003relaxing}, in order to show this fact it suffices to take any \sqes\ $\sch$ which is $\QAE$ secure, and slightly modify it into a new \sqes\ $\sch'$ so that a classical $0$-bit is appended to every encryption, which is then ignored upon decryption.
Now an adversary can flip the bit of a ciphertext that it got from the encryption oracle, and then query the decryption oracle on the new ciphertext: in the real setting it will get back the original message, while in the ideal setting it will get back $\Bot$, and can thus perfectly distinguish between the two, hence $\sch'$ cannot be $\QAE$ secure.
On the other hand, $\sch'$ is still $\QSEC$ secure because it can still be used to achieve the construction of a secure channel from an insecure one and a shared secret key.
This is possible by using a simulator which works essentially as $\simul_E^{\qae,\ell}\qc_E^\ell$ from \autoref{fig:sim}, but which ignores the bit.

\subsection{Relating $\QCCA2$ and $\QCNF$}

The goal of this section is to present and relate several $\QCCA2$ security definitions.
The original $\AGM\QCCA2$ definition from \cite{alagic2018unforgeable} is restated in \aref{app:agm}.
In this section we begin by introducing a new definition, $\RRC\QCCA2$ (where $\mathsf{RRC}$ stands for \emph{``real-or-random challenge''}), which is similar to $\AGM\QCCA2$. %, and is a more natural candidate for defining a game-based notion of confidentiality for quantum encryption schemes.
Both notions define a challenge phase, and thus we introduce a third variant, $\RRO\QCCA2$ (where $\mathsf{RRO}$ stands for \emph{``real-or-random oracles''}), in which there is no real-or-random challenge, but rather access to real-or-random oracles.
Crucially, the latter is identical to $\QAE$ as introduced by \cite{alagic2018unforgeable}, up to a small detail: \emph{upon decryption, if the ciphertext was not generated by the encryption oracle, instead of returning $\Bot$, return the decrypted plaintext.}
Finally, we show that for a restricted class of \sqes s, $\RRC\QCCA2$ implies $\RRO\QCCA2$, and for any \sqes s, $\RRO\QCCA2$ implies $\QCNF$.

\subsubsection{$\RRC\QCCA2$ Security Definition.}

We now introduce an alternative game-based security definition that seems more natural than $\AGM\QCCA2$.
This notion is defined in terms of a distinction problem (as opposed to $\AGM\QCCA2$), and essentially it is analogous to the test setting of the latter, $\sys G^\qccat2$, but where the decryption oracle provided to the distinguisher behaves differently: after the real-or-random challenge phase, upon querying the challenge ciphertext, it will respond with the plaintext originally submitted by the distinguisher, in both the real and ideal settings.
Note that this is possible in the ideal setting, because we make use of the same trick as in $\sys G^\qccaf2$, but we do not just set a flag whenever we detect that the adversary is cheating, but rather return the original message that it submitted as challenge.
Since a similar behavior is implemented in the real setting, the adversary must really be able to distinguish between ciphertexts in order to win.

\begin{definition}[$\RRC\QCCA2$ Security]
    \label{def:rrc}
    For \sqes\ $\sch\df(\Gen,\Enc,\Dec)$ (implicit in all defined systems) we define the \emph{$\RRC\QCCA2$-advantage} of $\sch$ for distinguisher $\sys D$ as
    \[\Adv{\rrc\qcca2}\sch{\sys D}\df\pr{}{\sys D[\sys G^{\rrc\qccar2}]=1}-\pr{}{\sys D[\sys G^{\rrc\qccai2}]=1},\]
    where the interactions of $\sys D$ with game systems $\sys G^{\rrc\qccar2}$ and $\sys G^{\rrc\qccai2}$ are defined in \autoref{fig:rrc}.
\end{definition}
\begin{figure}[tb]
    \begin{tcbraster}[raster columns=1, sidebyside, lefthand width=.302\textwidth]
        \begin{algobox}{\textwidth}{Experiments ${\sys G^{\rrc\qccar2}[\sys D]}$ and ${\sys G^{\rrc\qccai2}[\sys D]}$ for \sqes\ ${\sch\df(\Gen,\Enc,\Dec)}$}
            \begin{algorithmic}
                \State $k\gets\Gen()$
                \State $\msg\reg M\gets\sys D^{\Enc_k(\cdot),\Dec_k(\cdot)}$
                \State $\ctx\reg C\gets\Enc_k(\msg\reg M)$
                \State $b'\gets\sys D^{\Enc_k(\cdot),\Dec_k(\cdot)}(\sigma^C)$
                \State\Return $b'$
            \end{algorithmic}
            \tcblower
            \begin{algorithmic}
                \State $k\gets\Gen()$
                \State $\bar\msg\reg{\bar M}\gets\sys D^{\Enc_k(\cdot),\Dec_k(\cdot)}$\cmt{Keep $\bar\msg\reg{\bar M}$}
                \State $\hat r\samp\Rnd$\cmt{Keep $\hat r$}
                \State $\hat M\tilde M\gets\epr$
                \State $\hat\ctx\reg{\hat C}\gets V_k(\hat\msg\reg{\hat M}\otimes\Pi_{k,\hat r}\reg T)V_k^\dag$\cmt{Ignore $\bar\msg\reg M$}
                \State $b'\gets\sys D^{\Enc_k(\cdot),{\bf Dec}(\cdot)}(\hat\ctx\reg{\hat C})$
                \State\Return $b'$
                \vspace{1mm}\hrule\vspace{1mm}
                \Orac[Dec]{$\ctx\reg C$}
                    \State $MT\gets V_k^\dag\ctx\reg CV_k$
                    \If{$\{P_{\aux_k}\reg T,\One-P_{\aux_k}\reg T\}(\aux\reg T)\Out0$}
                        \If{$\{\Pi_{k,\hat r},\One-\Pi_{k,\hat r}\}(\aux\reg T)\Out0$}
                            \If{$\{\Pi_+,\One-\Pi_+\}(\varphi\reg{M\tilde M})\Out0$}
                                \State\Return $\bar\msg\reg{\bar M}$\cmt{Original challenge}
                            \EndIf
                        \EndIf
                    \Else
                        \State\Return $\hat D_k(\rho\reg{MT})$\cmt{Invalid ciphertext}
                    \EndIf
                    \State\Return $\msg\reg M$
                \EndOrac
            \end{algorithmic}
        \end{algobox}
    \end{tcbraster}\caption{$\RRC\QCCA2$ games $\sys G^{\rrc\qccar2}$ ({\bf left}) and $\sys G^{\rrc\qccai2}$ ({\bf right}).}
    \label{fig:rrc}
\end{figure}

\subsubsection{$\RRO\QCCA2$ Security Definition.}

In order to relate the latter definition with a constructive notion of confidentiality, it is helpful to have a game-based security definition which analogously to $\QAE$ defines a real and an ideal setting (by specifying real-or-random oracles, and in particular, not only a real-or-random challenge).
We do this by introducing the notion $\RRO\QCCA2$, which can be seen as a natural extension of $\RRC\QCCA2$.

\begin{definition}[$\RRO\QCCA2$ Security]
    \label{def:rro}
    For \sqes\ $\sch\df(\Gen,\Enc,\Dec)$ (implicit in all defined systems) we define the \emph{$\RRO\QCCA2$-advantage} of $\sch$ for distinguisher $\sys D$ as
    \[\Adv{\rro\qcca2}\sch{\sys D}\df\pr{}{\sys D[\sys G^{\rro\qccar2}]=1}-\pr{}{\sys D[\sys G^{\rro\qccai2}]=1},\]
    where the interactions of $\sys D$ with game systems $\sys G^{\rro\qccar2}$ and $\sys G^{\rro\qccai2}$ are defined in \autoref{fig:rro}.
\end{definition}
\begin{figure}[tb]
    \begin{tcbraster}[raster columns=1, sidebyside, lefthand width=.28\textwidth]
        \begin{algobox}{\textwidth}{Experiments ${\sys D[\sys G^{\rro\qccar2}]}$ and ${\sys D[\sys G^{\rro\qccai2}]}$ for \sqes\ ${\sch\df(\Gen,\Enc,\Dec)}$}
            \begin{algorithmic}
                \State $k\gets\Gen()$
                \State\Return $\sys D^{\Enc_k(\cdot),\Dec_k(\cdot)}$
            \end{algorithmic}
            \tcblower
            \begin{algorithmic}
                \State $k\gets\Gen()$
                \State $\cM\gets\es$
                \State\Return $\sys D^{{\bf Enc}(\cdot),{\bf Dec}(\cdot)}$
                \vspace{1mm}\hrule\vspace{1mm}
                \Orac[Enc]{$\msg\reg M$}
                    \State $\hat r\samp\Rnd$
                    \State $\hat M\tilde M\gets\epr$
                    \State $\hat\ctx\reg{\hat C}\gets V_k(\hat\msg\reg{\hat M}\otimes\Pi_{k,\hat r}\reg T)V_k^\dag$\cmt{Ignore $\msg\reg M$}
                    \State $\cM\gets\cM\cup\{(\hat r,\tilde M,M)\}$
                    \State\Return $\hat\ctx\reg{\hat C}$
                \EndOrac
                    \vspace{1mm}\hrule\vspace{1mm}
                \Orac[Dec]{$\ctx\reg C$}
                    \State $\hat MT\gets V_k^\dag\ctx\reg CV_k$
                    \For{\Each$(\hat r,\tilde M,M)\in\mc M$}
                        \If{$\{\Pi_{k,\hat r},\One-\Pi_{k,\hat r}\}(\aux\reg T)\Out0$}
                            \If{$\{\Pi_+,\One-\Pi_+\}(\varphi\reg{\hat M\tilde M})\Out0$}
                                \State\Return $\msg\reg M$
                            \EndIf
                        \EndIf
                    \EndFor
                    \If{$\{P_{\aux_k}\reg T,\One-P_{\aux_k}\reg T\}(\aux\reg T)\Out0$}
                        \State\Return $\hat\msg\reg{\hat M}$
                    \Else
                        \State\Return $\hat D_k(\rho\reg{\hat MT})$\cmt{Invalid ciphertext}
                    \EndIf
                \EndOrac
            \end{algorithmic}
        \end{algobox}
    \end{tcbraster}
    \caption{$\RRO\QCCA2$ games $\sys G^{\rro\qccar2}$ ({\bf left}) and $\sys G^{\rro\qccai2}$ ({\bf right}).}
    \label{fig:rro}
\end{figure}

\subsubsection{Relating $\AGM\QCCA2$ and $\RRC\QCCA2$.}

We feel that $\RRC\QCCA2$ is a much simpler and more natural definition than $\AGM\QCCA2$.
In fact, in \cite{alagic2018unforgeable} the authors claim that $\AGM\QCCA2$ is a ``natural'' security definition based on the fact that its classical analogon is shown to be equivalent to (a variation of) the standard classical $\IND\CCA2$ security definition.
We claim that our $\RRC\QCCA2$ is more natural in the sense that it is formulated as a normal distinction problem (as opposed to $\AGM\QCCA2$), and its classical analogon can be shown to be equivalent to standard classical $\IND\CCA2$ security much more directly (in particular, with no concrete security loss, as opposed to $\AGM\QCCA2$, where it is shown that the concrete reduction has a factor $2$ security loss).

Similarly as done in \cite{alagic2018unforgeable} for $\QAE$, whose classical restriction was shown to be equivalent to the common classical notion of authenticated encryption $\IND\CCA3$ from \cite{shrimpton2004characterization}, we now show that our $\RRC\QCCA2$ security notion, when casted to a classical definition, dubbed $\RRC\CCA2$, is equivalent (in particular, with no loss factors, as opposed to $\AGM\QCCA2$) to a common classical notion of $\IND\CCA2$.
The latter definition is the same mentioned in \cite{alagic2018unforgeable}, and comprises a real-or-random challenge, but the decryption oracle returns $\bot$ upon submitting the challenge ciphertext.
On the other hand, $\RRC\CCA2$ behaves exactly the same as $\IND\CCA2$, except that it always returns the challenge plaintext as originally submitted by the adversary upon querying the challenge ciphertext, independently from the (real or ideal) setting.
\begin{lemma}
    $\RRC\CCA2$ and $\IND\CCA2$ are equivalent.
\end{lemma}
\begin{proof}
    To transform $\RRC\CCA2$ into $\IND\CCA2$, the reduction simply stores the challenge ciphertext $\hat c$, and returns $\bot$ whenever the decryption oracle is queried upon $\hat c$.
    To transform $\IND\CCA2$ into $\RRC\CCA2$, the reduction simply stores the challenge plaintext $\hat m$ and the challenge ciphertext $\hat c$, and returns $\hat m$ whenever the decryption oracle is queried upon $\hat c$.
\end{proof}

\subsubsection{$\RRC\QCCA2$ Implies $\RRO\QCCA2$.}

As above, here we add as superscript the parameter $\ell$ to games and advantages to denote that the distinguisher is allowed to make at most $\ell$ queries to the oracles.
Note that we relate $\RRC\QCCA2$ and $\RRO\QCCA2$ for only the subclass of \sqes s which satisfy the following condition.

\newtheorem{condition}{\bf Condition}
\newcommand\conref[1]{\hyperref[#1]{Condition~\ref*{#1}}}
\begin{condition}
    \label{con:restricted}
    \textit{\sqes\ $\sch$ is such that the auxiliary state does not depend on the key (but possibly on the randomness), and it appends explicitly the randomness to the ciphertext, that is:
    \[\Enc_k(\msg\reg M)=U_{k,r}(\msg\reg M\otimes\Pi_r\reg T)U_{k,r}^\dagger\otimes\ketbra rr\reg R,\]
    for some unitary $U_{k,r}$ depending on both the key $k$ and the randomness $r$.}
\end{condition}
\noindent We remark that this restriction still captures all the
explicit protocols considered in \cite{alagic2018unforgeable}.

\begin{lemma}
    \label{lem:rrc-rro}
    Let $\sch$ be a \sqes\ satisfying \conref{con:restricted}.
    Then for reduction system $\sys C_I$ as specified in the proof, for any distinguisher $\sys D$ we have
    \[\Adv{\rro\qcca2,\ell}\sch{\sys D}\leq\ell\cdot\Adv{\rrc\qcca2,\ell-1}\sch{\sys D\sys C_I}.\]
\end{lemma}
\begin{proof}
    We need to provide a reduction system $\sys C_i$, parameterized on an index $i\in\{1,\ldots,\ell\}$, so that distinguishing between the two games $\sys G^{\rrc\qccar2,\ell-1}$ and $\sys G^{\rrc\qccai2,\ell-1}$ can be reduced to distinguishing between $\sys G^{\rro\qccar2,\ell}$ and $\sys G^{\rro\qccai2,\ell}$.
    Let $\sys D$ be a distinguisher for $\sys G^{\rro\qccar2,\ell}$ and $\sys G^{\rro\qccai2,\ell}$.
    Then we define an adversary $\sys D'=\sys D\sys C_I$, where $I$ is a uniform random variable over $\{1,\ldots,\ell\}$, with access to encryption oracle $\mathbf{Enc(\cdot)}$, decryption oracle $\mathbf{Dec(\cdot)}$, and challenge oracle $\mathbf{Chall(\cdot)}$, provided by $\sys G^{\rrc\qccar2,\ell}$/$\sys G^{\rrc\qccai2,\ell}$, where $\sys C_i$ internally keeps a set $\cM$ and exports oracles $\mathbf{Enc'(\cdot)}$ and $\mathbf{Dec'(\cdot)}$ towards $\sys D$ which behave as follows:
    \begin{itemize}
        \item Oracle $\mathbf{Enc'(\cdot)}$: on input the $j$-th message $\msg\reg M$:
        \begin{itemize}
            \item If $j<i$:
            \begin{enumerate}
                \item $\hat M\tilde M\gets\epr$.
                \item $\hat\ctx\reg{\hat C}\otimes\ketbra{\hat r}{\hat r}\reg R\gets{\bf Enc}(\hat\msg\reg{\hat M})$.
                \item $\cM\gets\cM\cup\{(\hat r,\tilde M,M)\}$.
                \item Return $\hat\ctx\reg{\hat C}\otimes\ketbra{\hat r}{\hat r}\reg R$.
            \end{enumerate}
            \item If $j=i$: return $\mathbf{Chall}(\msg\reg M)$.
            \item If $j>i$: return $\mathbf{Enc}(\msg\reg M)$.
        \end{itemize}
        \item Oracle $\mathbf{Dec'(\cdot)}$: on input ciphertext $\ctx\reg C\otimes\ketbra rr\reg R$:
        \begin{enumerate}
            \item $\hat\msg\reg{\hat M}\gets{\bf Dec}(\ctx\reg C\otimes\ketbra rr\reg R)$.
            \item If $\hat\msg\reg{\hat M}=\Bot$, return $\Bot$, otherwise proceed.
            \item For each $(\hat r,\tilde M,M)\in\cM$ do the following: if $\{\Pi_{\hat r},\One-\Pi_{\hat r}\}(\aux\reg T)\Out0$ and $\{\Pi_+,\One-\Pi_+\}(\varphi\reg{\hat M\tilde M})\Out0$, return $\msg\reg M$, otherwise proceed.
            \item Return $\Bot$.
        \end{enumerate}
    \end{itemize}
    Let now define the hybrid systems $\sys H_i$ resulting from attaching the reduction system $\sys C_i$ to the real $\RRC\QCCA2$ game, that is, define $\sys H_i\df\sys C_i\sys G^{\rrc\qccar2,\ell-1}$.
    Then observe that $\sys H_{i+1}\equiv\sys C_i\sys G^{\rrc\qccai2,\ell-1}$, $\sys H_1\equiv\sys G^{\rro\qccar2,\ell}$, and $\sys H_{\ell+1}\equiv\sys G^{\rro\qccai2,\ell}$.
    Finally, by the law of total probability, we have:
    \allowdisplaybreaks
    \begin{align*}
    &\Adv{\rrc\qcca2,\ell-1}\sch{\sys D\sys C}\\
    &\quad=\pr{}{\sys D\sys C[\sys G^{\rrc\qccar2,\ell-1}]=1}-\pr{}{\sys D\sys C[\sys G^{\rrc\qccai2,\ell-1}]=1}\\
    &\quad=\frac1\ell\sum_{i=1}^\ell\lp\pr{}{\sys D[\sys C_i\sys G^{\rrc\qccar2,\ell-1}]=1}-\pr{}{\sys D[\sys C_i\sys G^{\rrc\qccai2,\ell-1}]=1}\rp\\
    &\quad=\frac1\ell\sum_{i=1}^\ell\lp\pr{}{\sys D[\sys H_i]=1}-\pr{}{\sys D[\sys H_{i+1}]=1}\rp\\
    &\quad=\frac1\ell\cdot\lp\pr{}{\sys D[\sys H_1]=1}-\pr{}{\sys D[\sys H_{\ell+1}]=1}\rp\\
    &\quad=\frac1\ell\cdot\lp\pr{}{\sys D[\sys G^{\rro\qccar2,\ell}]=1}-\pr{}{\sys D[\sys G^{\rro\qccai2,\ell}]=1}\rp\\
    &\quad=\frac1\ell\cdot\Adv{\rro\qcca2,\ell}\sch{\sys D}.%\qedhere
    \end{align*}
    This concludes the proof.
\end{proof}

\begin{remark}
    \label{rem:rro-rrc}
    It is easy to show that the other direction of \autoref{lem:rrc-rro} also holds (for the same class of \sqes), that is, $\RRO\QCCA2$ implies $\RRC\QCCA2$.
    For this, the reduction $\sys C$ flips a bit $\tilde B$ and uses the $\RRO\QCCA2$ security game to emulate the $\RRC\QCCA2$ game, resulting in perfect emulation with probability $\frac12$, and perfect unguessability otherwise.
    Thus, with $\sys D\sys C$ outputting $1$ if and only if $\sys D$ correctly guesses $\tilde B$, we have $\Adv{\rrc\qcca2,\ell}\sch{\sys D}\leq2\cdot\Adv{\rro\qcca2,\ell-1}\sch{\sys D\sys C}$, and therefore the two notions are asymptotically equivalent, as we formalize in the following lemma.
\end{remark}

\begin{lemma}
    For \sqes\ satisfying \conref{con:restricted}, $\RRC\QCCA2$ and $\RRO\QCCA2$ are asymptotically equivalent.
\end{lemma}
\begin{proof}
    This follows directly by \autoref{lem:rrc-rro} and \autoref{rem:rro-rrc}.
\end{proof}

Just as we casted $\RRC\QCCA2$ into the classical definition $\RRC\CCA2$, we can cast $\RRO\QCCA2$ into $\RRO\CCA2$.
Then it is possible to obtain analogous results as above for the classical notions (without restrictions on the (classical) encryption scheme).

\begin{corollary}
    $\RRC\CCA2$ and $\RRO\CCA2$ are asymptotically equivalent.
\end{corollary}

\subsubsection{$\RRO\QCCA2$ Implies $\QCNF$.}

We can now finally relate $\QCCA2$ game-based security definitions to the constructive cryptography notion of confidentiality, $\QCNF$.
We do that by showing that $\RRO\QCCA2$ security implies $\QCNF$, and therefore, by \autoref{lem:rrc-rro}, so does $\RRC\QCCA2$ (with concrete security loss factor $\ell$).

\begin{theorem}
    \label{thm:qcca2}
    Let $\sch\df(\Gen,\Enc,\Dec)$ be a \sqes\ (implicit in all defined systems).
    Then with protocol $\pi^\qenc_{AB}=(\pi^\qenc_A,\pi^\qenc_B)$ making use of quantum computers $\qc^\ell_A$ and $\qc^\ell_B$ (already defined in \autoref{fig:enc-dec} for \autoref{thm:qae}), simulator $\simul^\qcca2_E$ making use of quantum computer $\qc^\ell_E$ as defined in \autoref{fig:sim} (until the end), and (trivial) reduction system $\sys C$ as specified in the proof, for any distinguisher $\sys D$ we have
    \[\Delta^{\sys D}(\pi^\qenc_{AB}\,[\key,\iqc^\ell,\qc^\ell_A,\qc^\ell_B],\simul^\qcca2_E\,[\nmqc^\ell,\qc^\ell_E])\leq\Adv{\rro\qcca2,\ell}\sch{\sys D\sys C}.\]
\end{theorem}
\begin{proof}
    Let $\sys R\df\pi^\qenc_{AB}\,[\key,\iqc^\ell,\qc^\ell_A,\qc^\ell_B]$ and $\sys S\df\simul^\qae_E\,[\nmqc^\ell,\qc^\ell_E]$.
    We need to provide a reduction system $\sys C$ so that distinguishing between the games $\sys G^{\rro\qccar2,\ell}$ and $\sys G^{\rro\qccai2,\ell}$ can be reduced to distinguishing between resources $\sys R$ and $\sys S$.
    Let $\sys D$ be a distinguisher for $\sys R$ and $\sys S$.
    Then we define an adversary $\sys D'=\sys D\sys C$ with access to real/ideal encryption oracle $\mathbf{Enc(\cdot)}$ and real/ideal decryption oracle $\mathbf{Dec(\cdot)}$ provided by $\sys G^{\rro\qccar2,\ell}$/$\sys G^{\rro\qccai2,\ell}$, where $\sys C$ behaves as follows towards $\sys D$:
    \begin{itemize}
        \item {\bf Interface $\iA$:} On input message $\msg$, output $\mathbf{Enc}(\msg)$ at interface $\iE$.
        \item {\bf Interface $\iE$:} On input state $\ctx$, output $\mathbf{Dec}(\ctx)$ at interface $\iB$.
    \end{itemize}
    Note that we trivially have that $\sys C\sys G^{\rro\qccar2,\ell}\equiv\sys R$ and $\sys C\sys G^{\rro\qccai2,\ell}\equiv\sys S$, hence
    \begin{align*}
    \Adv{\rro\qcca2,\ell}\sch{\sys D\sys C}
    &=\pr{}{\sys D\sys C[\sys G^{\rro\qccar2,\ell}]=1}-\pr{}{\sys D\sys C[\sys G^{\rro\qccai2,\ell}]=1}\\
    &=\pr{}{\sys D[\sys C\sys G^{\rro\qccar2,\ell}]=1}-\pr{}{\sys D[\sys C\sys G^{\rro\qccai2,\ell}]=1}\\
    &=\pr{}{\sys D[\sys R]=1}-\pr{}{\sys D[\sys S]=1}\\
    &=\Delta^{\sys D}(\sys R,\sys S).%\qedhere
    \end{align*}
    This concludes the proof.
\end{proof}

%\ifsub\else
\begin{corollary}
    \label{cor:qcca2}
    With $\e(\sys D):=\sup_{\sys D'\in\mathcal B(\sys D)}\Adv{\rro\qcca2,\ell}\sch{\sys D'}$, we have
    \begin{equation*}
        \left[\key,\iqc^\ell,\qc^\ell_A,\qc^\ell_B\right]
        \xrightarrow{\pi^{\qenc}_{AB},\eps}
        \left[\nmqc^\ell,\qc^{\qcca2,\ell}_E\right],
    \end{equation*}
    where the class $\mathcal B(\sys D)$ is defined in \eqnref{eq:distinguisher.class}, \autoref{sec:ac.theory}.
\end{corollary}

Using \autoref{lem:rrc-rro}, we finally obtain the following corollary.

\begin{corollary}
    \label{cor:qcca2rrc}
    With $\e(\sys D):=\sup_{\sys D'\in\mathcal B(\sys D)}\Adv{\rrc\qcca2,\ell}\sch{\sys D'}$, we have
    \begin{equation*}
        \left[\key,\iqc^\ell,\qc^\ell_A,\qc^\ell_B\right]
        \xrightarrow{\pi^{\qenc}_{AB},(\ell+1)\cdot\eps}
        \left[\nmqc^\ell,\qc^{\qcca2,\ell}_E\right],
    \end{equation*}
    where the class $\mathcal B(\sys D)$ is defined in \eqnref{eq:distinguisher.class}, \autoref{sec:ac.theory}.
\end{corollary}
%\fi

\subsubsection{$\RRO\QCCA2$ is Stronger than $\QCNF$.}

We remark that even though the security notion $\RRO\QCCA2$ implies $\QCNF$, the converse is not true for the same reason outlined above for $\QAE$ and $\QSEC$: it is possible to show that there are \sqes s that satisfy $\QCNF$ but not $\RRO\QCCA2$ by applying the same principle of extending a $\RRO\QCCA2$ secure scheme into one which is not anymore $\RRO\QCCA2$, but still satisfies $\QCNF$.

%%% Local Variables:
%%% TeX-master: "qccFull"
%%% End:
    \ifsub
    \else
    \section*{Acknowledgments}
\pdfbookmark[1]{Acknowledgments}{sec:ack}
% \addcontentsline{toc}{section}{Acknowledgments}

CP acknowledges support from the Zurich Information Security
and Privacy Center.

%%% Local Variables:
%%% TeX-master: "qcc"
%%% End:
    \fi
    %\newpage
    \appendix
    \section{Notation}
\label{app:notation}

%\subsection{Quantum States and Operators}
\paragraph{Pauli Operators.}
We write $P_k$ or $P_{x,z}$ to denote a Pauli operator on $m$ qubits, where $k =(x,z)$ are concatenation of two $m$-bits strings indicating in which qubit bit flips and phase flips occur. 
$$P_k = P_{x,z} =  \bigotimes_{i=1}^m P_{x_iz_i},$$
where 
$$P_{ab}  = 
\begin{cases}
               I &  a=0, b=0,\\
              X&  a=1, b=0,\\
              Z&  a=0, b=1,\\
              XZ&  a=1, b=1.\\
            \end{cases}
$$
Note that $P_k = P_k^{\dagger}$, therefore we simply write $P_k\rho P_k$ when applying a Pauli-operator $P_k$ on state $\rho$. To undo Pauli-operator $P_k$, we simply apply $P_k$ again, namely, $P_kP_k\rho P_kP_k = \rho$. 
\paragraph{Bell Basis.}
We write $\ket{\phi_0}$ as the maximum entangled state of $2m$ qubits.
$$\ket{\phi_0} = \left(\frac{\ket{00} + \ket{11}}{\sqrt{2}}\right)^{\otimes m}$$
and $\ket{\phi_k}$ as applying Pauli operator $I^{\otimes m} \otimes P_k$ on $\ket{\phi_0}$ 
$$\ket{\phi_k} = I^{\otimes m} \otimes P_k \ket{\phi_0}.$$
Then $\{\ket{\phi_k}\}_{k \in \{0,1\}^{2m}}$ forms the Bell basis for $2m$ qubits. 
\paragraph{Symplectic Inner Product.}
Symplectic inner product of two string $k=(x,z)$ and $l=(x',z')$ of
length $2m$ is given by
\[\spl(k,l)=\sum_{i=1}^m x_iz_i' - x_i'z_i.\] 

% \subsubsection{Technical Lemmas}
% \theoremstyle{theorem}
% \newtheorem*{lemma*}{Lemma}
% We state several technical lemmas of Pauli operators and symplectic inner product in the following:
% \begin{lemma*}
% $P_k^\dagger$ = $P_k$.
% \end{lemma*}

% \begin{lemma*}
% For any mixed state $\rho^{AB}$ held in register $\rA$ and $\rB$,
% $$\frac{1}{2^{2m}}\sum_{k \in \{0,1\}^{2m}}P_k^A\rho^{AB}P_k^A = \frac{1}{2^m}I^A \otimes \rho^B$$
% \end{lemma*}

Two Pauli operators $P_k$ and $P_l$ with $k=(x,z)$ and $l=(x',z')$ commute (anti-commute) if  $\spl(k,l)$
equals to 0 (equals to 1) module 2. Hence for any $P_k$ and $P_l$,
\[P_kP_l = (-1)^{\spl(k,l)}P_lP_k.\]

For any strings $k$ and $l$ we have,
$$\sum_{k \in \{0,1\}^m} (-1)^{\spl(k,l)}  = 
\begin{cases}
               2^n & l=0,\\
               0 & l \neq 0,\\
            \end{cases}
$$
where $l=0$ means all bits of the string $l$ are $0$.

%\subsection{Converters and Resources}

%\input{appendix-background.tex}

\section{Composition}
\label{app:composition}

In this section we prove that \eqnsref{eq:ac.composition1} and
\eqref{eq:ac.composition2} hold for the construction notion from \autoref{def:security}.

\begin{theorem}[Composition theorem]\label{thm:composition}
Suppose that $\pi_1$ constructs $\sys S$ from $\sys R$ within $\eps_1$
and $\pi_2$ constructs $\sys T$ from $\sys S$ within $\eps_2$, i.e.,
\[ \sys R \xrightarrow{\pi_1,\eps_1} \sys S \qquad \text{and} \qquad
  \sys S \xrightarrow{\pi_2,\eps_2} \sys T.\]
Then it holds that
\[ \sys R \xrightarrow{\pi_2\pi_1,\eps_1+\eps_2} \sys T.\]
Furthermore, for any resource $\sys U$, we have
\[ \left[\sys R, \sys U\right] \xrightarrow{\pi_1,\eps'} \left[\sys S,
  \sys U \right],\] with
\[\eps'(\sys D) \coloneqq \eps_1\left(\sys D\left[ \cdot ,\sys
      U\right]\right).\]
\end{theorem}

\begin{proof}
  Let $\simul$ and $\simul'$ be the simulators for which for all $\sys D$,
  $d^{\sys D}(\pi_1 \sys R,\sys S \simul) \leq \eps(\sys D)$ and
  $d^{\sys D}(\pi_2 \sys S,\sys T \simul') \leq \eps(\sys D)$. Then from
  the triangle inequality we have
  \begin{align*}
    d^{\sys D}(\pi_2 \pi_1 \sys R, T \simul' \simul)
       & \leq d^{\sys D}(\pi_2 \pi_1 \sys R, \pi_2 \sys S \simul) + d^{\sys D}(\pi_2 S \simul, T \simul' \simul)\\
       & = d^{\sys D}(\pi_1 \sys R,\sys S \simul) + d^{\sys D}(\pi_2 S, T \simul')\\
       & \leq \eps_1(\sys D) + \eps_2(\sys D).
  \end{align*}
  Furthermore
  \begin{align*}
    d^{\sys D}\left( \pi_1 \left[\sys R,\sys U\right], \left[\sys S,\sys U \right] \simul\right)
    & = d^{\sys D}\left( \left[\pi_1 \sys R,\sys U\right], \left[\sys S \simul,\sys U \right]\right)\\
    & \leq d^{\sys D \left[ \cdot , \sys U\right]}\left(\pi_1 \sys R,
      \sys S \simul\right) \leq \eps_1\left(\sys D\left[ \cdot ,\sys
      U\right]\right). \qedhere      
  \end{align*}
\end{proof}

\section{Quantum One-Time Pad}
\label{app:example.otp}

Here we show that the quantum one-time pad constructs a one-time Pauli\-/malleable 
confidential quantum channel $\overline{\xqc}^{1,m}$ from a
one-time insecure quantum channel $\overline{\iqc}^{1,m}$ and a key
resource $\key^{2m}$. Note that $\overline{\xqc}$ and
$\overline{\iqc}$ are a slightly modified versions of $\xqc$ and
$\iqc$. Which work as follows: after getting an input at interface
$\iE$, the channels reject any input at interface $\iA$. A weaker
statement has appeared in \cite{DFPR14}, where it was proven that the
one-time pad constructs a confidential channel \--- i.e., one which
allows the adversary to make arbitrary changes, but only leaks the
message size \--- which is strictly weaker than the resource
$\overline{\xqc}^{1,m}$ proven to be constructed here.

The protocol is described in
\autoref{fig:ic.xqc.one.protocol}. The construction also gives Eve
certain computation power. We present the simulator $\simul^\qconf$
and the quantum computer of Eve $\qc^\qconf$ in
\autoref{fig:ic.xqc.simulator} such that the ideal channel connecting
with the simulator and Eve's quantum computer can be proved equivalent
to the system of Alice's and Bob's converters, quantum computers
connecting with the insecure channel and the key resource.

% In the lemma, we use a slightly different definition of the one message insecure quantum channel , denoted as $\overline{\iqc}^{1,m}$.  Without such limitation on the channel, the distinguisher in the real world can input at interface $\iE$ an arbitrary state $\sigma$ and then forward the output $\state = P_k\sigma P_k$  from interface $\iB$ as input to interface $\iA$, the distinguisher will get the same state back at interface $\iE$, since $\tilde{\sigma} = P_k\state P_k = P_kP_k\sigma P_kP_k = \sigma$.  In the ideal world, this cannot happen, as a fully mixed state will be output at interface $\iE$.  This limitation will not be needed in the multi-message channel.

\begin{figure}[tb]
    \centering
    \begin{convbox}{\textwidth}{$\pi_A \quad \pi_B \quad \qc_A \quad \qc_B$}
		\begin{enumerate}
			\item $\pi_A, \pi_B$: Alice and Bob obtain uniform key $k$ from key resource $\key$.
			\item $\pi_A$: Alice inputs message $\state^{\rA}$ and requests her computer to encrypt $\state^{\rA}$ with $P_k$.
			\item $\qc_A$: On input $\state^{A}$, the computer apply Pauli operator $P_k$ on the state $\state^{A}$ and outputs $\sigma =P_k \state^{A} P_k $ to Alice.
			\item $\pi_A$: Alice sends ciphertext $\sigma$ to Bob through insecure quantum channel $\overline{\iqc}$.
			\item $\pi_B$: Bob receives ciphertext $\tilde{\sigma}$ and requests his computer to decrypt $\tilde{\sigma}$ with $P_k$.
			\item $\qc_B$: On input $\tilde{\sigma}$, the computer apply Pauli operator $P_k$ on the state $\tilde{\sigma}$.  Then the computer outputs $\tilde{\state} = P_k \tilde{\sigma} P_k $ to Bob.
			\item $\pi_B$: Bob outputs the decrypted message $\tilde{\rho}^{\rB}$.
		\end{enumerate}
  
    \end{convbox}
    \caption{Converters and computer resources to construct $\overline{\xqc}^{1,m}$ from $\overline{\iqc}^{1,m}$. The quantum computer $\qc^{1, m, m}_A$ and $\qc^{1, m, m}_B$  will be requested 1 time. The plaintext and ciphertext both have length $m$. $\key^{2m}$ gives a shared key of length $2m$.}\label{fig:ic.xqc.one.protocol}
    
	\begin{convbox}{\textwidth}{$\simul^\qconf \quad \qc^\qconf$}
		\begin{enumerate}
		\item $\simul^\qconf$: On input $\lNewMsg$ from the ideal channel $\overline{\xqc}^{1,m}$, it requests $\qc^\qconf$ to provide a quantum state and outputs the state at interface $\iE$.  
		\item $\simul^\qconf$: On input $\tilde{\sigma}$ at interface $\iE$, it requests $\qc^\qconf$ to provide a key $k$ of length $2m$ (or $\bot$).  The key (or $\bot$) is then forwarded to ideal channel $\overline{\xqc}^{1,m}$.
		\item $\qc^\qconf$: On input $\lNewMsg$ from $\simul^\qconf$, the computer prepares a maximum entangled state $\ketbra{\phi_0}{\phi_0}^{MC}$ in register $\rM$ and $\rC$. The computer  stores $\state^M$ in register M in its internal memory and outputs the state $\state^C$ in register C to $\simul_E$.
		\item $\qc^\qconf$ : On input $\tilde{\rho}^C$ from interface $\iE$, if $\state^M$ is stored internally, the computer measures $\tilde{\state}^{MC}$ in the bell basis and gets measurement result $\tilde{k}$. The computer outputs $\tilde{k}$ to $\simul^\qconf$. If no $\state^M$ is stored, then no $\lNewMsg$ has arrived yet, the computer outputs $\bot$ to $\simul^\qconf$.
		\end{enumerate}
    \end{convbox}
    \caption{The simulator $\simul^\qconf$ and computer resource $\qc^\qconf$ connecting to the ideal resource $\overline{\xqc}^{1,m}$. $\qc^\qconf$ is capable of doing 1 encryption and 1 decryption.}
    \label{fig:ic.xqc.simulator}
\end{figure}

\begin{lemma} \label{lem:ic.xqc.one}
Let $\pi_{AB} = (\pi_A, \pi_B), \qc^{1, m,m}_A, \qc^{1, m, m}_B, \simul^\qconf, \qc^\qconf$ denote the converters, simulator and computing resources described in \autoref{fig:ic.xqc.one.protocol} and \autoref{fig:ic.xqc.simulator}. Then for any distinguisher $\sys D$, we have
\begin{equation*}
d^{\sys D}\left(\pi_{AB}\left[\overline{\iqc}^{1,m},\key^{2m},	 \qc^{1,m,m}_A, \qc^{1,m,m}_B\right], \simul^\qconf\left[\overline{\xqc}^{1,m}, \qc^\qconf\right]\right) =0.
\end{equation*}
\end{lemma}
\begin{proof}
    We first consider the simpler case that distinguisher only inputs at interface $\iE$. In both worlds, the distinguisher prepares an arbitrary state $\rho^{AE}$ and inputs $\rho^A$ at interface $\iE$.  In the real world, the protocal applies a Pauli operator $P_k$ on the state and then output at interface $\iB$.  The distinguisher has state 
\[ \Phi_{real} = \frac{1}{2^{2m}} \sum_{k \in \{0,1\}^{2m}} P^A_k \rho^{AE} P^A_k.\]
In the ideal world, since no $\lNewMsg$ has arrived yet, the simulator output $\bot$ to channel $\overline{\xqc}^{1,m}$ and the channel outputs a fully mixed state $\frac{1}{2^m}I$ at interface B.  The distinguisher has state
\[\Phi_{ideal} = \frac{1}{2^m}I \otimes \rho^E\]
It is easy to see that for any mixed state $\rho^{AE}$,
\[\frac{1}{2^{2m}} \sum_{k \in \{0,1\}^{2m}} P^A_k \rho^{AE} P^A_k = \frac{1}{2^m}I \otimes \rho^E\] 
and therefore we have
\[\Phi_{real} = \Phi_{ideal}.\]

    Now we consider the more complicated case that distinguisher first input at interface $\iA$ and then at interface $\iE$. In both worlds, the distinguisher prepares an arbitrary state $\rho^{AE}$ and a CPTP map $\Lambda$. In the real world, the distinguisher inputs $\state^A$ at interface $\iA$. After getting the flipped (Alice applying a Pauli operator $P_k$) state $P_k\state^{A} P_k$ at interface $\iE$ , the distinguisher applies the CPTP map $\Lambda$ on register $\rA$ and $\rE$. Then the distinguisher inputs the result state $\tilde{\rho}^A$in register $\rA$ into the insecure channel $\overline{\iqc}^{1,m}$.  Then the distinguisher gets the flipped (Bob applying a Pauli operator $P_k$) state output at interface $\iB$. The distinguisher holds state $\Phi_{real}$
    $$\Phi_{real} = \frac{1}{2^{2m}}\sum_{k \in \{0,1\}^{2m}}P_{k}^{A}\Lambda^{AE}(P_k^A\state^{AE} P_k^{A}) P_{k}^A. $$
    
    In the ideal world, the distinguisher inputs $\state^A$ at interface $\iA$. The simulator prepares a maximum entangled state $\ketbra{\phi_0}{\phi_0}^{MC}$ in register $\rM$ and $\rC$ and outputs the state $\sigma^C$ in register $\rC$ at interface $\rE$. The distinguisher applies the CPTP map $\Lambda$ on register $\rC$ and $\rE$. Then the distinguisher inputs the result state $\tilde{\sigma}^C$ in register $\rC$ into the channel. The simulator measures the state in register $M$ and $C$ on bell basis, get the measurement result $k$ and the channel applies Pauli operator $P_k$ on the state in register $A$. After getting the modified state output from interface $\iB$, the distinguisher holds state $\Phi_{ideal}$.
    $$\tilde{\sigma}^{AMCE} = \Lambda^{CE}(\proj{\phi_0}^{MC} \otimes \state^{AE})$$
    $$\Phi_{ideal} = \Tr_{MC}\sum_{k \in \{0,1\}^{2m}} (P_k^{A} \otimes \proj{\phi_k}^{MC})\tilde{\sigma}^{AMCE} (P_k^{A} \otimes \proj{\phi_k}^{MC}).$$
    
    We prove that for any $\state^{AE}$ and CPTP map $\Lambda^{XE}$, where register $\rX$ is $\rA$ in the real world and $\rC$ in the ideal world,  we have $\Phi_{real} = \Phi_{ideal}$.
    
    For any CPTP map $\Lambda^{XE}$ applying on state $\rho^{XE}$ on register $\rX$ and $\rE$, we can decompose the map as 
    \begin{align} \label{eq:lambda decomposition}
    \Lambda^{XE}(\rho^{XE}) &= \sum_i \Gamma_i^{XE}\rho^{XE}\Gamma_i^{XE\dagger} \nonumber \\ 
    &= \sum_i (\sum_{p}P^X_{p} \otimes \Gamma^E_{i,p}) \rho^{XE} (\sum_{q}P^X_{q} \otimes \Gamma^{E\dagger}_{i,q}) \nonumber \\
    &= \sum_{i,p, q} (P^X_{p} \otimes \Gamma^E_{i,p}) \rho^{XE} (P^X_{q} \otimes \Gamma^{E\dagger}_{i,q})
    \end{align} 
    
    Plugging \eqnref{eq:lambda decomposition} into $\tilde{\sigma}^{AMCE}$, we get
    \begin{align*}
    \tilde{\sigma}^{AMCE} &= \Lambda^{CE}(\proj{\phi_0}^{MC} \otimes \state^{AE} )\\
    &= \sum_{i,p, q} (P^C_{p} \otimes \Gamma^E_{i,p}) (\proj{\phi_0}^{MC} \otimes \state^{AE})  (P^C_{q} \otimes \Gamma^{E\dagger}_{i,q})\\
    &= \sum_{i,p, q} P^C_{p} \proj{\phi_0}^{MC} P^C_{q} \otimes( \Gamma^E_{i,p}  \state^{AE} \Gamma^{E\dagger}_{i,q})\\
    &= \sum_{i,p, q} \ketbra{\phi_{p}}{\phi_{q}}^{MC} \otimes( \Gamma^E_{i,p}  \state^{AE} \Gamma^{E\dagger}_{i,q}),
    \end{align*}
where $\ket{\phi_p}^{MC} = P^C_{p} \ket{\phi_0}^{MC}$, see \aref{app:notation}.

    Now we can express $\Phi_{ideal}$ as 
    \begin{align*}
    \Phi_{ideal} & = \Tr_{MC}\sum_{k,i,p,q} (P_k^{A} \otimes \proj{\phi_k}^{MC})\tilde{\sigma}^{AMCE} (P_k^{A} \otimes \proj{\phi_k}^{MC})  \\
    &= \Tr_{MC}\sum_{k,i,p,q} (P_k^{A} \otimes \proj{\phi_k}^{MC})	(( \Gamma^E_{i,p}  \state^{AE} \Gamma^{E\dagger}_{i,q}) \otimes \ketbra{\phi_{p}}{\phi_{q}}^{MC}) (P_k^{A} \otimes \proj{\phi_k}^{MC}) \\
    &= \Tr_{MC}\sum_{k,i,p,q} (P_k^{A} \otimes \Gamma^E_{i,p} \rho^{AE} P_k^{A} \otimes \Gamma^{E\dagger}_{i,q})  \otimes (\proj{\phi_k}^{MC}  \ketbra{\phi_{p}}{\phi_{q}}^{MC}  \proj{\phi_k}^{MC}) \\
    &= \Tr_{MC}\sum_{k,i} (P_k^{A} \otimes \Gamma^E_{i,k} \rho^{AE} P_k^{A} \otimes \Gamma^{E\dagger}_{i,k})  \otimes \proj{\phi_k}^{MC} \\
    &=\sum_{k,i} P_k^{A} \otimes \Gamma^E_{i,k} \rho^{AE} P_k^{A} \otimes \Gamma^{E\dagger}_{i,k}
    \end{align*}
    
    Plugging \eqnref{eq:lambda decomposition} into $\Phi_{real}$, we get
    \begin{align*}
    \Phi_{real} &= \frac{1}{2^{2m}}\sum_{k}P_{k}^{A}\Lambda^{AE}(P_k^A\state^{AE} P_k^{A}) P_{k}^A  \\
    &= \frac{1}{2^{2m}}\sum_{k}P_{k}^{A}(\sum_{i,p, q} (P^A_{p} \otimes \Gamma^E_{i,p}) P_k^A\state^{AE} P_k^{A} (P^A_{q} \otimes \Gamma^{E\dagger}_{i,q})) P_{k}^A  \\
    &= \frac{1}{2^{2m}}\sum_{k,i,p,q}(P_{k}^{A}P^A_{p}P^A_{k} \otimes \Gamma^E_{i,p}) \state^{AE} (P_k^{A} P^A_{q}P^A_{k} \otimes \Gamma^{E\dagger}_{i,q}).
    \end{align*}
    
    Note that $P_kP_p = (-1)^{\spl(k,p)}P_pP_k$ (see \aref{app:notation}), so we have
    \begin{align*}
    \sum_{k,p,q}P_{k}P_{p}P_{k} (\cdot) P_k P_{q}P_{k}
    &=  \sum_{k,p,q}(-1)^{\spl(k,p)}P_{p}P_{k}P_{k} (\cdot) (-1)^{\spl(k,q)} P_q P_{k}P_{k}  \\
    &= \sum_{k,p,q}(-1)^{\spl(k,p)+\spl(k,q)}P_p (\cdot) P_q  \\
    &= \sum_{p,q} P_p (\cdot) P_q \sum_k(-1)^{\spl(k,p \oplus q)}  \\
    &= \sum_{p \oplus q =0} P_p (\cdot) P_q \cdot 2^{2m} \\
    &= 2^{2m} \sum_{p} P_p (\cdot) P_p. 
    \end{align*} 
    Therefore we can further simplify $\Phi_{real}$
    \begin{align*}
    \Phi_{real} 
    &= \frac{1}{2^{2m}}\sum_{k,i,p,q}(P_{k}^{A}P^A_{p}P^A_{k} \otimes \Gamma^E_{i,p}) \state^{AE} (P_k^{A} P^A_{q}P^A_{k} \otimes \Gamma^{E\dagger}_{i,q})  \\
    &= \frac{1}{2^{2m}}2^{2m}\sum_{p,i}(P_{p}^{A} \otimes \Gamma^E_{i,p}) \state^{AE} (P_p^{A} \otimes \Gamma^{E\dagger}_{i,p})  \\
    &= \sum_{p,i}P_{p}^{A} \otimes \Gamma^E_{i,p} \state^{AE} P_p^{A} \otimes \Gamma^{E\dagger}_{i,p}.
    \end{align*}
    Therefore for any state $\rho^{AE}$ and CPTP map $\Lambda^{XE}$, $\Phi_{real} = \Phi_{ideal}$.
\end{proof}

\section{Constructing $\osqc^{\ell, m}$ from $\sqc^{\ell,m+\log\ell}$}
\label{app:osc}

In this section we present a construction from a secure quantum channel to an ordered secure quantum channel. In the protocol, Alice appends to each message an index. When Bob receives the message, he  checks the index. He accepts the message only when the index is expected. The protocol works as following.

\begin{figure}[tb]
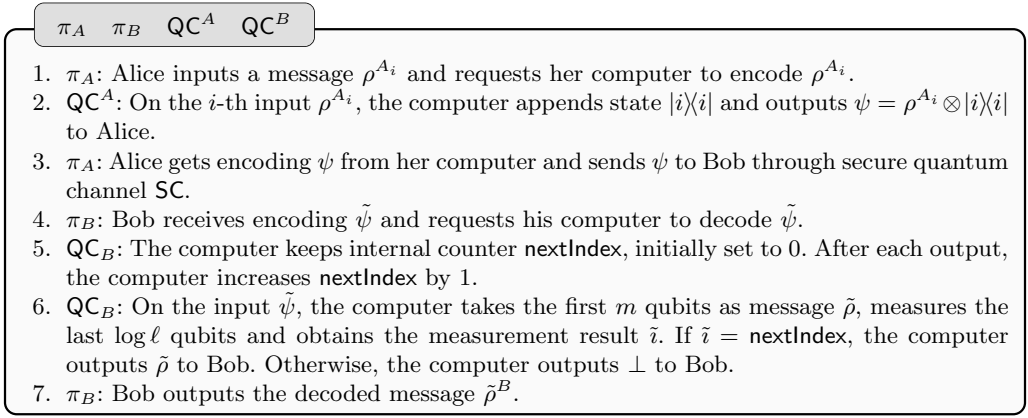
 
    \centering
    \begin{convbox}{\textwidth}{$\pi_A \quad \pi_B \quad \qc^A \quad \qc^B$}
        \begin{enumerate}		
            \item $\pi_A$: Alice inputs a message $\state^{\rA_i}$ and requests her computer to encode $\state^{\rA_i}$.
            \item $\qc^A$: On the $i$-th input $\state^{\rA_i}$, the computer appends state $\proj{i}$ and outputs $\psi=\state^{\rA_i} \otimes \proj{i}$ to Alice.
            \item $\pi_A$: Alice gets encoding $\psi$ from her computer and sends $\psi$ to Bob through secure quantum channel $\sqc$.
            \item $\pi_B$: Bob receives encoding $\tilde{\psi}$ and requests his computer to decode $\tilde{\psi}$.
            \item $\qc_B$: The computer keeps internal counter $\vNextIndex$, initially set to 0. After each output, the computer increases $\vNextIndex$ by 1. 
            \item $\qc_B$: On the input $\tilde{\psi}$, the computer takes the first $m$ qubits as message $\tilde{\state}$, measures the last $\log\ell$ qubits and obtains the measurement result $\tilde{\imath}$. If $\tilde{\imath} = \vNextIndex$, the computer outputs $\tilde{\state}$ to Bob. Otherwise,  the computer outputs $\bot$ to Bob.
            \item $\pi_B$: Bob outputs the decoded message $\tilde{\rho}^{\rB}$.
        \end{enumerate}
        
    \end{convbox}
    \caption{Converters and computing resources to construct $\osqc^{\ell, m}$ from $\sqc^{\ell,m+\log\ell}$.  $\qc^{\ell, m, m+\log\ell}_A$ and $\qc^{\ell, m, m+\log\ell}_B$  will be requested $\ell$ times. The message has length $m$ and encoding has length $m + \log \ell$. }. 
    \label{fig:sc.osc.protocol}
\end{figure}

%\todomarginJiamin{Describing $\qc_E$ in theorem. Memory? Computing power?  Randomness?}

\begin{theorem} \label{thm:iqc.sqc}
    Let $\pi_{AB}=(\pi_A, \pi_B)$, $\qc^{\ell, m, m+\log\ell}_A$ and $\qc^{\ell, m, m+\log\ell}_B$ denote converters and computing resources corresponding to the protocol from \autoref{fig:sc.osc.protocol}. Let $\qc^\ell_E$ be a computing resource for Eve, capable of doing $\ell$ add operations.  Then $\pi_{AB}$ constructs an ordered secure quantum channel $\osqc^{\ell, m}$ from a secure quantum channel $\sqc^{\ell,m+\log\ell}$ with $\epsilon =0$, i.e.,
    \[\left[\sqc^{\ell,m+\log\ell}, \qc^{\ell, m, m+\log\ell}_A , \qc^{\ell, m, m+\log\ell}_B \right]
    \xrightarrow{\pi_{AB},0}
    \left[\osqc^{\ell, m}, \qc^{\ell}_E\right].\]
\end{theorem}

\begin{proof} 
    In order to prove this theorem, we need to find a simulator and a program for Eve's computer such that the real system and ideal system are indistinguishable with 0 advantage. 
    
    The simulator $\simul_E$ works as following. On input $\lNewMsg$ from the ideal resource $\osqc^{\ell,m}$, $\simul_E$ outputs $\lNewMsg$ at interface $\iE$. On input $i$ at interface $\iE$, $\simul_E$ inputs $i$ to $\qc_E^{\ell}$, requests it to output a $\lSend$ or $\lSkip$ signal and forwards the signal to ideal resource $\osqc^\ell$.
    
    Eve's computer $\qc_E^{\ell}$ works as following. The computer keeps an internal counter $\vNextIndex$, initially set to 0. On input $i$ from $\simul_E$, if  $i = \vNextIndex$, it outputs $\lSend$ to $\simul_E$, otherwise, it outputs $\lSkip$ to $\simul_E$. After the output, it increases $\vNextIndex$ by 1. The computer will be requested at most $\ell$ times.
    
    One can see that the real system $\pi_{AB}\left[\sqc^{\ell,m+\log\ell}, \qc^{\ell, m, m+\log\ell}_A , \qc^{\ell, m, m+\log\ell}_B \right]$ and the ideal world $\simul_E\left[\osqc^{\ell, m}, \qc^{\ell}_E\right]$ are two equivalent systems. If the distinguisher reorders the messages, in both worlds, Bob will only receive messages with consecutive indices. Therefore, two systems are equivalent and the distinguishing advantage is 0.
\end{proof}

\begin{remark*}
    \autoref{thm:iqc.sqc} is meaningful only if the protocol
    also provides correctness. This is trivially the case, since if
    the  distinguisher is honest, i.e., always preserves the order of
    messages, then Bob will receive them all in the correct order.
\end{remark*}

%%% Local Variables:
%%% TeX-master: "qccFull"
%%% End:

%%%%%%%%%%%%%%%%%%%%%%%%%%%%%%%%%%%%%%%%%%%%%%%%%%%%%%%%%%%%%%

%\section{Relating Distinction and Bit-Guessing Problems}
%\label{app:dist-guess}
%
%\begin{lemma}
%    \label{lem:dist-guess}
%    Let $\sys D$ be a distinguisher for systems $\sys S$ and $\sys T$.
%    Further, let $\sys R$ be a system which internally defines a uniform random bit $\beta(\sys R)$,
%    and which is such that if $\beta(\sys R)=1$, it behaves exactly as $\sys S$, whereas if $\beta(\sys R)=0$, it behaves exactly as $\sys T$.
%    Then we have that $\Lambda^{\sys D}(\sys R)=\Delta^{\sys D}(\sys S,\sys T)$.
%\end{lemma}
%\begin{proof}
%    \begin{align*}
%        \Lambda^{\sys D}(\sys R)
%        &=2\cdot\pr{}{\sys D[\sys R]=\beta(\sys R)}-1\\
%        &=2\cdot\pr{}{\sys D[\sys R]=\beta(\sys R)\St\beta(\sys R)=1}\cdot\pr{}{\beta(\sys R)=1}+\\
%        &\hspace{6mm}+2\cdot\pr{}{\sys D[\sys R]=\beta(\sys R)\St\beta(\sys R)=0}\cdot\pr{}{\beta(\sys R)=0}-1\\
%        &=2\cdot\pr{}{\sys D[\sys S]=1}\cdot\frac12+2\cdot\pr{}{\sys D[\sys T]=0}\cdot\frac12-1\\
%        &=\pr{}{\sys D[\sys S]=1}-(1-\pr{}{\sys D[\sys T]=0})\\
%        &=\pr{}{\sys D[\sys S]=1}-\pr{}{\sys D[\sys T]=1}\\
%        &=\Delta^{\sys D}(\sys S,\sys T).\qedhere
%    \end{align*}
%\end{proof}

\section{$\AGM\QCCA2$ Security Definition (\cite{alagic2018unforgeable}).}
\label{app:agm}

Here we restate what it means for a \sqes\ $\sch$ to be secure in the $\AGM\QCCA2$ sense according to \cite{alagic2018unforgeable}.
Contrary to the $\QAE$ definition, the authors did not formulate $\AGM\QCCA2$ in terms of a distinction problem (between a real and an ideal setting); rather, the interaction of a distinguisher $\sys D$ with two different games is considered:
\begin{itemize}
    \item The first offers (true) encryption and decryption oracles, as well as a challenge phase: $\sys D$ is required to input a plaintext, for which it will get back either the true encryption, or the encryption of a random plaintext ($\tau\reg M$, half of a maximally-entangled state).
    $\sys D$ wins the game if it guesses which of the two states was actually encrypted.
    After the challenge phase, $\sys D$ has still access to the same oracles, and could therefore in principle trivially win the game by submitting the challenge ciphertext and compare the result with the challenge plaintext;
    \item The second game countermeasures this by exposing towards $\sys D$ the same game, but where the challenge plaintext is always replaced by $\tau\reg M$, and if $\sys D$ tries to cheat by submitting the challenge ciphertext to the decryption oracle, it instantly loses the game (and wins it with probability $\frac12$ otherwise).
\end{itemize}
Then the advantage of $\sys D$ is measured as the probability of winning the first game minus the probability of cheating in the second.
The rationale behind this definition is that $\sys D$ should not be able to have a larger probability of winning the first game than of cheating in the second.
We now state our adaption of the corresponding definition from \cite{alagic2018unforgeable} (Definition 9 therein).

\begin{definition}[$\AGM\QCCA2$ Security \cite{alagic2018unforgeable}]
    \label{def:agm}
    For \sqes\  (implicit in all defined systems) $\sch\df(\Gen,\Enc,\Dec)$, we define the \emph{$\AGM\QCCA2$-advantage} of $\sch$ for distinguisher $\sys D\in\bD$ as
    \[\Adv{\agm\qcca2}\sch{\sys D}\df\pr{}{\sys G^{\agm\qccat2}[\sys D]=1}-\pr{}{\sys G^{\agm\qccaf2}[\sys D]=1},\]
    where the interactions of $\sys D$ with game systems $\sys G^{\agm\qccat2}$ and $\sys G^{\agm\qccaf2}$ are defined in \autoref{fig:qcca2}.
\end{definition}

\begin{figure}[htb]
    \begin{tcbraster}[raster columns=1, sidebyside, lefthand width=.297\textwidth]
        \begin{algobox}{\textwidth}{Experiments ${\sys G^{\agm\qccat2}[\sys D]}$ and ${\sys G^{\agm\qccaf2}[\sys D]}$ for \sqes\ ${\sch\df(\Gen,\Enc,\Dec)}$}
            \begin{algorithmic}
                \State $b\uar\bin$
                \State $k\gets\Gen()$
                \State $\msg\reg M\gets\sys D^{\Enc_k(\cdot),\Dec_k(\cdot)}$
                \If{$b=1$}
                \State $\ctx\reg C\gets\Enc_k(\msg\reg M)$
                \Else
                \State $\ctx\reg C\gets\Enc_k(\tau^M)$
                \EndIf
                \State $b'\gets\sys D^{\Enc_k(\cdot),\Dec_k(\cdot)}(\sigma^C)$
                \State\Return $\One\{b'=b\}$
            \end{algorithmic}
            \tcblower
            \begin{algorithmic}
                \State $b\uar\bin$
                \State $k\gets\Gen()$
                \State $\lCheat\gets0$
                \State $\msg\reg M\gets\sys D^{\Enc_k(\cdot),\Dec_k(\cdot)}$
                \State $\hat r\samp\Rnd$\cmt{Keep $\hat r$}
                \State $\hat M\tilde M\gets\epr$
                \State $\hat\ctx\reg{\hat C}\gets V_k(\hat\msg\reg{\hat M}\otimes\Pi_{k,\hat r}\reg T)V_k^\dag$\cmt{Ignore $\msg\reg M$}
                \State $b'\gets\sys D^{\Enc_k(\cdot),{\bf Dec}(\cdot)}(\hat\ctx\reg{\hat C})$
                \State\Return $\ite\lCheat1b$\cmt{Ignore $b'$}
                \vspace{1mm}\hrule\vspace{1mm}
                \Orac[Dec]{$\ctx\reg C$}
                \State $MT\gets V_k^\dag\ctx\reg CV_k$
                \If{$\{P_{\aux_k}\reg T,\One-P_{\aux_k}\reg T\}(\aux\reg T)\Out0$}
                \If{$\{\Pi_{k,\hat r},\One-\Pi_{k,\hat r}\}(\aux\reg T)\Out0$}
                \If{$\{\Pi_+,\One-\Pi_+\}(\varphi\reg{M\tilde M})\Out0$}
                \State $\lCheat\gets1$
                \EndIf
                \EndIf
                \Else
                \State\Return $\hat D_k(\rho\reg{MT})$\cmt{Invalid ciphertext}
                \EndIf
                \State\Return $\msg\reg M$
                \EndOrac
            \end{algorithmic}
        \end{algobox}
    \end{tcbraster}\caption{$\AGM\QCCA2$ games $\sys G^{\agm\qccat2}$ ({\bf left}) and $\sys G^{\agm\qccaf2}$ ({\bf right}).}
    \label{fig:qcca2}
\end{figure}

%%% Local Variables:
%%% TeX-master: "qccFull"
%%% End:

%%% Local Variables:
%%% TeX-master: "qccFull"
%%% End:
    \bibliography{refsFull}
    \else

    \appendix
    
    \ifsub
    \else
    
    \fi
    \bibliography{refsFull}
    \fi

\end{document}